\documentclass[reqno]{amsart}
\usepackage{fullpage,amsfonts,amssymb,amsthm,mathrsfs,graphicx,color,framed,esint}
\usepackage{cancel,bm}
\usepackage{tikz}
\usetikzlibrary{decorations.markings}
\usepackage{booktabs,longtable}

\usepackage[colorlinks=true, pdfstartview=FitV, linkcolor=blue, citecolor=blue, urlcolor=blue]{hyperref}
\definecolor{shadecolor}{rgb}{0.95, 0.95, 0.86}

\renewcommand{\d}{{\mathrm d}}

\newcommand{\im}{\mathrm{i}}
\newcommand{\e}{\mathrm{e}}

\def\tr{\mathop{\mathrm{tr}}\limits}

\numberwithin{equation}{section}

\newtheorem{theo}{Theorem}[section]

\newtheorem{lem}[theo]{Lemma}
\newtheorem{rem}[theo]{Remark}
\newtheorem{problem}[theo]{Riemann-Hilbert Problem}

\newtheorem{prop}[theo]{Proposition} 
\newtheorem{cor}[theo]{Corollary} 
\newtheorem{definition}[theo]{Definition}

\newtheorem{conv}[theo]{Convention}

\begin{document}

\title[Ginibre meets Zakharov-Shabat]{The largest real eigenvalue in the real Ginibre ensemble and its relation to the Zakharov-Shabat system}

\author{Jinho Baik}
\address{Department of Mathematics, University of Michigan, 2074 East Hall, 530 Church Street, Ann Arbor, MI 48109-1043, United States}
\email{baik@umich.edu}

\author{Thomas Bothner}
\address{Department of Mathematics, University of Michigan, 2074 East Hall, 530 Church Street, Ann Arbor, MI 48109-1043, United States}
\email{bothner@umich.edu}



\thanks{T.B. acknowledges support of the AMS and the Simons Foundation through a travel grant and is grateful to B. Rider, P. Deift,  and P. Miller for stimulating discussions about this project. J.B. is supported by the NSF grant DMS-1664692. Both authors initiated this work during the 2017 PCMI summer session on random matrices, funded in part by the NSF grant DMS-1441467. The authors would also like to thank P. Forrester for bringing \cite{F2} to their attention.}

\begin{abstract}
The real Ginibre ensemble consists of $n\times n$ real matrices ${\bf X}$ whose entries are i.i.d. standard normal random variables. In sharp contrast to the complex and quaternion Ginibre ensemble, real eigenvalues in the real Ginibre ensemble attain positive likelihood. In turn, the spectral radius $R_n=\max_{1\leq j\leq n}|z_j({\bf X})|$ of the eigenvalues $z_j({\bf X})\in\mathbb{C}$ of a real Ginibre matrix ${\bf X}$ follows a different limiting law (as $n\rightarrow\infty$) for $z_j({\bf X})\in\mathbb{R}$ than for $z_j({\bf X})\in\mathbb{C}\setminus\mathbb{R}$. Building on previous work by Rider, Sinclair \cite{RS} and Poplavskyi, Tribe, Zaboronski \cite{PTZ}, we show that the limiting distribution of $\max_{j:z_j\in\mathbb{R}}z_j({\bf X})$ admits a closed form expression in terms of a distinguished solution to an inverse scattering problem for the Zakharov-Shabat system.
As byproducts of our analysis we also obtain a new determinantal representation for the limiting distribution of $\max_{j:z_j\in\mathbb{R}}z_j({\bf X})$ and extend recent tail estimates in \cite{PTZ} via nonlinear steepest descent techniques.
\end{abstract}

\date{\today}
\maketitle
\section{Introduction and statement of results}\label{sec:11} This paper is foremost concerned with the derivation of an integrable system for the limiting distribution function
\begin{equation*}
	\lim_{n\rightarrow\infty}\mathbb{P}\left(\max_{j:z_j\in\mathbb{R}}z_j({\bf X})\leq\sqrt{n}+t\right),\ \ \ \ t\in\mathbb{R},
\end{equation*}
of the largest real eigenvalue of a random matrix ${\bf X}\in\mathbb{R}^{n\times n}$ chosen from the real Ginibre ensemble.
\begin{definition}[Ginibre \cite{G}, 1965] A random matrix ${\bf X}\in\mathbb{R}^{n\times n}$ is said to belong to the real Ginibre ensemble \textnormal{(GinOE)} if its entries are independently chosen with pdf's
\begin{equation*}
	\frac{1}{\sqrt{2\pi}}\e^{-\frac{1}{2}x_{jk}^2},\ \ \ \ 1\leq j,k\leq n.
\end{equation*}
Equivalently, the joint pdf of all the independent entries equals
\begin{equation*}
	f({\bf X})=\prod_{1\leq j,k\leq n}\frac{1}{\sqrt{2\pi}}\e^{-\frac{1}{2}x_{jk}^2}=(2\pi)^{-\frac{1}{2}n^2}\e^{-\frac{1}{2}\sum_{j,k=1}^nx_{jk}^2}=(2\pi)^{-\frac{1}{2}n^2}\e^{-\frac{1}{2}\textnormal{tr}({\bf X}{\bf X}^{\intercal})}.
\end{equation*}
\end{definition}
\noindent The GinOE displays certain similarities to the classical Gaussian Orthogonal Ensemble (GOE) but the presence of, both, real and complex eigenvalues introduces also new phenomena. For instance, on a global scale, Wigner's semicircle law in the GOE is replaced by the following circular law \cite{E0}: let
\begin{equation*}
	\mu_{\bf X}(s,t)=\frac{1}{n}\#\big\{1\leq j\leq n:\ \Re z_j({\bf X})\leq s,\ \Im z_j({\bf X})\leq t\big\},\ \ \ s,t\in\mathbb{R}
\end{equation*}
denote the empirical spectral distribution of the eigenvalues $\{z_j({\bf X})\}_{j=1}^n$ of a matrix ${\bf X}\in\textnormal{GinOE}$, then the random measure $\mu_{\bf X}/\sqrt{n}$  converges almost surely (as $n\rightarrow\infty$) to the uniform distribution on the unit disk, see Figure \ref{fig1} below. 
\begin{rem} The circular law is a universal limiting law: it holds true for any $n\times n$ random matrix ${\bf X}$ whose entries are i.i.d. complex random variables with mean zero and variance one, see \cite{TV} and references therein to the long and rich history of the circular law.
\end{rem}
\begin{figure}[tbh]
\includegraphics[width=0.24\textwidth]{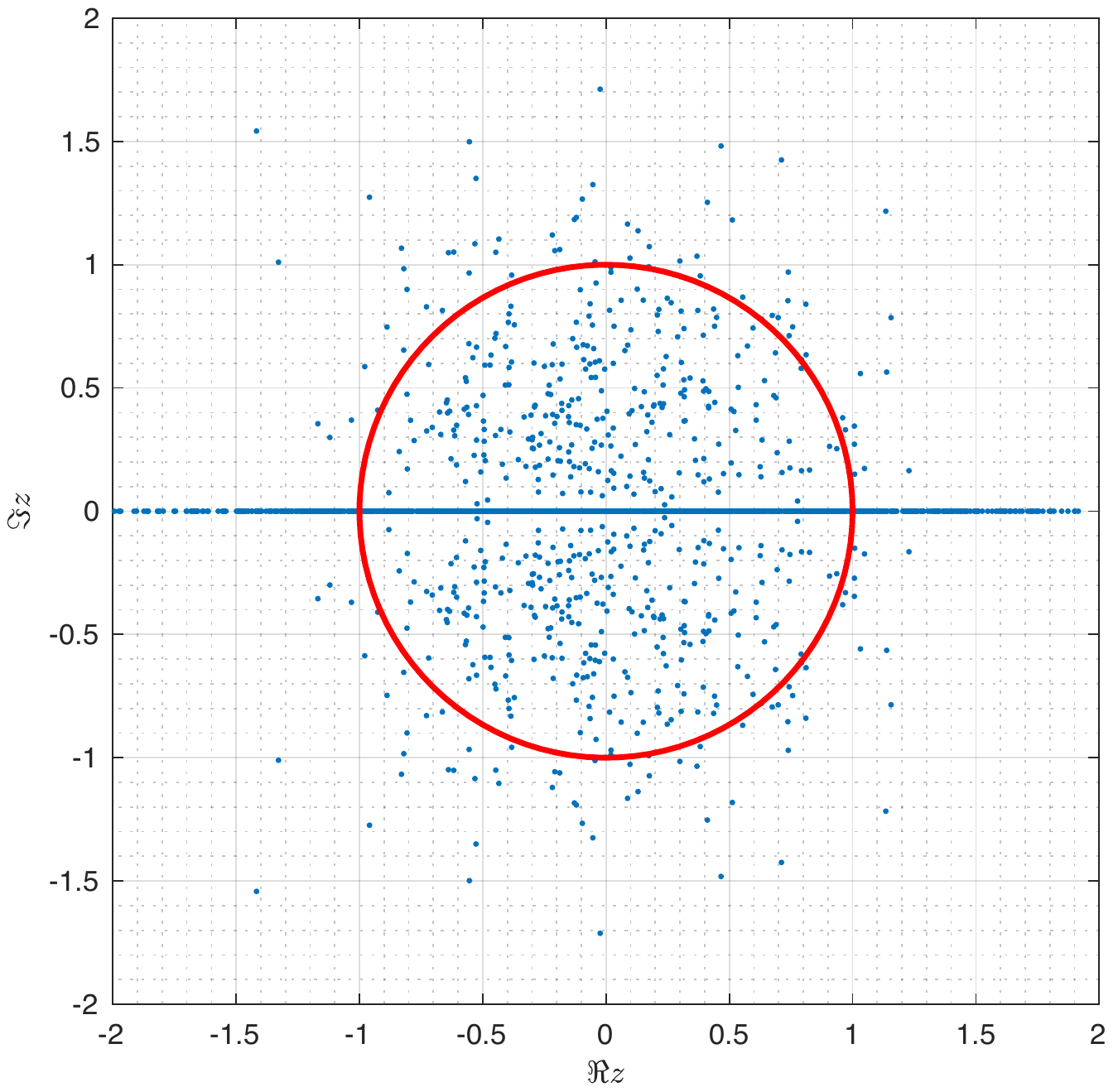}\hspace{0.18cm}
\includegraphics[width=0.24\textwidth]{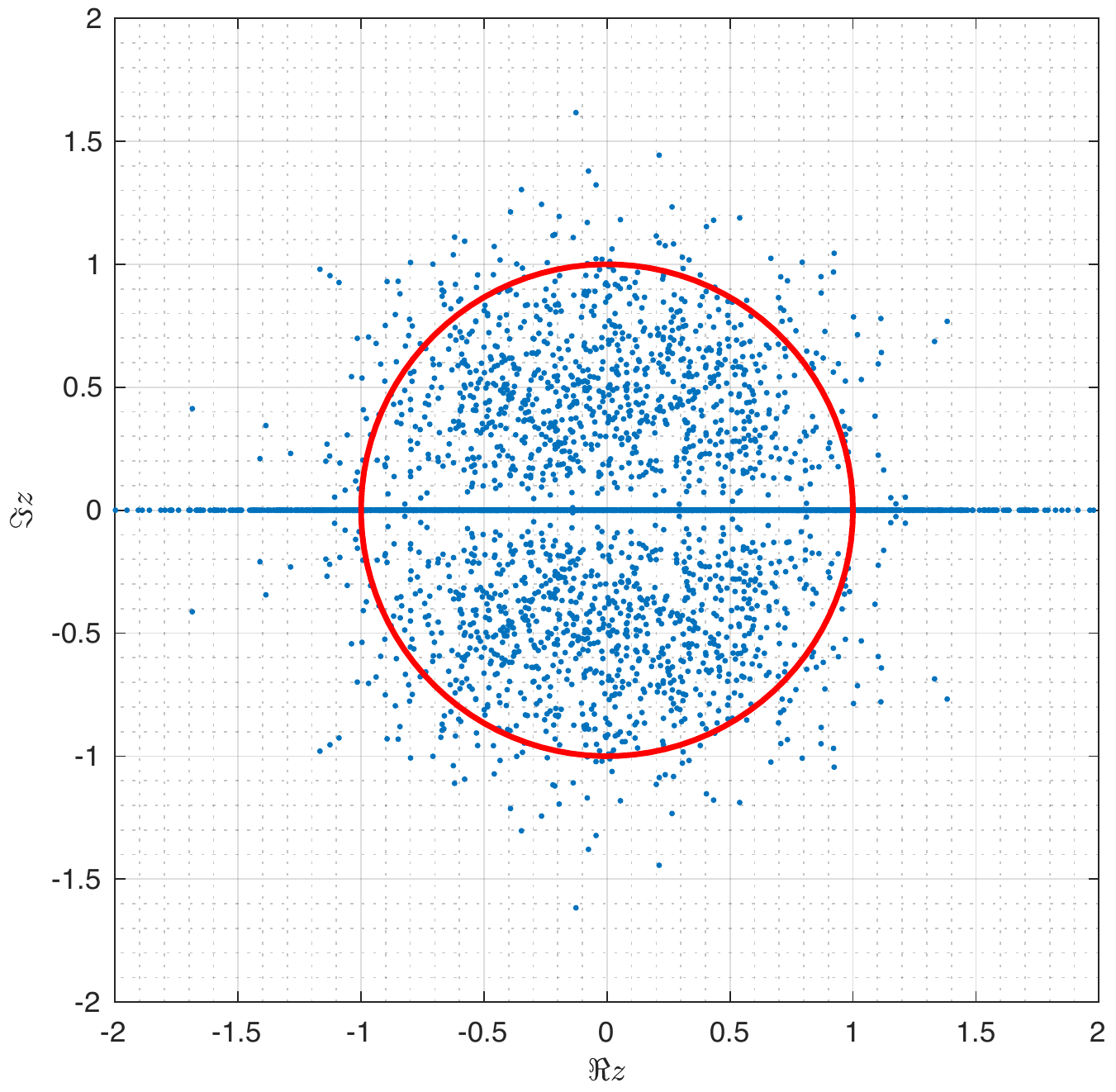}\hspace{0.18cm}\includegraphics[width=0.24\textwidth]{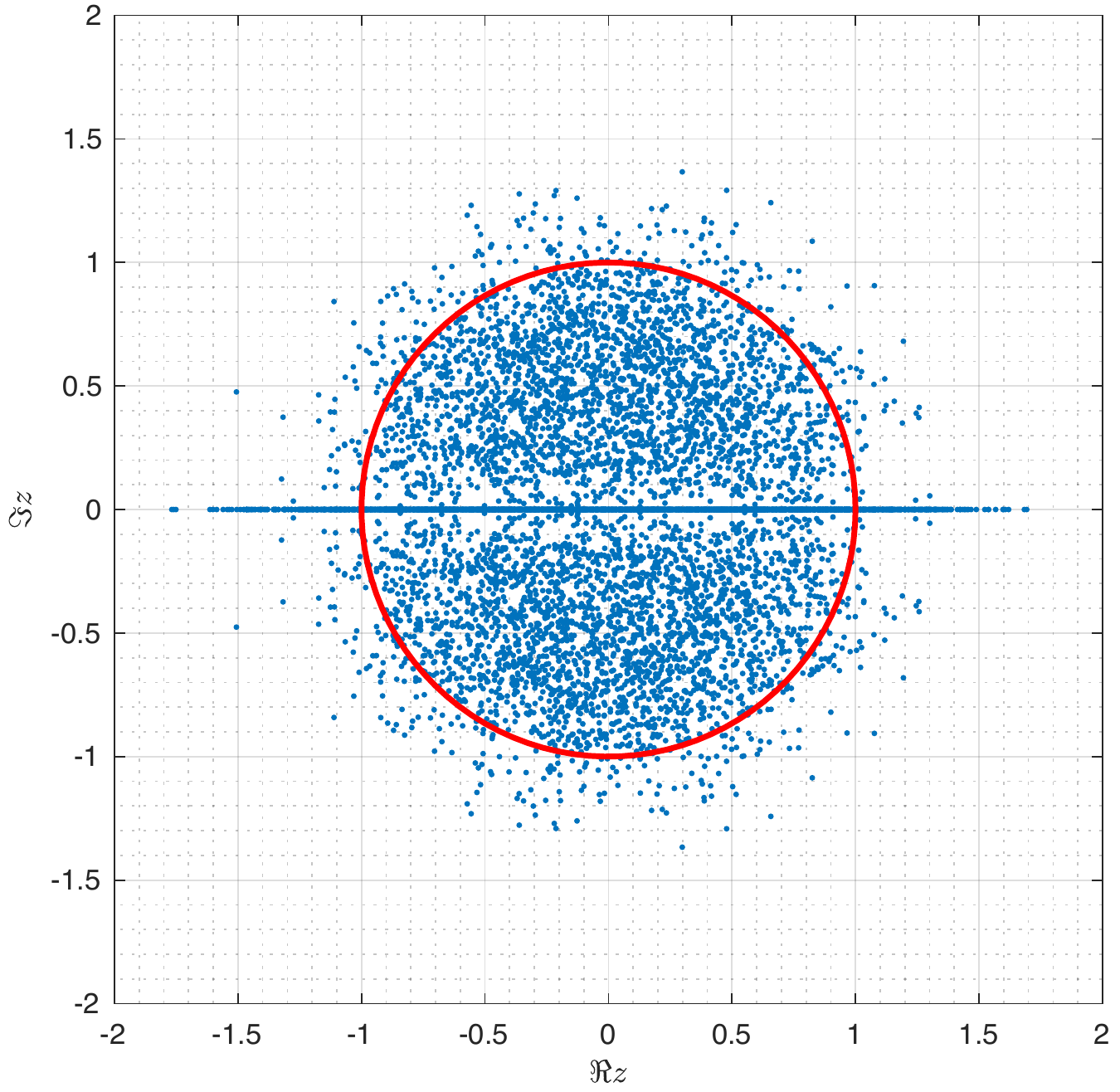}\hspace{0.18cm}\includegraphics[width=0.24\textwidth]{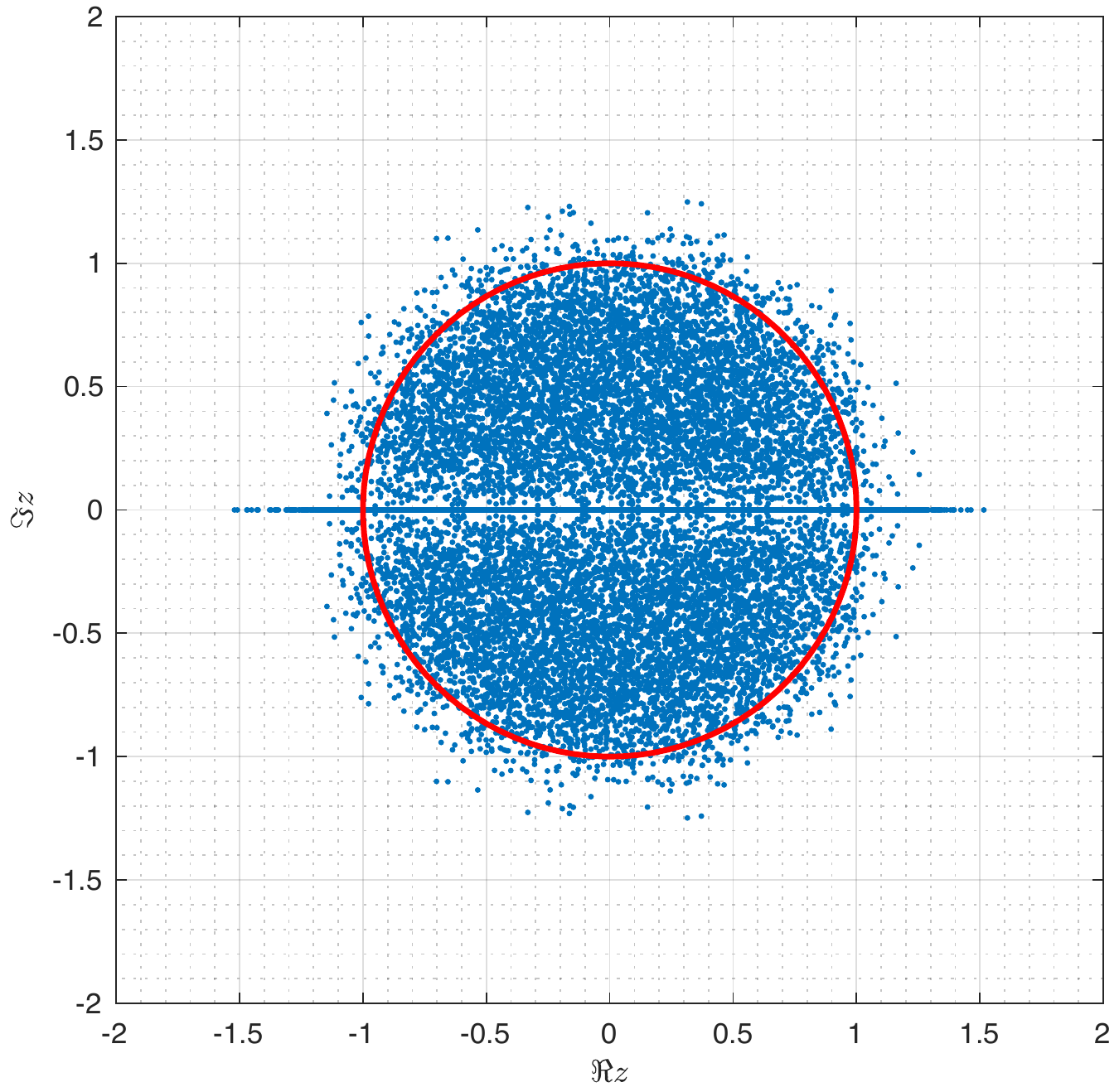}
\caption{The circular law for $1000$ real (rescaled) Ginibre matrices of varying dimensions $n\times n$ in comparison with the unit circle boundary. We plot $n=2,4,8,16$ from left to right. A saturn effect is clearly visible on the real line.}
\label{fig1}
\end{figure}
On a local scale, Figure \ref{fig1} indicates that fluctuations of the spectral radius $R_n=\max_{1\leq j\leq n}|z_j({\bf X})|$ around $\sqrt{n}$ behave differently depending on whether $z_j\in\mathbb{R}$ or $z_j\in\mathbb{C}\setminus\mathbb{R}$. And indeed, the above-mentioned saturn effect was quantified recently and the following central limit theorem derived.
\begin{theo}[Rider, Sinclair \cite{RS}, 2014; Poplavskyi, Tribe, Zaboronski \cite{PTZ}, 2017] Let $\{z_j({\bf X})\}_{j=1}^n$ denote the eigenvalues of a $n\times n$ random matrix ${\bf X}\in\textnormal{GinOE}$. Then,
\begin{equation*}
	\lim_{n\rightarrow\infty}\mathbb{P}\left(\max_{j:z_j\in\mathbb{C}\setminus\mathbb{R}}|z_j({\bf X})|\leq\sqrt{n}+\sqrt{\frac{\gamma_n}{4}}+\frac{t}{\sqrt{4\gamma_n}}\right)=\e^{-\frac{1}{2}\e^{-t}},\ \ \ \ t\in\mathbb{R}
\end{equation*}
with $\gamma_n=\ln(n/(2\pi(\ln n)^2))$. In addition,
\begin{equation}\label{e:1}
	\lim_{n\rightarrow\infty}\mathbb{P}\left(\max_{j:z_j\in\mathbb{R}}z_j({\bf X})\leq\sqrt{n}+t\right)=\sqrt{\det(1-T\chi_t\upharpoonright_{L^2(\mathbb{R})})\Gamma_t},\ \ \ \ \ t\in\mathbb{R},
\end{equation}
where $\chi_t$ is the operator of multiplication by $\chi_{(t,+\infty)}(x)$, the characteristic function of $(t,+\infty)\subset\mathbb{R}$, and $T:L^2(\mathbb{R})\rightarrow L^2(\mathbb{R})$ the trace-class integral operator with kernel
\begin{equation}\label{e:2}
	T(x,y)=\frac{1}{\pi}\int_0^{\infty}\e^{-(x+u)^2}\e^{-(y+u)^2}\,\d u.
\end{equation}
Moreover,
\begin{equation}\label{e:3}
	\Gamma_t=1-\int_t^{\infty}G(x)\big((1-T\chi_t\upharpoonright_{L^2(\mathbb{R})})^{-1}g\big)(x)\,\d x
\end{equation}
with $g(x)=\frac{1}{\sqrt{\pi}}\e^{-x^2}$ and $G(x)=\int_{-\infty}^xg(y)\,\d y$.
\end{theo} 
When compared to the GOE, \eqref{e:1} plays the analogue of the celebrated Tracy-Widom edge law for the largest eigenvalue $\lambda_{\max}$, cf. \cite{TW2}. Indeed we recall that in the GOE, as $n\rightarrow\infty$,
\begin{equation*}
	\lambda_{\max}\Rightarrow \sqrt{2n}+\frac{1}{\sqrt{2}\,n^{\frac{1}{6}}}F_1,
\end{equation*}
with the cdf
\begin{equation}\label{e:4}
	F_1(t)=\sqrt{\det(1-K\chi_t\upharpoonright_{L^2(\mathbb{R})})\widehat{\Gamma}_t},\ \ \ \ \ t\in\mathbb{R},
\end{equation}
where $K:L^2(\mathbb{R})\rightarrow L^2(\mathbb{R})$ is the integral operator with Airy kernel $K(x,y)=\int_0^{\infty}\textnormal{Ai}(x+u)\textnormal{Ai}(y+u)\,\d u$ in terms of the Airy function $\textnormal{Ai}(z)$, \cite{NIST}. Moreover, see e.g. \cite[Section $9.7$]{F},
\begin{equation}\label{e:5}
	\widehat{\Gamma}_t=1-\int_t^{\infty}A(x)\big((1-K\chi_t\upharpoonright_{L^2(\mathbb{R})})^{-1}\textnormal{Ai}\big)(x)\,\d x;\ \ \ \ \ \ \ A(x)=\int_{-\infty}^x\textnormal{Ai}(y)\,\d y.
\end{equation}
\subsection{An integrable system for \eqref{e:1}}
The formal similarities between \eqref{e:4}, \eqref{e:5} and \eqref{e:1}, \eqref{e:3} are quite obvious, still while the operator $K$ is of integrable type in the sense of \cite{IIKS}, i.e. has a kernel of the form
\begin{equation}\label{TWkernel}
	K(x,y)=\frac{{\bf f}^{\intercal}(x){\bf g}(y)}{x-y},\ \ \ \ {\bf f}(z)=\begin{bmatrix}\textnormal{Ai}(z)\\-\textnormal{Ai}'(z)\end{bmatrix},\ \ \ {\bf g}(z)=\begin{bmatrix}\textnormal{Ai}'(z)\\ \textnormal{Ai}(z)\end{bmatrix},\end{equation}
this is not true for $T$ with kernel \eqref{e:2}, see explicitly \cite[Section 4]{RS}. For this reason neither the standard Tracy-Widom method \cite{TW3} used in the derivation of an integrable system (a.k.a. a closed form expression) for the limiting distribution function \eqref{e:1} nor the Riemann-Hilbert problem based techniques of Borodin and Deift \cite{BD} are directly applicable. However, as we will show below, the situation with \eqref{e:2} is not too bad, since the operator $T\chi_t$ is of integrable type up to Fourier conjugation, see Proposition \ref{cool} below. This observation combined with certain additional manipulations for the Fredholm determinant and the factor $\Gamma_t$ in \eqref{e:1}, see Sections \ref{prelim}, \ref{contin} and \ref{contin2} below, yields an explicit integrable system for $F(t)$ and a subsequent closed form, Tracy-Widom like, formula. In fact we shall state the sought after closed form expression for the following generalization of \eqref{e:1} that contains a generating function parameter $\gamma\in[0,1]$,
\begin{equation}\label{gdef}
	F(t;\gamma):=\sqrt{\det(1-\gamma T\chi_t\upharpoonright_{L^2(\mathbb{R})})\Gamma_{t\gamma}},\ \  \textnormal{with}\ \ 
	\Gamma_{t\gamma}:=1-\gamma\int_t^{\infty}G(x)\big((1-\gamma T\chi_t\upharpoonright_{L^2(\mathbb{R})})^{-1}g\big)(x)\,\d x.
\end{equation}
The main result of our paper, see Theorem \ref{main1} below, is a closed form expression for $F(t;\gamma)$ in terms of a distinguished solution to an inverse scattering problem for the Zakharov-Shabat (ZS) system \cite{ZS,AC}. As it is standard in scattering theory, we shall formulate this inverse problem as a Riemann-Hilbert problem (RHP):
\begin{problem}\label{master0} For any $(x,\gamma)\in\mathbb{R}\times[0,1]$, determine ${\bf X}(z)={\bf X}(z;x,\gamma)\in\mathbb{C}^{2\times 2}$ such that
\begin{enumerate}
	\item[(1)] ${\bf X}(z)$ is analytic for $z\in\mathbb{C}\setminus\mathbb{R}$ and has a continuous extension on the closed upper and lower half-planes.
	\item[(2)] The limiting values ${\bf X}_{\pm}(z)=\lim_{\epsilon\downarrow 0}{\bf X}(z\pm\im\epsilon),z\in\mathbb{R}$ satisfy the jump condition
	\begin{equation}\label{djump}
		{\bf X}_+(z)={\bf X}_-(z)\begin{bmatrix}1-|r(z)|^2 & -\bar{r}(z)\e^{-2\im xz}\smallskip\\ r(z)\e^{2\im xz} & 1\end{bmatrix},\ \ z\in\mathbb{R}\ \ \ \ \ \textnormal{with}\ \ \ r(z)=r(z;\gamma)=-\im\sqrt{\gamma}\,\e^{-\frac{1}{4}z^2}.
	\end{equation}
	\item[(3)] As $z\rightarrow\infty$, we require the normalization
	\begin{equation*}
		{\bf X}(z)=\mathbb{I}+{\bf X}_1z^{-1}+{\bf X}_2z^{-2}+\mathcal{O}\big(z^{-3}\big);\ \ \ \ \ {\bf X}_i={\bf X}_i(x,\gamma)=\big[X_i^{jk}(x,\gamma)\big]_{j,k=1}^2.
	\end{equation*}
\end{enumerate}
\end{problem}
Note that $r(z;\gamma)\in\mathcal{S}(\mathbb{R})$, the Schwartz space on the line, but $\|r\|_{\infty}=\sup_{z\in\mathbb{R}}|r(z)|<1$ only for $\gamma\in[0,1)$. 
Hence $r(z;\gamma)$, the so-called reflection coefficient, does not belong to the standard Beals-Coifman class of reflection coefficients, cf. \cite{BC0,BDT}, in the case \eqref{e:1} most relevant to the GinOE. For this reason we will prove unique solvability of RHP \ref{master0} for all $(x,\gamma)\in\mathbb{R}\times[0,1]$ and thus also existence of the coefficients ${\bf X}_i(x,\gamma)$ directly in the sections below. We now present our main result.
\begin{theo}\label{main1}
For any $(t,\gamma)\in\mathbb{R}\times[0,1]$,
\begin{equation}\label{e:6}
	\big(F(t;\gamma)\big)^2=\exp\bigg[-\frac{1}{4}\int_t^{\infty}(x-t)\left|y\left(\frac{x}{2};\gamma\right)\right|^2\d x\bigg]\bigg\{\cosh\mu(t;\gamma)-\sqrt{\gamma}\sinh\mu(t;\gamma)\bigg\},
\end{equation}
using the abbreviation
\begin{equation*}
	\mu(t;\gamma):=-\frac{\im}{2}\int_t^{\infty}y\left(\frac{x}{2};\gamma\right)\,\d x,
\end{equation*}
and where $y=y(x;\gamma):\mathbb{R}\times[0,1]\rightarrow\im\mathbb{R}$ equals $y(x;\gamma):=2\im X_1^{12}(x,\gamma)$ in terms of the matrix coefficient ${\bf X}_1(x,\gamma)$ in condition (3) of RHP \ref{master0} above.
\end{theo}
%
%
%
%
%
%
%
Identity \eqref{e:6} is the analogue of the Tracy-Widom Painlev\'e-II formula for $F_1(t)$ in case $\gamma=1$, see \cite{TW2}. For $\gamma\in[0,1)$, our definition \eqref{gdef} is motivated by the generating function of the soft-edge scaled $(t,+\infty)$ gap probabilities for the superimposed, cf. \cite[Section $6.6$]{F}, orthogonal ensemble 
\begin{equation*}\label{SGOE}
	\textnormal{odd}\big(\textnormal{OE}_n(\e^{-x^2})\cup\textnormal{OE}_n(\e^{-x^2})\big).
\end{equation*} 
In this context, \eqref{e:6} is the direct analogue of \cite[$(9.150)$]{F}, modulo the replacement of the Painlev\'e-II transcendent with the above solution entry $y(x;\gamma)$ of RHP \ref{master0}.
\begin{rem} \label{rem:thin} Another possible motivation for the introduction of $\gamma$ in \eqref{gdef} could arise from studying a thinned version of the real eigenvalues in the GinOE, i.e. from analyzing the point process 
\begin{equation*}
	\mathfrak{X}^{\gamma}=\{z_j^{\gamma}({\bf X}):z_j^{\gamma}({\bf X})\in\mathbb{R}\}_{j=1}^{n(\gamma)},\ \ \ 1\leq n(\gamma)\leq n,
\end{equation*}
obtained from $\mathfrak{X}=\{z_j({\bf X})\in\mathbb{R}\}_{j=1}^n$ with ${\bf X}\in\textnormal{GinOE}$ by independently removing each real eigenvalue with likelihood $\gamma\in[0,1]$. For GOE, the largest eigenvalue distribution function after thinning admits a Painlev\'e closed form expression, see \cite[$(1.6)$]{BB}, in case of GinOE the corresponding result is unknown. It is also not immediately clear whether \eqref{gdef} has a probabilistic interpretation in the thinned superimposed orthogonal ensemble. We plan to address this question in a future publication.
\end{rem}
\begin{rem} We prove existence of $y(x;\gamma)$ for $(x,\gamma)\in\mathbb{R}\times[0,1]$ in Theorem \ref{solv:2} and continuity of $y(x;\gamma),x\in\mathbb{R}$ for any fixed $\gamma\in[0,1]$ in Lemma \ref{reg} and Corollary \ref{solv:1}. Moreover we show that $y(x;\gamma)$ for $(x,\gamma)\in\mathbb{R}\times[0,1]$ is purely imaginary and
\begin{equation}\label{easyesti}
	y(x;\gamma)=2\im\sqrt{\frac{\gamma}{\pi}}\,\e^{-4x^2}\left(1+\mathcal{O}\left(\e^{-4x^2}\right)\right),\ \ \ x\rightarrow+\infty,
\end{equation}
i.e. the right-hand side in \eqref{e:6} is well-defined.
\end{rem}
Returning to the afore-mentioned comparison between GinOE and GOE we see from \eqref{e:6} and \cite[$(53)$]{TW2} that, overall, the main difference in GinOE arises from the presence of the inverse scattering type RHP \ref{master0} and its solution entry $y(x;\gamma)$ instead of the more common Painlev\'e transcendents in the Gaussian invariant ensembles. For this reason we shall briefly review a few selected aspects of the integrability theory of RHP \ref{master0}, see \cite{AC,BC0,BDT,ZS} for more details.
\subsection{The Zakharov-Shabat system in a nutshell}\label{IST} Note that
\begin{equation*}
	{\bf\Psi}(z):={\bf X}(z)\e^{-\im xz\sigma_3},\ \ \ z\in\mathbb{C}\setminus\mathbb{R}
\end{equation*}
solves a RHP with an $x$-independent jump on $\mathbb{R}$, thus $\frac{\partial{\bf\Psi}}{\partial x}{\bf \Psi}^{-1}$ is an entire function. In fact, using condition (3) in RHP \ref{master0} and Liouville's theorem, we find
\begin{equation}\label{ZS:0}
	\frac{\partial{\bf\Psi}}{\partial x}=\left\{-\im z\sigma_3+2\im\begin{bmatrix}0 & X_1^{12}\smallskip\\ -X_1^{21} & 0\end{bmatrix}\right\}{\bf\Psi}.
\end{equation}
But since RHP \ref{master0} enjoys the symmetry
\begin{equation}\label{dNLSsymm}
	{\bf X}(z;x,\gamma)=\sigma_1\overline{{\bf X}(\bar{z};x,\gamma)}\sigma_1,\ \ \ \ z\in\mathbb{C}\setminus\mathbb{R};\ \ \ \ \sigma_1=\begin{bmatrix}0 & 1\\ 1 & 0\end{bmatrix},
\end{equation}
we learn that $X_i^{11}(x,\gamma)=\overline{X_i^{22}(x,\gamma)}$ as well as $X_i^{21}(x,\gamma)=\overline{X_i^{12}(x,\gamma)}$ and thus with $y=2\im X_1^{12}$ from \eqref{ZS:0},
\begin{equation}\label{ZS:1}
	\frac{\partial{\bf \Psi}}{\partial x}=\left\{-\im z\sigma_3+\begin{bmatrix}0 & y\\ \bar{y} & 0\end{bmatrix}\right\}{\bf\Psi}\equiv{\bf U}(z;x,\gamma){\bf \Psi}.
\end{equation}
This celebrated first order system, known as ZS-system, is directly related to several of the most interesting nonlinear evolution equations in $1+1$ dimensions which are solvable by the inverse scattering method. For instance, in order to solve the Cauchy problem for the defocusing nonlinear Schr\"odinger equation,
\begin{equation}\label{dNLS}
	\im y_t+y_{xx}-2|y|^2y=0,\ \ \ \ y(x,0)=y_0(x)\in\mathcal{S}(\mathbb{R});\ \ \ \ \ y=y(x,t):\mathbb{R}^2\rightarrow\mathbb{C},
\end{equation}
one first computes the reflection coefficient $r(z)\in\mathcal{S}(\mathbb{R})$ associated to the initial data $y_0$ through the direct scattering transform. A basic fact of the scattering theory for the Zhakarov-Shabat system \eqref{ZS:1} states that this transform, i.e. the map $y_0\rightarrow r$, is a bijection from $\mathcal{S}(\mathbb{R})$ onto $\mathcal{S}(\mathbb{R})\cap\{r:\,\|r\|_{\infty}=\sup_{z\in\mathbb{R}}|r(z)|<1\}$, cf. \cite{BC0}. Second, one considers RHP \ref{master0} above subject to the replacement
\begin{equation*}
	\e^{2\im xz}\rightarrow\e^{2\im(2tz^2+xz)},\ \ \ t\in\mathbb{R},
\end{equation*}
and provided this problem is solvable, its (unique) solution in turn leads to a solution of \eqref{dNLS} with $y(x,0)=y_0(x)$ via the formula $y(x,t)=2\im X_1^{12}(x,t)$. Thus, in order to solve \eqref{dNLS}, one must solve the $t$-modified RHP \ref{master0} (a.k.a. the inverse scattering transform) for the given reflection coefficient $r(z)$, determined under the aforementioned bijection $y_0\rightarrow r$. Returning now to our context, we see that \eqref{e:6} therefore depends on a distinguished solution $y(x;\gamma)$ of the inverse scattering transform for the Zakharov-Shabat system \eqref{ZS:1} subject to the reflection coefficient $r(z;\gamma)=-\im\sqrt{\gamma}\,\e^{-\frac{1}{4}z^2}$.
\begin{rem} As outlined above, the operator $T\chi_t$ is of integrable type once viewed in Fourier space. This idea was first used in the analysis of single- and multi-time processes in \cite{BCe,BC}. Specifically, loc. cit. showed that certain matrix Fredholm determinants are expressible as determinants of integrable matrix kernels and thus connected to RHPs. As a direct application of this technique, Bertola and Cafasso re-derived  for instance the Adler-van Moerbeke PDE for the joint distributions of the Airy-$2$ process by Riemann-Hilbert techniques. Our approach to the GinOE and \eqref{gdef} is clearly inspired by these works.
\end{rem}
\subsection{Tail asymptotics} The advantage of the exact formula \eqref{e:6} lies in the fact that $y=y(x;\gamma)$ admits a Riemann-Hilbert formulation as outlined in RHP \ref{master0}. Thus its large space/long time behavior can be systematically computed via nonlinear steepest descent techniques \cite{DZ} and this paths the way to large tail estimates for \eqref{gdef}. We summarize our second result.  
\begin{cor}\label{main2} Let $\gamma\in[0,1]$ and $F(t;\gamma)$ be defined as in \eqref{gdef}. Then, as $t\rightarrow+\infty$,
\begin{equation}\label{e:9}
	F(t;\gamma)=1-\frac{\gamma}{4}\textnormal{erfc}(t)+\mathcal{O}\left(\gamma^{\frac{3}{2}}t^{-1}\e^{-2t^2}\right),
\end{equation}
in terms of the complementary error function $\textnormal{erfc}(z)$, cf. \cite[$7.2.2$]{NIST}. On the other hand, as $t\rightarrow-\infty$,
\begin{equation}\label{e:10}
	F(t;\gamma)=\e^{\eta_1(\gamma)t}\eta_0(\gamma)\big(1+o(1)\big),\ \ \ \eta_1(\gamma)=\frac{1}{2\sqrt{2\pi}}\,\textnormal{Li}_{\frac{3}{2}}(\gamma),
\end{equation}
in terms of the polylogarithm $\textnormal{Li}_s(z)$, cf. \cite[$25.12.10$]{NIST}, and with a $t$-independent positive factor $\eta_0(\gamma)$.
\end{cor}
Estimate \eqref{e:9} for $\gamma=1$ is standard and known from \cite{FN}, say. The leading order in (1.14) for $\gamma = 1$ was obtained by Forrester \cite{F2} and also by Poplavskyi, Tribe and Zaboronski \cite{PTZ} using an interesting connection to coalescence processes. Here we employ nonlinear steepest descent techniques to confirm \eqref{e:9} for all $\gamma\in[0,1]$ and derive \eqref{e:10} for $\gamma\in[0,1)$.\smallskip

The paper \cite{F2} also discusses the thinned version of the real eigenvalues, or equivalently, the generating function for the probability of the number of eigenvalues. Namely, the discussions in the third paragraph of \cite[Section $4$]{F2} claim a leading order term which, after evaluating the integral \cite[$(4.1)$]{F2} explicitly, is the same as our $\eta_1(\gamma)$ with $\gamma=2\xi-\xi^2$ where $1-\xi$ denotes the removal probability of an eigenvalue. As mentioned above in Remark \ref{rem:thin} it is interesting to consider the relationship between our $F(t;\gamma)$ and the thinned version.

\begin{rem}\label{ASegur} As can be seen from \eqref{e:10}, the left tails of $F(t;\gamma)$ in the leading order decay to zero exponentially fast for all $\gamma\in(0,1]$. This is in sharp contrast to the GOE where $\ln F_1(t)\sim\frac{1}{24}t^3$ as $t\rightarrow-\infty$ and $\ln F_1(t;\gamma)\sim-\frac{2v}{3\pi}(-t)^{3/2}$ for $\gamma\in(0,1)$ fixed with $v=-\ln(1-\gamma)<+\infty$, cf. \cite{BB}. This means that the large negative $x$ behavior of $\Im(y(\frac{x}{2};\gamma))$ cannot be as sensitive to a small change in $\gamma\in(0,1]$ near $\gamma=1$ as the corresponding behavior for the Painlev\'e-II transcendent $u(x;\gamma)$, see Figure \ref{figure111} below for a visualization.
\end{rem}
\subsection{Numerical comparison} The closed form expression \eqref{e:6} is not as optimal for numerical purposes as a single Fredholm determinant formula, compare the discussions in \cite{B}. For this reason we derive a new determinantal formula for the limiting distribution of $\max_{j:z_j\in\mathbb{R}}z_j({\bf X})$ in our third result below. This identity is completely analogous to the Ferrari-Spohn formula \cite{FS} in the GOE.
\begin{theo}\label{main3}
Let $F(t)$ be defined as in \eqref{e:1}. Then
\begin{equation}\label{e:11}
	F(t)=\det(1-S\chi_t\upharpoonright_{L^2(\mathbb{R})}),
\end{equation}
where $S:L^2(\mathbb{R})\rightarrow L^2(\mathbb{R})$ is the integral operator with kernel
\begin{equation}\label{e:12}
	S(x,y)=\frac{1}{2\sqrt{\pi}}\e^{-\frac{1}{4}(x+y)^2}.
\end{equation}
\end{theo}
Identity \eqref{e:11} allows us to numerically simulate several statistical quantities of $\max_{j:z_j\in\mathbb{R}}z_j({\bf X})$ by implementing Bornemann's algorithm \cite{B} in MATLAB. In more detail we discretize the determinant \eqref{e:11} by the Nystr\"om method using a Gauss-Legendre quadrature rule with $m=50$ quadrature points. 
Once the values of the cdf are then numerically accessible, computing associated quantities such as moments is straightforward. We summarize a few values in Table \ref{tab:1} below.\smallskip
\begin{table}[h]
\caption{Some moments of the GinOE in comparison to GOE moments}\label{tab:1}
\begin{center}
\begin{tabular}{ l  c  c c c}
\toprule
ensemble & mean & variance& skewness& kurtosis\\[2pt] \midrule
GinOE \eqref{e:1}& -1.30319& 3.97536& -1.76969&5.14560\\[2pt] \midrule
GOE \eqref{e:4}& -1.20653& 1.60778& 0.29346&0.16524\\[2pt] \bottomrule
\end{tabular}
\end{center}
\end{table} 

In the upcoming figures we first plot the distribution function $F(t)$ of the largest real eigenvalue in the GinOE in comparison to $F_1(t)$, the cdf of the largest eigenvalue in the GOE, compare Figure \ref{figure1}. After that we compare the asymptotic expansions \eqref{e:9} and \eqref{e:10} to our numerical results in Figure \ref{figure11}.
\begin{center}
\begin{figure}[tbh]
\resizebox{0.448\textwidth}{!}{\includegraphics{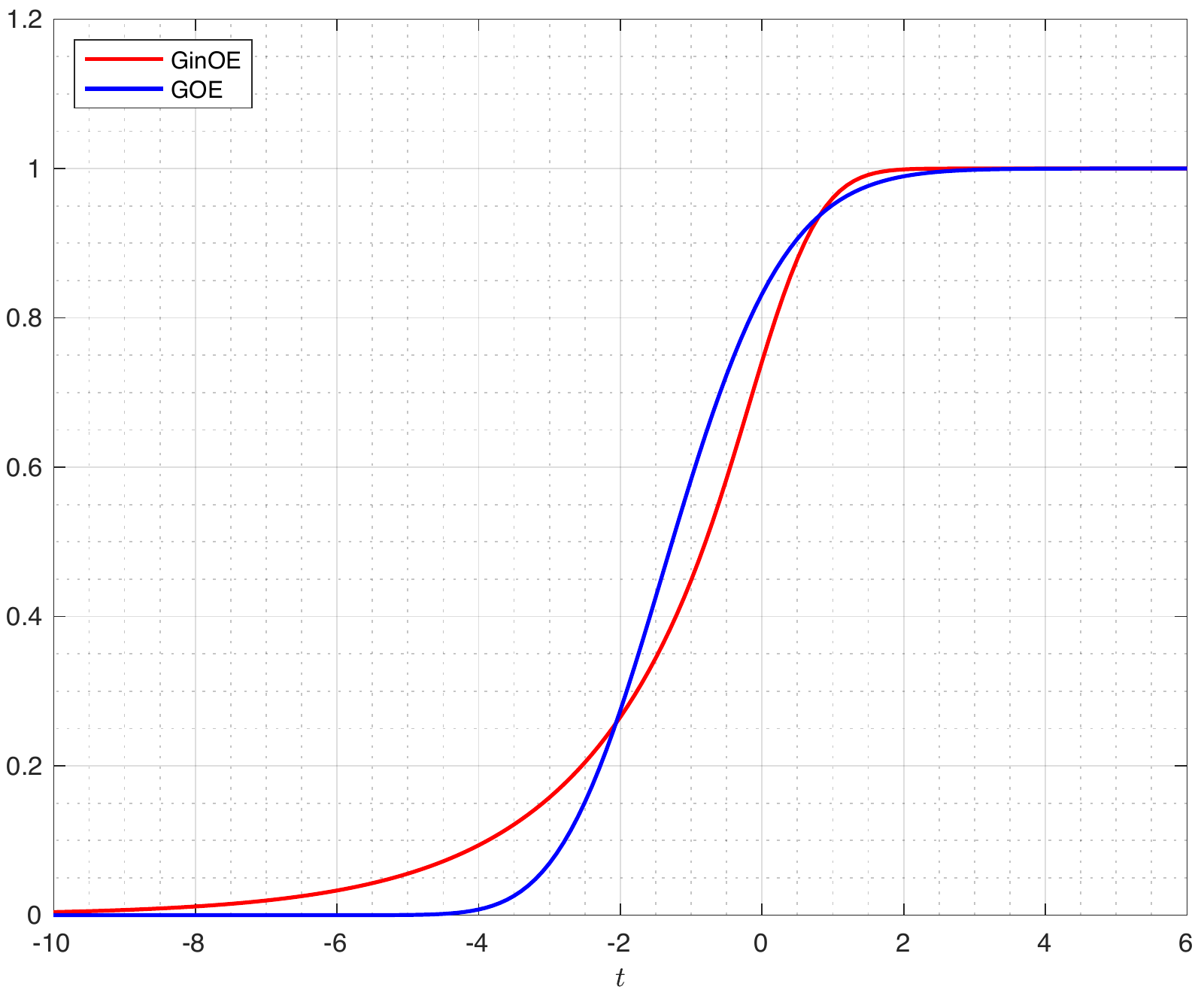}}\ \ \ \ \ \resizebox{0.455\textwidth}{!}{\includegraphics{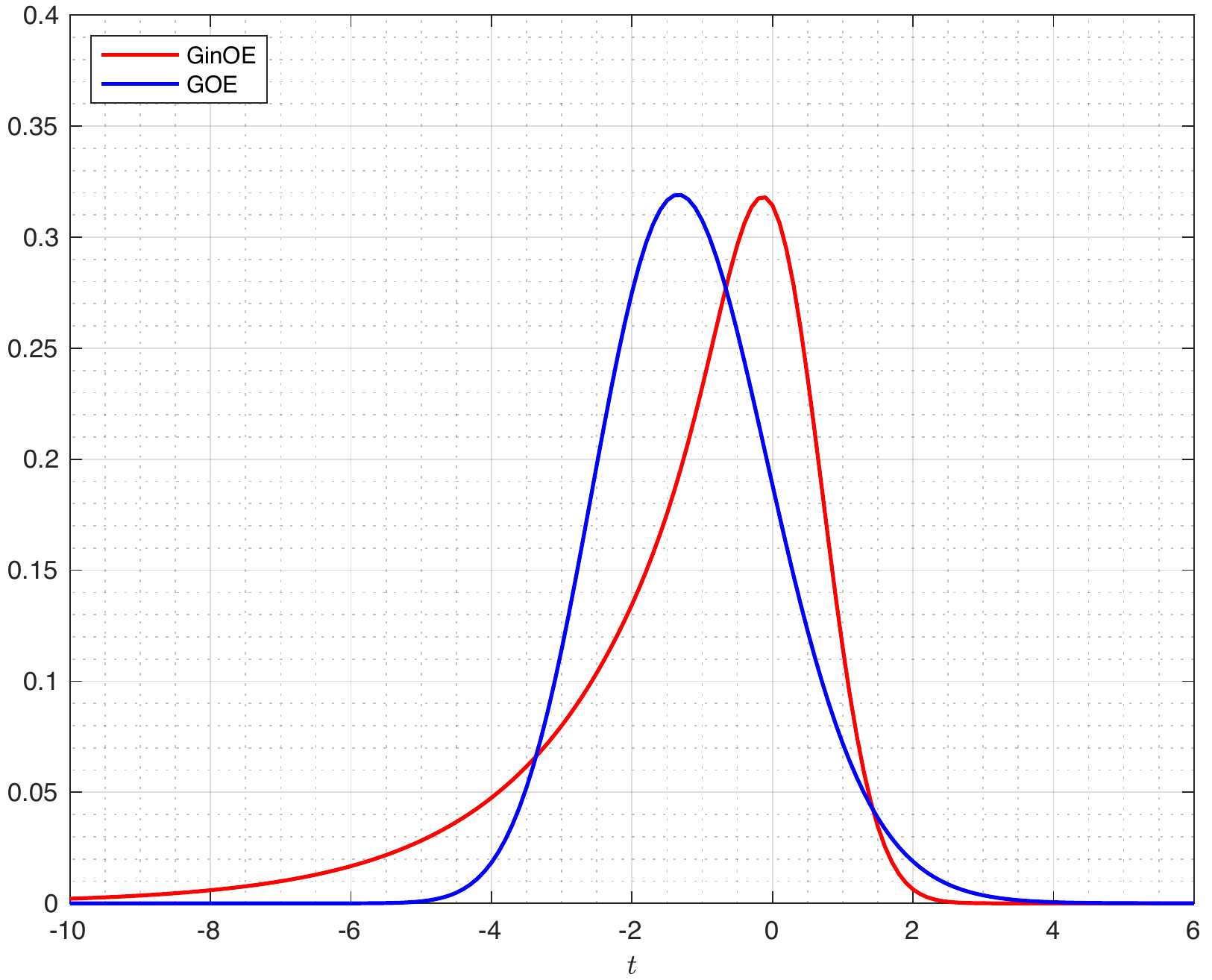}}
\caption{The distribution function $F(t)$ of the largest real GinOE eigenvalue in red versus GOE Tracy-Widom $F_1(t)$ in blue. The plots are generated in MATLAB with $m=50$ quadrature points using the Nystr\"om method with Gauss-Legendre quadrature. On the left cdfs, on the right pdfs.}
\label{figure1}
\end{figure}
\end{center} 
\begin{center}
\begin{figure}
\resizebox{0.451\textwidth}{!}{\includegraphics{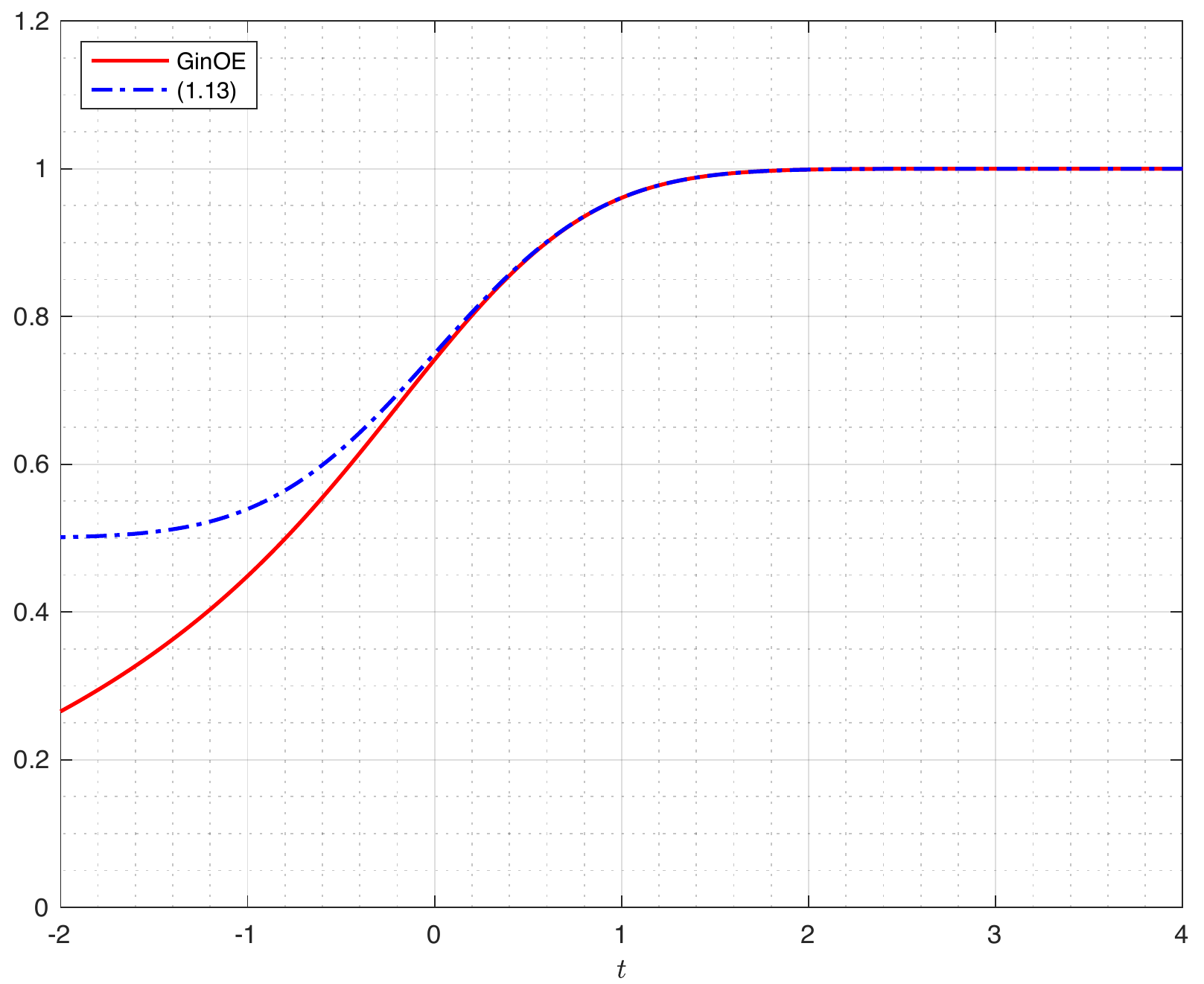}}\ \ \ \ \ \resizebox{0.455\textwidth}{!}{\includegraphics{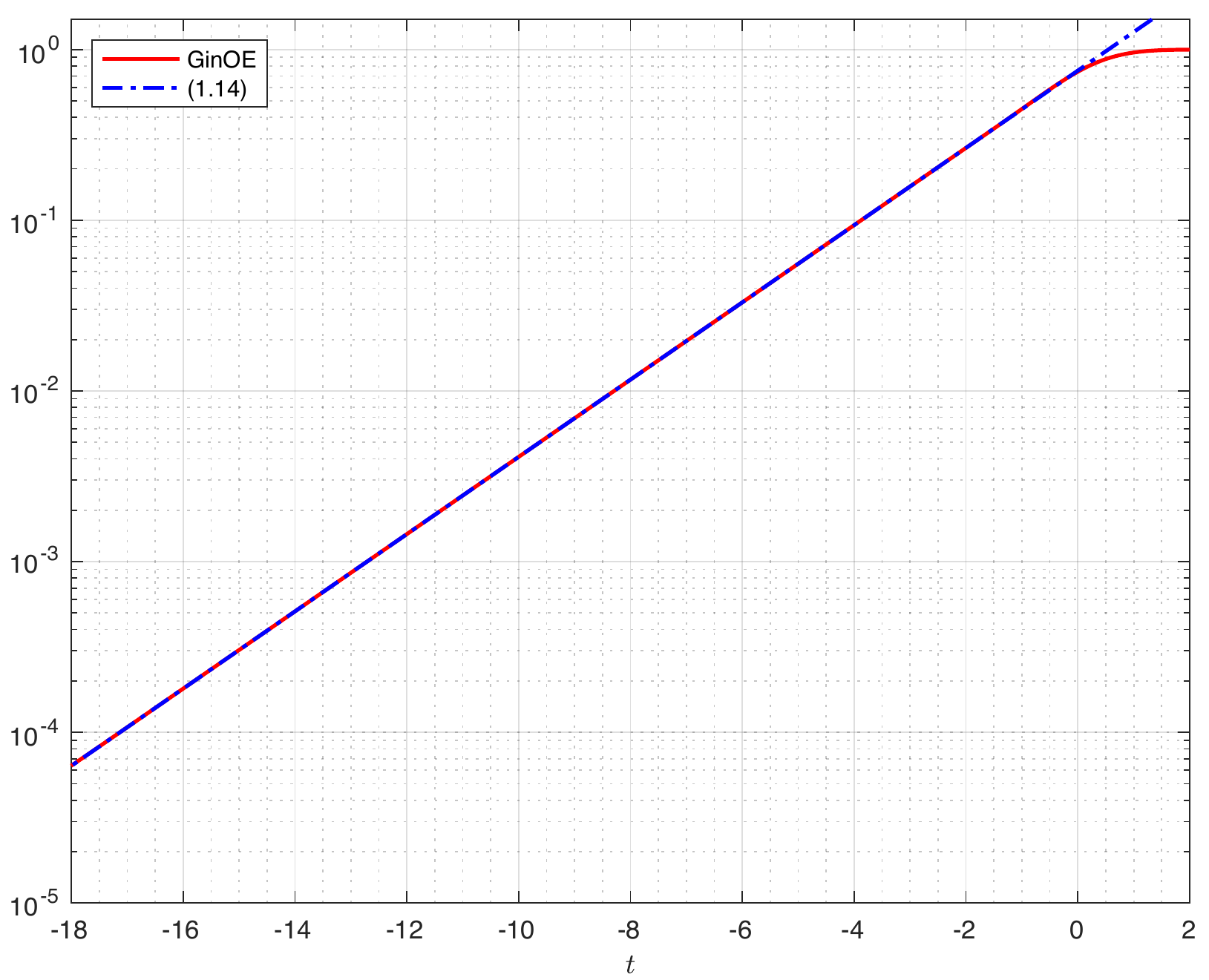}}
\caption{We double-check \eqref{e:9} on the left and \eqref{e:10} on the right (in a semilogarithmic plot) against the numerically computed values of $F(t;1)$ based on \eqref{e:11} with $\eta_0(1)=0.75277069$. Once more we have used the Nystr\"om method with a Gauss-Legendre quadrature rule and $m=50$ quadrature points.}
\label{figure11}
\end{figure}
\end{center}

A closed form computation of $\eta_0(\gamma)$ in \eqref{e:10} is beyond the methods developed in this paper. In \cite[$(2.26),(2.30)$]{F2}, Forrester derives a closed form series representation for $\eta_0(1)$, namely
\begin{equation*}
	\exp\left[\ln 2-\frac{1}{4}+\frac{1}{4\pi}\sum_{n=2}^{\infty}\frac{1}{n}\left(-\pi+\sum_{m=1}^{n-1}\frac{1}{\sqrt{m(n-m)}}\right)\right]\approx 1.06470738.
\end{equation*}	
However, after using a simple approximation for $\eta_0(1)$ obtained by numerically computing the ratio of $F(t;1)$ from \eqref{e:11} and $\e^{\eta_1(1)t}$ for large negative $t$, our result 
\begin{equation*}
	\eta_0(1)=0.75277069,
\end{equation*}
does not match \cite[$(2.26),(2.30)$]{F2}. This discrepancy needs to be further investigated. 
Finally, in the remaining Figure \ref{figure111} we showcase the qualitatively different asymptotic behaviors of $\Im\big(y(\frac{x}{2};\gamma)\big)$ on one hand and $u(x;\gamma)$ on the other, compare Remark \ref{ASegur}.
\begin{center}
\begin{figure}[tbh]
\resizebox{0.451\textwidth}{!}{\includegraphics{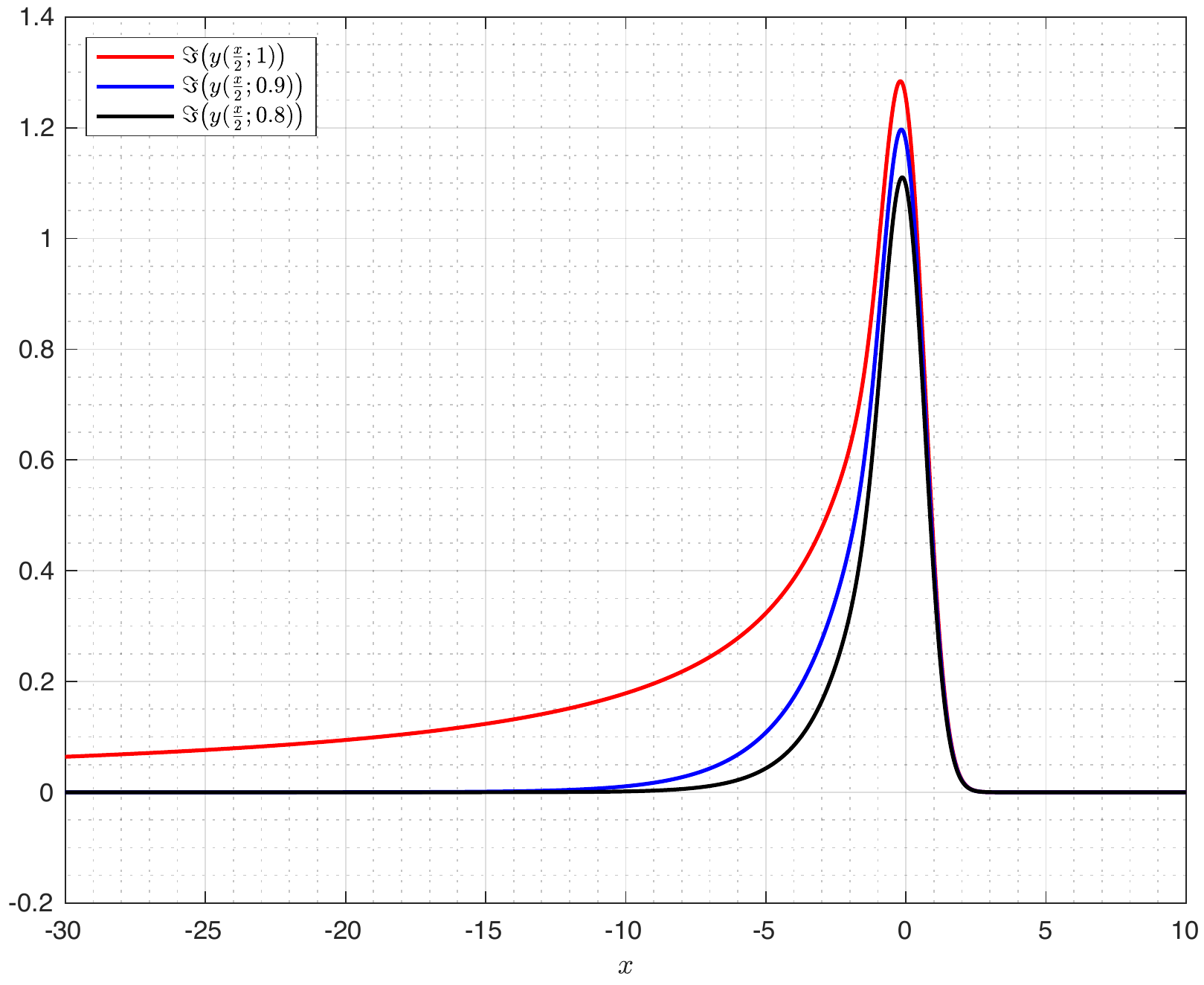}}\ \ \ \ \ \resizebox{0.451\textwidth}{!}{\includegraphics{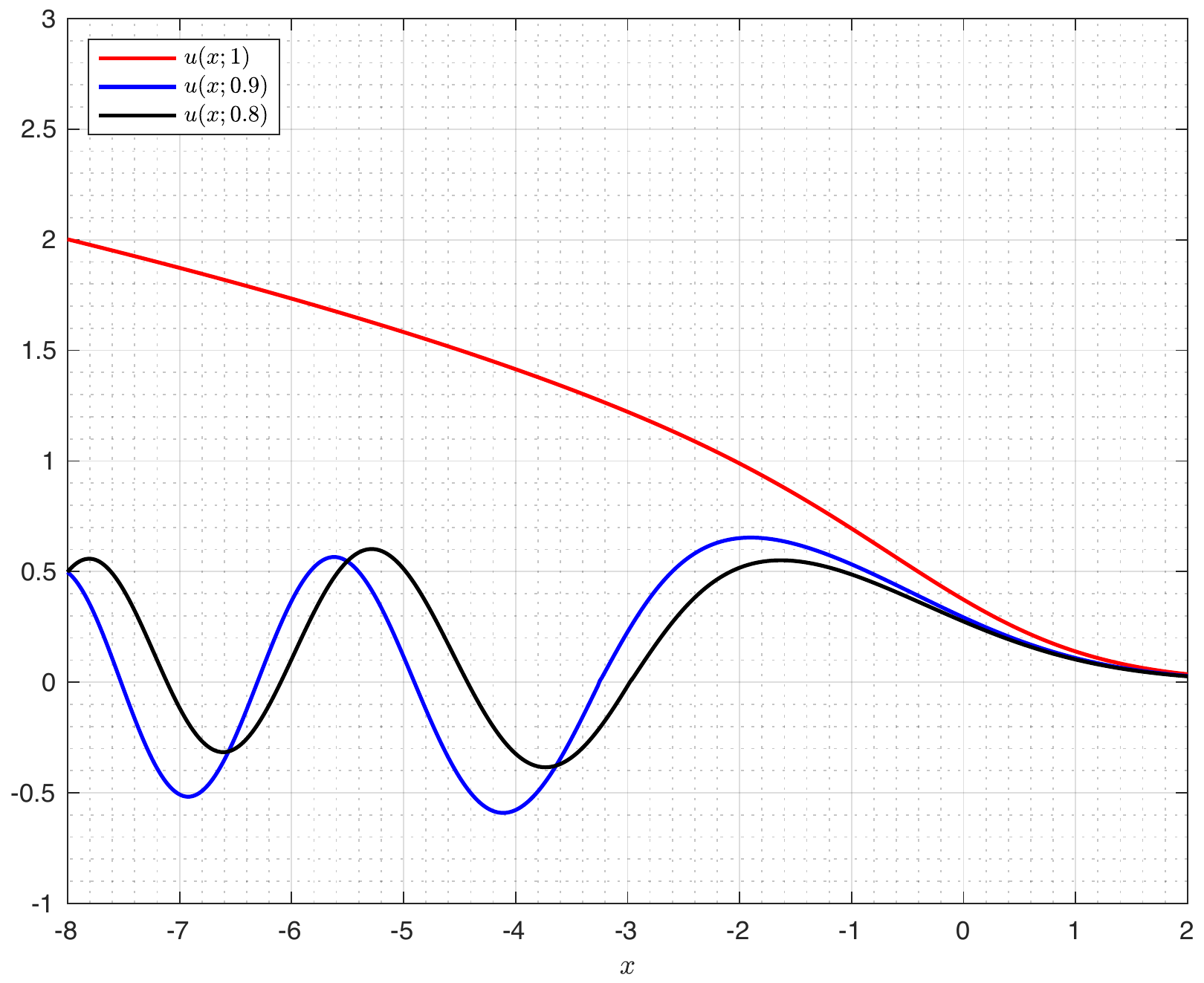}} 
\caption{We compare $\Im\big(y(\frac{x}{2};\gamma)\big)$ on the left to $u(x;\gamma)$ on the right for varying values of $\gamma$. While the solution entry $y(x;\gamma)$ to RHP \ref{master0} decays to zero as $x\rightarrow-\infty$ for all $\gamma\in(0,1]$, this is not true for the Painlev\'e-II transcendent.}
\label{figure111}
\end{figure}
\end{center} 

\subsection{Outline of paper} Towards the end of our introduction we now offer a short outline for the remaining sections of the paper. In section \ref{prelim} we first summarize a few basic properties of the operator $T$ on $L^2(\mathbb{R})$ with kernel \eqref{e:2} and show that $T\chi_t$ is indeed of integrable type \cite{IIKS}, up to Fourier conjugation. Still, instead of deriving an integrable system for $F(t;\gamma)$ at this point we employ further simplification steps in Section \ref{contin} that allow to match the thereby obtained RHP \ref{master} almost immediately with RHP \ref{master0}. This will in turn prove the first part of Theorem \ref{main1} in \eqref{l:41} below once combined with appropriate right tail estimates that we derive by nonlinear steepest descent arguments, compare Subsection \ref{right1}. These steps are then followed up in Section \ref{contin2} by an explicit evaluation of $\Gamma_{t\gamma}$ in \eqref{gdef} in terms of Riemann-Hilbert data, and the second part in \eqref{e:6} is then also proven. While carrying out the aforementioned steps we derive estimate \eqref{e:9} en route and complete the proof of Corollary \ref{main2} afterwards in Section \ref{left1}. The nonlinear steepest descent techniques for the left tail are standard except for the appearance of certain collapsing jump contours. For this reason we provide the necessary small norm estimates of the underlying (unbounded) Cauchy operators in Appendix \ref{appA}. The paper closes with the derivation of \eqref{e:11} in Section \ref{Spohn} which heavily relies on the proof technique presented in \cite{FS} for the corresponding GOE result.

\section{Preliminary steps}\label{prelim}

We begin with the following result which is standard for, say, the Airy operator (see for instance \cite[Lemma $6.15$]{BDS}) but which does not appear in the literature for $T\chi_t$, to the best of our knowledge.
\begin{lem}\label{reg} For every $t\in\mathbb{R}$, the self-adjoint operator $T\chi_t$ satisfies $0\leq T\chi_t\leq 1$ and thus $\|T\chi_t\|\leq 1$. Moreover, $1-\gamma T\chi_t$ is invertible on $L^2(\mathbb{R})$ for all $\gamma\in[0,1]$.
\end{lem}
\begin{proof} Recall the Gaussian integral 
\begin{equation}\label{l:1}
	\e^{-x^2}=\frac{1}{2\sqrt{\pi}}\int_{-\infty}^{\infty}\e^{-\frac{1}{4}y^2+\im xy}\,\d y,\ \ \ x\in\mathbb{R},
\end{equation}
and note that for $f\in L^2(\mathbb{R})$,
\begin{equation}\label{l:2}
	0\leq\langle f,T\chi_tf\rangle_{L^2(\mathbb{R})}=\frac{1}{\pi}\int_0^{\infty}\left|\int_{-\infty}^{\infty}\e^{-(x+u)^2}f_t(x)\,\d x\right|^2\d u\leq\frac{1}{\pi}\int_{-\infty}^{\infty}\left|\int_{-\infty}^{\infty}\e^{-(x+u)^2}f_t(x)\,\d x\right|^2\d u,
\end{equation}
where we abbreviate $f_t(x):=f(x)\chi_{(t,+\infty)}(x)$. Hence, with \eqref{l:1}, we compute
\begin{equation*}
	\int_{-\infty}^{\infty}\e^{-(x+u)^2}f_t(x)\,\d x=\frac{1}{\sqrt{2}}\int_{-\infty}^{\infty}\e^{-\frac{1}{4}y^2}\widehat{f}_t(-y)\e^{\im u y}\,\d y=\sqrt{\pi}\,\widehat{g}(-u),
\end{equation*}
where $\widehat{f}_t(y)=\frac{1}{\sqrt{2\pi}}\int_{-\infty}^{\infty}f_t(x)\e^{-\im yx}\,\d y$ is the Fourier transform of $f_t$ and $g(y):=\e^{-\frac{1}{4}y^2}\widehat{f}_t(-y)$. Thus together in \eqref{l:2},
\begin{eqnarray}
	0\leq\langle f,T\chi_t f\rangle_{L^2(\mathbb{R})}&\leq&\int_{-\infty}^{\infty}|\widehat{g}(-u)|^2\,\d u=\int_{-\infty}^{\infty}|g(y)|^2\,\d y=\int_{-\infty}^{\infty}\e^{-\frac{1}{2}y^2}|\widehat{f}_t(-y)|^2\,\d y\label{l:3}\\
	&\leq &\int_{-\infty}^{\infty}|\widehat{f}_t(-y)|^2\,\d y=\int_{-\infty}^{\infty}|f_t(y)|^2\,\d y\leq \int_{-\infty}^{\infty}|f(y)|^2\,\d y=\langle f,f\rangle_{L^2(\mathbb{R})},\nonumber
\end{eqnarray}
using Plancherel's theorem in the first and third equality. Hence $0\leq T\chi_t\leq 1$ and by self-adjointness also
\begin{equation*}
	 \|T\chi_t\|=\sup_{\|f\|_{L^2(\mathbb{R})}=1}|\langle f,T\chi_t f\rangle_{L^2(\mathbb{R})}|\leq 1.
\end{equation*}
For invertibility, we assume there is $f\in L^2(\mathbb{R})$, not identically zero, such that $\gamma T\chi_tf=f$ with $\gamma\in(0,1]$. By \eqref{l:3} we must therefore have equality
\begin{equation*}
	\langle f,T\chi_tf\rangle_{L^2(\mathbb{R})}=\gamma\int_{-\infty}^{\infty}|\widehat{g}(-u)|^2\,\d u,
\end{equation*}
but from \eqref{l:2} (without estimating) then also
\begin{equation*}
	\langle f,T\chi_t f\rangle_{L^2(\mathbb{R})}=\gamma\int_0^{\infty}|\widehat{g}(-u)|^2\,\d u.
\end{equation*}
So $\widehat{g}(u)=0$ for $u<0$, i.e.
\begin{equation}\label{l:4}
	\int_t^{\infty}\e^{-(x+u)^2}f(x)\,\d x=0\ \ \ \ \textnormal{for}\ \ u<0.
\end{equation}
Using analytic properties of the exponential, we then conclude that \eqref{l:4} must also hold for $u>0$ and therefore $f\equiv 0$, a contradiction.
\end{proof}
Our next steps will make use of a slight generalization of \eqref{l:1}, namely the following contour integral formula: for any smooth non self-intersecting contour $\Gamma$ oriented from $\infty\cdot\e^{\im\alpha}$ to $\infty\cdot\e^{\im\beta}$ with $\alpha\in(\frac{3\pi}{4},\frac{5\pi}{4})$ and $\beta\in(-\frac{\pi}{4},\frac{\pi}{4})$, see e.g. Figure \ref{fig2}, we have
\begin{equation}\label{l:5}
	\e^{-x^2}=\frac{1}{2\sqrt{\pi}}\int_{\Gamma}\e^{-\frac{1}{4}\lambda^2\pm\im x\lambda}\,\d\lambda,\ \ \ \ x\in\mathbb{R}.
\end{equation}
\begin{figure}[tbh]
\begin{tikzpicture}[xscale=0.6,yscale=0.6]
\draw [->] (-6,0) -- (6,0) node[below]{{\small $\Re\lambda$}};
\draw [->] (0,-4) -- (0,4) node[left]{{\small $\Im\lambda$}};
\draw [very thin, dashed, color=darkgray,-] (0,0) -- (3.5,3.5) node[right]{$\frac{\pi}{4}$};
\draw [very thin, dashed, color=darkgray,-] (0,0) -- (3.5,-3.5) node[right]{$-\frac{\pi}{4}$};
\draw [very thin, dashed, color=darkgray,-] (0,0) -- (-3.5,3.5) node[left]{$\frac{3\pi}{4}$};
\draw [very thin, dashed, color=darkgray,-] (0,0) -- (-3.5,-3.5) node[left]{$\frac{5\pi}{4}$};
\draw [fill=cyan, dashed,opacity=0.7] (0,0) -- (2.828427124,-2.828427124) arc (-45:45:4cm) -- (0,0);
\draw [fill=cyan, dashed,opacity=0.7] (0,0) -- (-2.828427124,2.828427124) arc (135:225:4cm) -- (0,0);
\draw [thick,color=red,decoration={markings, mark=at position 0.25 with {\arrow{>}}, mark=at position 0.75 with {\arrow{>}}}, postaction={decorate}] plot [smooth, tension=0.5] coordinates { (-5,1.25) (-4,1) (-3.5,0.8) (-3,0.55) (-2.5,0.2) (-2,-0.2) (-1.5,-0.5) (-1,-0.7) (-0.5,-0.7) (0,-0.5) (0.5,0) (1,0.5) (1.5,0.9) (2,1.25) (2.5,1.55) (3,1.76) (3.5,1.95) (4,2.1) (4.5,2.2)} node[right]{{\small $\Gamma$}};
\end{tikzpicture}
\caption{An admissible choice for the contour $\Gamma$ in \eqref{l:5}.}
\label{fig2}
\end{figure}
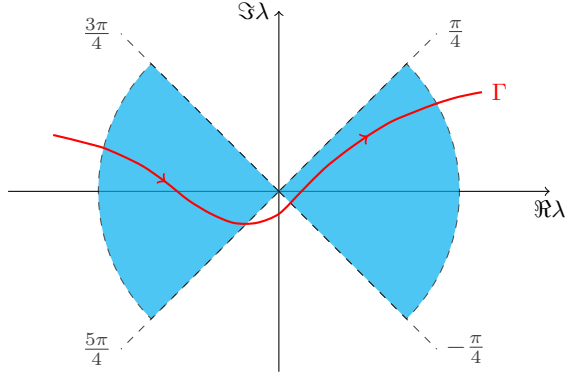

Now fix $(x,y)\in\mathbb{R}^2$ throughout and substitute \eqref{l:5} twice into \eqref{e:2}, once with the $(+)$ sign in \eqref{l:5} and once with the $(-)$ sign:
\begin{equation*}
	T(x,y)=\frac{1}{(2\pi)^2}\int_{\Gamma_{\lambda}}\int_{\Gamma_w}\e^{-\frac{1}{4}(\lambda^2+w^2)}\e^{-\im(x\lambda-yw)}\left[\int_0^{\infty}\e^{-\im u(\lambda-w)}\,\d u\right]\d w\,\d\lambda.
\end{equation*}
So provided we choose $(\lambda,w)\in\Gamma_{\lambda}\times\Gamma_w$ such that $\Im w>\Im\lambda$,
\begin{equation}\label{l:6}
	T(x,y)=\frac{1}{(2\pi)^2}\int_{\Gamma_{\lambda}}\int_{\Gamma_w}\frac{\e^{-\frac{1}{4}(\lambda^2+w^2)-\im(x\lambda-yw)}}{\im(\lambda-w)}\,\d w\,\d\lambda.
\end{equation}
Next we use residue theorem.
\begin{lem}\label{residue} Suppose $w\in\Gamma_w$ satisfies $\Im w>0$. Then for any $y,t\in\mathbb{R}:y\neq t$,
\begin{equation}\label{l:7}
	\frac{1}{2\pi\im}\int_{-\infty}^{\infty}\e^{\im(\mu-w)(y-t)}\frac{\d\mu}{\mu-w}=\chi_{(t,+\infty)}(y).
\end{equation}
\end{lem}
\begin{proof} The integral on the left-hand side in \eqref{l:7} is well-defined under the given assumptions and can be evaluated as follows: define the entire function $f(\mu):=\e^{\im\mu(y-t)},\mu\in\mathbb{C}$ and consider for $y>t$ a contour integral of $f$ along the closed rectangle $\Sigma_+(R)$ shown in Figure \ref{fig3} below where $R>|w|+1$ and $w\in\mathbb{C}:\Im w>0$ is fixed. By residue theorem,
\begin{equation}\label{l:8}
	\frac{1}{2\pi\im}\ointctrclockwise_{\Sigma_+(R)}f(\mu)\frac{\d\mu}{\mu-w}=\e^{\im w(y-t)},
\end{equation}
and since with some $C=C(w)>0$
\begin{equation*}
	\left|\int_{\Sigma_{1,3}}f(\mu)\frac{\d\mu}{\mu-w}\right|\leq \frac{C}{R(y-t)}\rightarrow 0,\ \ \ \ \textnormal{as well as}\ \ \ 
	\left|\int_{\Sigma_2}f(\mu)\frac{\d\mu}{\mu-w}\right|\leq C\e^{-R(y-t)}\rightarrow 0,\ \ \ \textnormal{as} \ \ R\rightarrow+\infty,
\end{equation*}
identity \eqref{l:7} follows for $y>t$ from \eqref{l:8} in the limit $R\rightarrow+\infty$.
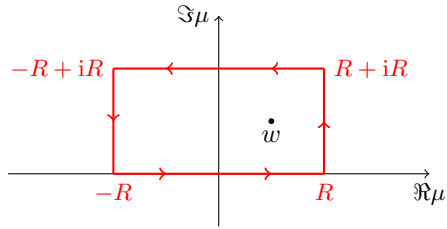
\begin{figure}[tbh]
\begin{tikzpicture}[xscale=0.7,yscale=0.7]
\draw [->] (-4,0) -- (4,0) node[below]{{\small $\Re\mu$}};
\draw [->] (0,-1) -- (0,3) node[left]{{\small $\Im\mu$}};
\draw[thick, color=red, decoration={markings, mark=at position 0.5 with {\arrow{>}}}, postaction={decorate}] (2,0) -- (2,2) node[right]{{\small $R+\im R$}};
\draw[thick, color=red, decoration={markings, mark=at position 0.25 with {\arrow{>}}, mark=at position 0.75 with {\arrow{>}}}, postaction={decorate}] (2,2) -- (-2,2) node[left]{{\small $-R+\im R$}};
\draw[thick, color=red, decoration={markings, mark=at position 0.5 with {\arrow{>}}}, postaction={decorate}] (-2,2) -- (-2,0) node[below]{{\small $-R$}};
\draw[thick, color=red, decoration={markings, mark=at position 0.25 with {\arrow{>}}, mark=at position 0.75 with {\arrow{>}}}, postaction={decorate}] (-2,0) -- (2,0) node[below]{{\small $R$}};
\draw[fill] (1,1) circle [radius=0.045];
\node [below] at (1,1) {$w$};
\end{tikzpicture}
\caption{Red contour $\Sigma_+(R)$ in the upper half plane consisting of two vertical pieces $\Sigma_{1,3}$ (left, right), top piece $\Sigma_2$ as well as $[-R,R]\subset\mathbb{R}$.}
\label{fig3}
\end{figure}

For $y<t$ we evaluate the contour integral \eqref{l:8} along an analogous contour $\Sigma_-(R)$ in the lower half-plane, however this time the residue theorem does not yield a non-vanishing contribution. Hence, \eqref{l:7} follows also for $y<t$ in the limit $R\rightarrow+\infty$.
\end{proof}
Let us now choose $\Gamma_{\lambda}=\mathbb{R}$ in \eqref{l:6} and $\Gamma_w\equiv\Gamma$ as any smooth non self-intersecting contour in the upper $w$ half-plane. Together with Lemma \ref{residue} we find
\begin{equation}\label{l:9}
	T(x,y)\chi_{(t,+\infty)}(y)=\int_{-\infty}^{\infty}\frac{\e^{-\im x\lambda}}{\sqrt{2\pi}}\underbrace{\left[\frac{1}{(2\pi)^2}\int_{\Gamma}\frac{\e^{-\frac{1}{4}(\lambda^2+w^2)-\im t(\mu-w)}}{(\lambda-w)(w-\mu)}\,\d w\right]}_{=:E(\lambda,\mu)}\frac{\e^{\im y\mu}}{\sqrt{2\pi}}\,\d\mu\,\d\lambda.
\end{equation}
Thus, provided we let $\mathcal{F}:L^1(\mathbb{R})\cap L^2(\mathbb{R})\rightarrow L^2(\mathbb{R})$ denote the standard Fourier transform, i.e.
\begin{equation*}
	\widehat{f}(x)\equiv(\mathcal{F}f)(x):=\frac{1}{\sqrt{2\pi}}\int_{-\infty}^{\infty}\e^{-\im x\lambda}f(\lambda)\,\d\lambda,\ \ \ x\in\mathbb{R},
\end{equation*}
which extends to a unitary operator on $L^2(\mathbb{R})$ by classical theory, then \eqref{l:9} shows that $T\chi_t$ on $L^2(\mathbb{R})$ is simply equal to the operator composition $\mathcal{F}E\mathcal{F}^{-1}$ where $E:L^2(\mathbb{R})\rightarrow L^2(\mathbb{R})$ is the integral operator on $L^2(\mathbb{R})$ with kernel $E(\lambda,\mu)$ given in \eqref{l:9}. In order to verify certain regularity properties of $E$ (and other operators to follow) it will be convenient to abide to the following convention.
\begin{conv}\label{id:1}
From now on we shall think of our operators $T\chi_t,E$ and $\mathcal{F}$ as acting, not on $L^2(\mathbb{R},\d x)$, but on an extended space $L^2(\Omega,|\d\lambda|)$ (where $\mathbb{R}\subset\Omega\subset\mathbb{C}$ for some oriented contour $\Omega$ to be specified below and $|\d\lambda|$ is the arc-length measure) and to have kernel
\begin{equation*}
	T_{\textnormal{ext}}(\lambda,\mu):=T(\lambda,\mu)\chi_{(t,+\infty)}(\mu)\chi_{\mathbb{R}}(\lambda),\ \ \ \ (\lambda,\mu)\in\Omega\times\Omega,
\end{equation*}
as well as $E_{\textnormal{ext}}(\lambda,\mu):=E(\lambda,\mu)\chi_{\mathbb{R}}(\lambda)\chi_{\mathbb{R}}(\mu)$ and 
$\mathcal{F}_{\textnormal{ext}}(\lambda,\mu):=\frac{1}{\sqrt{2\pi}}\e^{-\im\lambda\mu}\chi_{\mathbb{R}}(\lambda)\chi_{\mathbb{R}}(\mu)$.
\end{conv}
Provided we modify the distributional kernel of the identity in accordance with Convention \ref{id:1}, equation \eqref{l:9} thus establishes the operator identity,
\begin{equation}\label{Fid:1}
	 1-\gamma T\chi_t\upharpoonright_{L^2(\mathbb{R})}=1-\gamma T_{\textnormal{ext}}\upharpoonright_{L^2(\Omega)}=\mathcal{F}_{\textnormal{ext}}(1-\gamma E_{\textnormal{ext}}\upharpoonright_{L^2(\Omega)})\mathcal{F}_{\textnormal{ext}}^{-1},\ \ \ \ \gamma\in[0,1].
\end{equation}
The main motivation behind Convention \ref{id:1} comes from the following result.
\begin{lem}\label{tracecl1} The operator $E_{\textnormal{ext}}$ on $L^2(\Omega=\mathbb{R}\sqcup\Gamma,|\d\lambda|)$ with kernel given in Convention \ref{id:1} is trace-class.
\end{lem}
\begin{proof} Observe the factorization $E_{\textnormal{ext}}=E_1E_2$ where $E_1:L^2(\Omega,|\d\lambda|)\rightarrow L^2(\Omega,|\d\lambda|)$ has kernel
\begin{equation*}
	E_1(\lambda,w)=\frac{1}{2\pi}\frac{\e^{-\frac{1}{4}\lambda^2-\frac{1}{8}w^2}}{\lambda-w}\chi_{\mathbb{R}}(\lambda)\chi_{\Gamma}(w),
\end{equation*}
and $E_2:L^2(\Omega,|\d\lambda|)\rightarrow L^2(\Omega,|\d\lambda|)$ has kernel
\begin{equation*}
	E_2(w,\mu)=\frac{1}{2\pi}\frac{\e^{-\frac{1}{8}w^2-\im t(\mu-w)}}{w-\mu}\chi_{\Gamma}(w)\chi_{\mathbb{R}}(\mu).
\end{equation*}
But both, $E_1$ and $E_2$, are Hilbert-Schmidt integral operators on $L^2(\Omega,|\d\lambda|)$ since
\begin{equation*}
	\int_{\Omega}\int_{\Omega}|E_j(z_1,z_2)|^2|\d z_1|\,|\d z_2|<\infty,\ \ j=1,2,
\end{equation*}
so $E_{\textnormal{ext}}=E_1E_2$ is trace-class on $L^2(\Omega,|\d\lambda|)$.
\end{proof}
Observe that $E_{\textnormal{ext}}$ is already of integrable type in the sense of \cite{IIKS} once we use partial fractions. Still, it is preferable to massage $E_{\textnormal{ext}}$ a bit further: Introduce $M:L^2(\Omega,|\d\lambda|)\rightarrow L^2(\Omega,|\d\lambda|)$ as multiplication by $\e^{-\frac{1}{8}\lambda^2}\chi_{\mathbb{R}}(\lambda)$ and note that both operators $E_{\textnormal{ext}}=(E_{\textnormal{ext}}M)M^{-1}$ and $N:=M^{-1}E_{\textnormal{ext}}M$ are trace-class on $L^2(\Omega,|\d\lambda|)$. 
\begin{rem}\label{tc} The operator $N$ has kernel
\begin{equation*}
	N(\lambda,\mu)=\frac{1}{(2\pi)^2}\int_{\Gamma}\frac{\e^{-\frac{1}{8}(\lambda^2+\mu^2)-\frac{1}{4}w^2+\im t(w-\mu)}}{(\lambda-w)(w-\mu)}\,\d w\,\chi_{\mathbb{R}}(\lambda)\chi_{\mathbb{R}}(\mu),
\end{equation*}
and is therefore trace-class since it can be factorized into $N=N_1N_2$ where $N_j:L^2(\Omega,|\d\lambda|)\rightarrow L^2(\Omega,|\d\lambda|)$ have Hilbert-Schmidt kernels
\begin{equation}\label{l:14}
	N_1(\lambda,w)=\frac{1}{2\pi}\frac{\e^{-\frac{1}{8}(\lambda^2+w^2)}}{\lambda-w}\chi_{\mathbb{R}}(\lambda)\chi_{\Gamma}(w),\ \ \ \ N_2(w,\mu)=\frac{1}{2\pi}\frac{\e^{-\frac{1}{8}(w^2+\mu^2)+\im t(w-\mu)}}{w-\mu}\chi_{\Gamma}(w)\chi_{\mathbb{R}}(\mu).
\end{equation}
\end{rem}\bigskip
Concluding our preliminary steps, we note that by the conjugation invariance of the Fredholm determinant and Sylvester's determinant identity \cite[Chapter IV, $(5.9)$]{GK},
\begin{eqnarray}
	\det(1-\gamma T\chi_t\upharpoonright_{L^2(\mathbb{R})})&\stackrel{\eqref{Fid:1}}{=}&\det(1-\gamma E_{\textnormal{ext}}\upharpoonright_{L^2(\Omega)})=\det(1-\gamma(E_{\textnormal{ext}}M)M^{-1}\upharpoonright_{L^2(\Omega)})\nonumber\\
	&=&\det(1-\gamma M^{-1}E_{\textnormal{ext}}M\upharpoonright_{L^2(\Omega)})
	=\det(1-\gamma N\upharpoonright_{L^2(\Omega)}).\label{Fid:2}
\end{eqnarray}
\section{Riemann-Hilbert problem and proof of Theorem \ref{main1}, part 1}\label{contin}
The seemingly quickest way to derive the first parts of \eqref{e:6} makes use of the factorization $N=N_1N_2$ in \eqref{l:14} and a subsequent operator identity that further simplifies our above $1-N\upharpoonright_{L^2(\Omega)}$. 
\subsection{Fredholm determinant identities}
We begin with the following Lemma which improves the statement of Remark \ref{tc} at the cost of further extending $\Omega$.
\begin{lem}\label{alltrace} The integral operators $N_j:L^2(\mathring{\Omega},|\d\lambda|)\rightarrow L^2(\mathring{\Omega},|\d\lambda|)$ with kernels \eqref{l:14} defined on the extended space $\mathring{\Omega}=\Omega\sqcup(\mathbb{R}+\im\delta)$ for some sufficiently small $\delta>0$ are trace-class.
\end{lem}
\begin{proof} We keep $\lambda\in\mathbb{R}$ and $w\in\Gamma$ (as before with $\Gamma$ any smooth non self-intersecting contour in the upper half-plane) fixed throughout. Now choose $0<\delta<\Im w$ and observe that 
\begin{equation}\label{l:16}
	\frac{1}{2\pi\im}\int_{\mathbb{R}+\im\delta}\frac{\d\nu}{(\lambda-\nu)(\nu-w)}=\frac{1}{\lambda-w}.
\end{equation}
Indeed, the integral in the left-hand side of \eqref{l:16} is well-defined and with $\Sigma(R)$ as shown in Figure \ref{fig4} below, where $R>|w|+1$, we find from residue theorem
\begin{equation}\label{l:17}
	\frac{1}{2\pi\im}\ointctrclockwise_{\Sigma(R)}\frac{\d\nu}{(\lambda-\nu)(\nu-w)}=\frac{1}{\lambda-w}.
\end{equation}
But with some $C=C(w,\lambda)>0$,
\begin{equation*}
	\left|\int_{\Sigma_{1,3}}\frac{\d\nu}{(\lambda-\nu)(\nu-w)}\right|\leq\frac{C}{R}\rightarrow 0,\ \ \ \ \textnormal{as well as}\ \ \ \ \left|\int_{\Sigma_2}\frac{\d\nu}{(\lambda-\nu)(\nu-w)}\right|\leq\frac{C}{R}\rightarrow 0,\ \ \ \textnormal{as}\ \ R\rightarrow+\infty,
\end{equation*}
so identity \eqref{l:16} follows from \eqref{l:17} in the limit $R\rightarrow+\infty$.
\begin{figure}[tbh]
\begin{tikzpicture}[xscale=0.7,yscale=0.7]
\draw [->] (-4,0) -- (4,0) node[below]{{\small $\Re\nu$}};
\draw [->] (0,-1) -- (0,4) node[left]{{\small $\Im\nu$}};
\draw [thick, color=red, decoration={markings, mark=at position 0.25 with {\arrow{>}}, mark=at position 0.75 with {\arrow{>}}}, postaction={decorate}] (-2,1) -- (2,1) node[right]{{\small $R+\im\delta$}};
\draw [thick, color=red, decoration={markings, mark=at position 0.5 with {\arrow{>}}}, postaction={decorate}] (2,1) -- (2,3) node[right]{{\small $R+\im(\delta+R)$}};
\draw [thick, color=red, decoration={markings, mark=at position 0.25 with {\arrow{>}}, mark=at position 0.75 with {\arrow{>}}}, postaction={decorate}] (2,3) -- (-2,3) node[left]{{\small $-R+\im(\delta+R)$}};
\draw [thick, color=red, decoration={markings, mark=at position 0.5 with {\arrow{>}}}, postaction={decorate}] (-2,3) -- (-2,1) node[left]{{\small $-R+\im\delta$}};
\draw[fill] (1,0) circle [radius=0.045];
\node [above] at (1,0) {$\lambda$};
\draw[fill] (-1,2) circle [radius=0.045];
\node [below] at (-1,2) {$w$};
\end{tikzpicture}
\caption{Red contour $\Sigma(R)$ in the upper half plane consisting of two vertical pieces $\Sigma_{1,3}$ (left, right), top piece $\Sigma_2$ as well as $[-R+\im\delta,R+\im\delta]\subset\mathbb{R}+\im\delta$.}
\label{fig4}
\end{figure}
Now use \eqref{l:16} and rewrite \eqref{l:14} in the following way,
\begin{equation*}
	N_1(\lambda,w)=\frac{1}{2\pi\im}\int_{\mathbb{R}+\im\delta}\frac{\e^{-\frac{1}{8}(\lambda^2+w^2)}}{2\pi(\lambda-\nu)(\nu-w)}\,\d\nu\,\chi_{\mathbb{R}}(\lambda)\chi_{\Gamma}(w),
\end{equation*}
so that $N_1=N_{11}N_{12}$ where $N_{1j}:L^2(\mathring{\Omega},|\d\lambda|)\rightarrow L^2(\mathring{\Omega},|\d\lambda|)$ have Hilbert-Schmidt kernels
\begin{equation*}
	N_{11}(\lambda,\nu)=\frac{1}{2\pi\im}\frac{\e^{-\frac{1}{8}\lambda^2}}{\lambda-\nu}\chi_{\mathbb{R}}(\lambda)\chi_{\mathbb{R}+\im\delta}(\nu),\ \ \ \ \ \ N_{12}(\nu,w)=\frac{1}{2\pi}\frac{\e^{-\frac{1}{8}w^2}}{\nu-w}\chi_{\mathbb{R}+\im\delta}(\nu)\chi_{\Gamma}(w).
\end{equation*}
Similarly with \eqref{l:17} for $N_2(w,\mu)$ in \eqref{l:14},
\begin{equation*}
	N_2(w,\mu)=\frac{1}{2\pi\im}\int_{\mathbb{R}+\im\delta}\frac{\e^{-\frac{1}{8}(w^2+\mu^2)+\im t(w-\mu)}}{(\nu-\mu)(\nu-w)}\,\d\nu\,\chi_{\Gamma}(w)\chi_{\mathbb{R}}(\mu),
\end{equation*}
and thus $N_2=N_{21}N_{22}$ with $N_{2j}:L^2(\mathring{\Omega},|\d\lambda|)\rightarrow L^2(\mathring{\Omega},|\d\lambda|)$ once more Hilbert-Schmidt,
\begin{equation*}
	N_{21}(w,\nu)=\frac{1}{2\pi\im}\frac{\e^{-\frac{1}{8}w^2+\im tw}}{\nu-w}\chi_{\Gamma}(w)\chi_{\mathbb{R}+\im\delta}(\nu),\ \ \ \ \ \ N_{22}(\nu,\mu)=\frac{1}{2\pi}\frac{\e^{-\frac{1}{8}\mu^2-\im t\mu}}{\nu-\mu}\chi_{\mathbb{R}+\im\delta}(\nu)\chi_{\mathbb{R}}(\mu).
\end{equation*}
Hence, $N_1$ and $N_2$ are trace-class on $L^2(\mathring{\Omega},|\d\lambda|)$ as claimed.
\end{proof}
The last Lemma allows us to compute for $j=1,2$,
\begin{equation*}
	\tr_{L^2(\mathring{\Omega})} N_j=\int_{\mathring{\Omega}}N_j(\lambda,\lambda)\,\d\lambda=0,
\end{equation*}
i.e. $N_j$ are traceless and even more, they are nilpotent on $L^2(\mathring{\Omega},|\d\lambda|)$ with $N_j^2=0$ (simply recall that $\mathbb{R}$ is disjoint from $\Gamma$). The last observation lies at the heart of the following useful operator identity.
\begin{lem}\label{clean} We have on $L^2(\mathring{\Omega},|\d\lambda|)$,
\begin{equation}\label{neat}
	(1-\sqrt{\gamma}\,N_1)^{-1}\big(1-\sqrt{\gamma}\,(N_1+N_2)\big)(1-\sqrt{\gamma}\,N_2)^{-1}=1-\gamma N_1N_2.
\end{equation}
\end{lem}
\begin{proof} By nilpotency of $N_j$ we know that $1-N_j$ is invertible on $L^2(\mathring{\Omega},|\d\lambda|)$, in fact
\begin{equation*}
	(1-\sqrt{\gamma}\,N_j\upharpoonright_{L^2(\mathring{\Omega})})^{-1}=1+\sqrt{\gamma}\,N_j\upharpoonright_{L^2(\mathring{\Omega})},\ \ j=1,2.
\end{equation*}
Now substitute the latter into the left-hand side of \eqref{neat}, multiply out and simplify using $N_j^2=0$. The desired equality follows at once.
\end{proof}
Since all factors in \eqref{neat} are of the form identity plus trace-class (by Lemma \ref{alltrace} and the triangle inequality for the trace norm) we are allowed to use the multiplicative nature of the Fredholm determinant \cite[Chapter II, $(5.1)$]{GK} and conclude from \eqref{neat},
\begin{equation}\label{neat:2}
	\det(1+\sqrt{\gamma}N_1\upharpoonright_{L^2(\mathring{\Omega})})\det\big(1-\sqrt{\gamma}(N_1+N_2)\upharpoonright_{L^2(\mathring{\Omega})}\big)\det(1+\sqrt{\gamma}N_2\upharpoonright_{L^2(\mathring{\Omega})})=\det(1-\gamma N_1N_2\upharpoonright_{L^2(\mathring{\Omega})}).
\end{equation}
But from the Plemelj-Smithies formula \cite[Chapter II, Theorem $3.1$]{GK} and our previous comments about vanishing traces and nilpotency of $N_j$,
\begin{equation}\label{neat:3}
	\det(1-\sqrt{\gamma}N_j\upharpoonright_{L^2(\mathring{\Omega})})=\exp\left[-\sum_{k=1}^{\infty}\frac{\gamma^{\frac{k}{2}}}{k}\tr_{L^2(\mathring{\Omega})} N_j^k\right]=1,\ \ \ \ j=1,2,
\end{equation}
i.e. we have just established the following Fredholm determinant equality.
\begin{prop}\label{cool} We have 
\begin{equation}\label{neat:4}
	\det(1-\gamma T\chi_t\upharpoonright_{L^2(\mathbb{R})})=\det(1-G\upharpoonright_{L^2(\mathring{\Omega})}),
\end{equation}
where $G:L^2(\mathring{\Omega},|\d\lambda|)\rightarrow L^2(\mathring{\Omega},|\d\lambda|)$ is trace class with kernel
\begin{equation*}
	G(\lambda,\mu;\gamma)=\sqrt{\gamma}\big(N_1(\lambda,\mu)+N_2(\lambda,\mu)\big)=\frac{{\bf f}^{\intercal}(\lambda){\bf g}(\mu)}{\lambda-\mu},\ \ \ \ (\lambda,\mu)\in\mathring{\Omega}\times\mathring{\Omega},
\end{equation*}
and
\begin{equation}\label{l:21} 
	{\bf f}(\lambda)=\sqrt{\frac{\gamma}{2\pi}}\,\e^{-\frac{1}{8}\lambda^2}\begin{bmatrix}\chi_{\mathbb{R}}(\lambda)\\ \e^{\im t\lambda}\chi_{\Gamma}(\lambda)\end{bmatrix},\ {\bf g}(\mu)=\frac{\e^{-\frac{1}{8}\mu^2}}{\sqrt{2\pi}}\begin{bmatrix}\chi_{\Gamma}(\mu)\\ \e^{-\im t\mu}\chi_{\mathbb{R}}(\mu)\end{bmatrix}.
\end{equation}
\end{prop}
\begin{proof} Use \eqref{Fid:2}, \eqref{neat:2} and \eqref{neat:3}. The formul\ae\,\eqref{l:21} for the kernel of $G=\sqrt{\gamma}(N_1+N_2)$ are read off from \eqref{l:14} and this completes our proof.
\end{proof}

\subsection{Riemann-Hilbert problem} The following RHP is our starting point for the derivation of the ZS system \eqref{dNLS} and in turn \eqref{e:6}. This problem is naturally associated with the integrable operator $G$ in \eqref{neat:4} by classical theory, cf. \cite{IIKS}.
\begin{problem}[Its, Izergin, Korepin, Slavnov \cite{IIKS}, 1990]\label{master} For any $t\in\mathbb{R},\gamma\in[0,1]$, determine ${\bf Y}(z)={\bf Y}(z;t,\gamma)\in\mathbb{C}^{2\times 2}$ such that
\begin{enumerate}
	\item[(1)] ${\bf Y}(z)$ is analytic for $z\in\mathbb{C}\setminus\Omega$ and we orient $\Omega=\mathbb{R}\sqcup\Gamma$ from ``left to right" as shown in Figure \ref{fig5} below.
	\item[(2)] The boundary values ${\bf Y}_{\pm}(z)$ from the left/right side of the oriented contour $\Omega$ exist and are related by the jump condition
	\begin{equation*}
		{\bf Y}_+(z)
		={\bf Y}_-(z)\begin{bmatrix}1 & -\im\sqrt{\gamma}\,\e^{-\frac{1}{4}z^2-\im tz}\chi_{\mathbb{R}}(z)\\ -\im\sqrt{\gamma}\,\e^{-\frac{1}{4}z^2+\im tz}\chi_{\Gamma}(z) & 1\end{bmatrix},\ z\in\Omega.
	\end{equation*}
	\item[(3)] As $z\rightarrow\infty$, we have
	\begin{equation*}
		{\bf Y}(z)=\mathbb{I}+{\bf Y}_1z^{-1}+\mathcal{O}\big(z^{-2}\big),\ \ \ \ \ {\bf Y}_1={\bf Y}_1(t,\gamma)=\big[Y_1^{jk}(t,\gamma)\big]_{j,k=1}^2.
	\end{equation*}
\end{enumerate}
\end{problem}
\begin{figure}[tbh]
\begin{tikzpicture}[xscale=0.7,yscale=0.7]
\draw [->] (-5,0) -- (5,0) node[below]{{\small $\Re z$}};
\draw [->] (0,-1) -- (0,4) node[left]{{\small $\Im z$}};
\draw [very thin, dashed, color=darkgray,-] (0,0) -- (3.5,3.5) node[right]{$\frac{\pi}{4}$};
\draw [very thin, dashed, color=darkgray,-] (0,0) -- (-3.5,3.5) node[left]{$\frac{3\pi}{4}$};
\draw [fill=cyan, dashed,opacity=0.7] (0,0) -- (2.828427124,2.828427124) arc (45:0:4cm) -- (0,0);
\draw [fill=cyan, dashed,opacity=0.7] (0,0) -- (-2.828427124,2.828427124) arc (135:180:4cm) -- (0,0);
\draw [thick, color=red, decoration={markings, mark=at position 0.25 with {\arrow{>}}}, decoration={markings, mark=at position 0.75 with {\arrow{>}}}, postaction={decorate}] (-4.5,0) -- (4.5,0) node[above]{{\small $\mathbb{R}$}};
\draw [thick,color=red,decoration={markings, mark=at position 0.25 with {\arrow{>}}, mark=at position 0.85 with {\arrow{>}}}, postaction={decorate}] plot [smooth, tension=0.5] coordinates { (-5,2.75) (-4,2.5) (-3.5,2.3) (-3,2.05) (-2.5,1.7) (-2,1.3) (-1.5,1) (-1,0.8) (-0.5,0.75) (0,0.8) (0.5,1) (1,1.3) (1.5,1.6) (2,1.85) (2.5,2.05) (3,2.26) (3.5,2.45) (4,2.6) (4.5,2.7)} node[right]{{\small $\Gamma$}};
\end{tikzpicture}
\caption{The oriented jump contour $\Omega=\mathbb{R}\sqcup\Gamma$ in RHP \ref{master} with an admissible choice for $\Gamma$ in the upper half-plane.}
\label{fig5}
\end{figure}
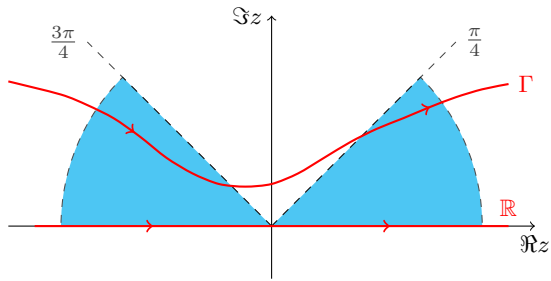
The general theory of \cite{IIKS}, see also \cite[Section $5.6$]{BDS}, asserts that for an integrable integral operator (such as our $G$ on $L^2(\mathring{\Omega},|\d\lambda|)$) its resolvent $R=(1-G)^{-1}-1$, if existent, is again of integrable type with kernel
\begin{equation}\label{l:24}
	R(\lambda,\mu;\gamma)=\frac{{\bf F}^{\intercal}(\lambda){\bf G}(\mu)}{\lambda-\mu},\ \ \ {\bf F}(\lambda)=\big((1-G\upharpoonright_{L^2(\mathring{\Omega})})^{-1}{\bf f}\big)(\lambda),\ \ {\bf G}(\mu)=\big((1-G^{\ast}\upharpoonright_{L^2(\mathring{\Omega})}{\bf g}\big)(\mu).
\end{equation}
Most importantly $R$ can be computed in terms of the solution to RHP \ref{master},
\begin{equation}\label{IIKSform}
	{\bf F}(z)={\bf Y}_{\pm}(z){\bf f}(z),\ \ \ \ {\bf G}(z)=\big({\bf Y^{\intercal}}_{\pm}(z)\big)^{-1}{\bf g}(z),\ \ \ \ z\in\Omega,
\end{equation}
where the choice of limiting values $(\pm)$ is immaterial. Thus, RHP \ref{master} is (uniquely) solvable if and only if $1-G$ is invertible on $L^2(\mathring{\Omega})$. Furthermore, the solution ${\bf Y}(z)$ to RHP \ref{master} takes the form
\begin{equation}\label{l:25}
	{\bf Y}(z)=\mathbb{I}-\int_{\Omega}{\bf F}(\lambda){\bf g}^{\intercal}(\lambda)\frac{\d\lambda}{\lambda-z},\ \ \ z\notin\Omega.
\end{equation}
We now prove solvability of RHP \ref{master} using two different arguments:\bigskip

\noindent i. {\it Argument 1}: By Proposition \ref{cool},
\begin{equation}\label{det:imp}
	\det(1-\gamma T\chi_t\upharpoonright_{L^2(\mathbb{R})})=\det(1-G\upharpoonright_{L^2(\mathring{\Omega})}),\ \ \ t\in\mathbb{R},\ \gamma\in[0,1].
\end{equation}
The left-hand side in \eqref{det:imp} is non vanishing by Lemma \ref{reg} and since $G$ is trace-class this implies that $1-G$ is invertible on $L^2(\mathring{\Omega},|\d\lambda|)$ for all $t\in\mathbb{R},\gamma\in[0,1]$, cf. \cite[Theorem $6.1$]{GK}. Thus RHP \ref{master} is solvable for all $t\in\mathbb{R},\gamma\in[0,1]$, cf. \cite{IIKS}.\bigskip

\noindent ii. {\it Argument 2}: We provide a solvability proof for RHP \ref{master} based on a vanishing lemma argument (a standard technique in Riemann-Hilbert analysis, cf. \cite{Z,FMZ,FZ}).
\begin{lem}[Vanishing lemma]\label{vanish} Let $t\in\mathbb{R},\gamma\in[0,1]$ and suppose ${\bf Y}$ satisfies conditions (1) and (2) in RHP \ref{master} above but instead of condition (3) we enforce
\begin{equation*}
	{\bf Y}(z)=\mathcal{O}(z^{-1}),\ \ \ z\rightarrow\infty.
\end{equation*}
Then ${\bf Y}\equiv 0$.
\end{lem}
\begin{proof} Let $\Delta$ denote the open region in between $\mathbb{R}$ and $\Gamma$.  Introduce the auxiliary function
\begin{equation*}
	{\bf N}(z):={\bf Y}(z)\begin{cases}\Bigl[\begin{smallmatrix} 1 & 0\smallskip\\ -\im\sqrt{\gamma}\e^{-\frac{1}{4}z^2+\im tz} & 1\end{smallmatrix}\Bigr],&z\in\Delta\smallskip\\ \mathbb{I},&z\notin\mathbb{C}\setminus\overline{\Delta}\end{cases},
\end{equation*}
and note that ${\bf N}(z)$ also satisfies ${\bf N}(z)=\mathcal{O}(z^{-1}), z\rightarrow\infty$. However, ${\bf N}(z)$ is jump-free on $\Gamma$, instead we have collapsed the jumps to the real line,
\begin{equation}\label{l:26}
	{\bf N}_+(z)={\bf N}_-(z)\begin{bmatrix}1-\gamma\e^{-\frac{1}{2}z^2} & -\im\sqrt{\gamma} \e^{-\frac{1}{4}z^2-\im tz}\\ -\im\sqrt{\gamma}\e^{-\frac{1}{4}z^2+\im tz} & 1\end{bmatrix},\ \ z\in\mathbb{R}.
\end{equation}
Next define ${\bf H}(z)={\bf N}(z){\bf N}^{\dagger}(\bar{z})$ for $z\in\mathbb{C}\setminus\mathbb{R}$ where ${\bf N}^{\dagger}$ is the Hermitian conjugate of ${\bf N}$. Since ${\bf H}$ is analytic in the upper $z$-plane, continuous down to the real line and decays of $\mathcal{O}(z^{-2})$ as $z\rightarrow\infty$, we find from Cauchy's theorem $\int_{\mathbb{R}}{\bf H}_+(z)\,\d z=0$. We add to this equation its Hermitian conjugate, 
\begin{equation}\label{l:27}
	0=\int_{-\infty}^{\infty}\left\{{\bf N}_+(z){\bf N}_-^{\dagger}(z)+{\bf N}_-(z){\bf N}_+^{\dagger}(z)\right\}\,\d z\stackrel{\eqref{l:26}}{=}2\int_{-\infty}^{\infty}\left\{{\bf N}_-(z)\begin{bmatrix}1-\gamma\e^{-\frac{1}{2}z^2} & 0\\ 0 & 1\end{bmatrix}{\bf N}_-^{\dagger}(z)\right\}\,\d z.
\end{equation}
Reading off the diagonal entries in \eqref{l:27} we find in turn
\begin{equation*}
	\int_{-\infty}^{\infty}\left\{|N_-^{21}(z)|^2\left(1-\gamma\e^{-\frac{1}{2}z^2}\right)+|N_-^{22}(z)|^2\right\}\,\d z=0=\int_{-\infty}^{\infty}\left\{|N_-^{11}(z)|^2\left(1-\gamma\e^{-\frac{1}{2}z^2}\right)+|N_-^{12}(z)|^2\right\}\,\d z,
\end{equation*}
so that ${\bf N}_-(z)\equiv 0$ by continuity of ${\bf N}_-(z)$ on $\mathbb{R}$ and thus with \eqref{l:26} also ${\bf N}_+(z)\equiv 0$. In summary, ${\bf N}(z)$ is analytic for $\Im z>0$, continuous for $\Im z\geq 0$ and we have
\begin{equation*}
	|{\bf N}(z)|\leq C,\ \ \Im z>0,\ \ \ \ \ \ \ \ \ \ \sup_{z\in\mathbb{R}}|{\bf N}_+(z)|=0.
\end{equation*}
Hence, by Carlson's theorem (cf. \cite{SI}, Theorem $5.1.2$ and Corollary $5.1.3$), ${\bf N}(z)\equiv 0$ for $\Im z\geq 0$. Dealing with $\Im z\leq 0$ in a similar fashion we establish triviality of ${\bf N}(z)$ in the whole $z$-plane and thus also ${\bf Y}(z)\equiv 0$.
\end{proof}
\begin{cor}\label{solv:1} The Riemann-Hilbert problem \ref{master} for ${\bf Y}(z;t,\gamma)$ has a unique solution for every $t\in\mathbb{R},\gamma\in[0,1]$. Moreover, the coefficient ${\bf Y}_1(t,\gamma)$ is continuous in $t\in\mathbb{R}$ for every $\gamma\in[0,1]$.
\end{cor}
\begin{proof} The RHP \ref{master} is equivalent to a singular integral equation, cf. \cite{IIKS}, which can be stated using a Fredholm operator of index zero. The above vanishing lemma then states that the kernel of this operator is trivial, i.e the operator itself is onto. Thus the singular integral equation, equivalently the RHP, is solvable. We refer the interested reader to \cite{Z} for more on this subject. Once solvability is established uniqueness follows from a standard Liouville argument: by RHP \ref{master}, the scalar function $\det{\bf Y}(z)$ is entire, approaching unity at infinity, i.e. for any solution of the RHP we have $\det{\bf Y}(z)\equiv 1$. Thus, given two solutions ${\bf Y}_1(z)$ and ${\bf Y}_2(z)$ to RHP \ref{master}, we consider ${\bf Q}(z)={\bf Y}_1(z){\bf Y}_2(z)^{-1}$ which is again entire. Since in addition ${\bf Q}(z)\rightarrow\mathbb{I}$ as $z\rightarrow\infty$, equality of ${\bf Y}_1$ and ${\bf Y}_2$ follows from another application of Liouville's theorem. Finally, continuity of ${\bf Y}_1(\cdot,\gamma)$ for fixed $\gamma\in[0,1]$ follows from continuity of the jump matrix, the fact that RHP \ref{master} is solvable for all $(t,\gamma)\in\mathbb{R}\times[0,1]$ and a standard small norm argument.
%
\end{proof}
\subsection{The ZS-system}
We are now prepared to take the first step in the proof of Theorem \ref{main1}, namely the derivation of a closed form expression for $\det(1-\gamma T\chi_t\upharpoonright_{L^2(\mathbb{R})})$. First we state a standard connection formula between the same determinant and RHP \ref{master}.
\begin{prop} For any fixed $t\in\mathbb{R},\gamma\in[0,1]$,
\begin{equation}\label{l:28}
	\frac{\partial}{\partial t}\ln\det(1-\gamma T\chi_t\upharpoonright_{L^2(\mathbb{R})})=-\im Y_1^{22}(t,\gamma),
\end{equation}
in terms of the matrix ${\bf Y}_1(t,\gamma)=[Y_1^{jk}(t,\gamma)]_{j,k=1}^2$ defined in condition (3) of RHP \ref{master}.
\end{prop}
\begin{proof} Through \eqref{det:imp} and the Jacobi variation formula,
\begin{eqnarray}
	\frac{\partial}{\partial t}\ln\det(1-\gamma T\chi_t\upharpoonright_{L^2(\mathbb{R})})&\stackrel{\eqref{det:imp}}{=}&\frac{\partial}{\partial t}\ln\det(1-G\upharpoonright_{L^2(\mathring{\Omega})})=-\tr_{L^2(\mathring{\Omega})}\left((1-G)^{-1}\frac{\partial G}{\partial t}\right)\nonumber\\
	&=&-\int_{\mathring{\Omega}}\int_{\mathring{\Omega}}(1-G)^{-1}(\lambda,\mu)\frac{\partial G}{\partial t}(\mu,\lambda)\,\d\mu\,\d\lambda.\label{l:29}
\end{eqnarray}
But from \eqref{l:21} we can compute explicitly the $t$-derivative,
\begin{equation*}
	\frac{\partial G}{\partial t}(\mu,\lambda)=\im f_2(\mu)g_2(\lambda),\ \ \ \ \ {\bf f}(\lambda)=\begin{bmatrix}f_1(\lambda)\\ f_2(\lambda)\end{bmatrix},\ \ {\bf g}(\mu)=\begin{bmatrix}g_1(\mu)\\ g_2(\mu)\end{bmatrix},
\end{equation*}
which gives back in \eqref{l:29},
\begin{equation}\label{l:30}
	\frac{\partial}{\partial t}\ln\det(1-\gamma T\chi_t\upharpoonright_{L^2(\mathbb{R})})=-\im\int_{\mathring{\Omega}}\big((1-G\upharpoonright_{L^2(\mathring{\Omega})})^{-1}f_2\big)(\lambda)g_2(\lambda)\,\d\lambda\stackrel{\eqref{l:24}}{=}-\im\int_{\mathring{\Omega}}F_2(\lambda)g_2(\lambda)\,\d\lambda.
\end{equation}
Now use \eqref{l:25} and take the limit $z\rightarrow\infty,z\notin\Omega$,
\begin{equation*}
	{\bf Y}(z)=\mathbb{I}+\frac{1}{z}\int_{\Omega}{\bf F}(\lambda){\bf g}^{\intercal}(\lambda)\,\d\lambda+\mathcal{O}\left(z^{-2}\right),
\end{equation*}
so that by comparison with RHP \ref{master}, condition (3),
\begin{equation}\label{l:31}
	{\bf Y}_1=\int_{\Omega}{\bf F}(\lambda){\bf g}^{\intercal}(\lambda)\,\d\lambda=\int_{\mathring{\Omega}}{\bf F}(\lambda){\bf g}^{\intercal}(\lambda)\,\d\lambda.
\end{equation}
Identity \eqref{l:31} used in the right-hand side of \eqref{l:30} implies \eqref{l:28} and completes our proof.
\end{proof}
Second, we derive the ZS-system \eqref{dNLS} from RHP \ref{master} as follows: Define (compare the proof of Lemma \ref{vanish})
\begin{equation}\label{l:32}
	{\bf N}(z):={\bf Y}(z)\begin{cases}\Bigl[\begin{smallmatrix} 1 & 0\smallskip\\ -\im\sqrt{\gamma}\e^{-\frac{1}{4}z^2+\im tz} & 1\end{smallmatrix}\Bigr],&z\in\Delta\smallskip\\ \mathbb{I},&z\notin\mathbb{C}\setminus\overline{\Delta}\end{cases},
\end{equation}
where $\Delta$ lies in between $\mathbb{R}$ and $\Gamma$. As noted before, ${\bf N}(z)$ solves the problem summarized below.
\begin{problem}\label{NLS} For any $t\in\mathbb{R},\gamma\in[0,1]$, the matrix-valued function ${\bf N}(z)={\bf N}(z;t,\gamma)\in\mathbb{C}^{2\times 2}$ has the following properties:
\begin{enumerate}
	\item ${\bf N}(z)$ is analytic for $z\in\mathbb{C}\setminus\mathbb{R}$ and has a continuous extension on the closed upper and lower half-planes. 
	\item The square integrable boundary values ${\bf N}_{\pm}(z)=\lim_{\epsilon\downarrow 0}{\bf N}(z\pm\im\epsilon),z\in\mathbb{R}$ obey the jump condition
	\begin{equation}\label{l:33}
		{\bf N}_+(z)={\bf N}_-(z)\begin{bmatrix}1-\gamma\e^{-\frac{1}{2}z^2} & -\im\sqrt{\gamma} \e^{-\frac{1}{4}z^2-\im tz}\\ -\im\sqrt{\gamma} \e^{-\frac{1}{4}z^2+\im tz} & 1\end{bmatrix},\ \ z\in\mathbb{R}.
	\end{equation}
	\item As $z\rightarrow\infty$, the leading order behavior of ${\bf N}(z)$ remains unchanged from condition (3) in RHP \ref{master}, 
	\begin{equation*}
		{\bf N}(z)=\mathbb{I}+{\bf Y}_1(t,\gamma)z^{-1}+\mathcal{O}\big(z^{-2}\big).
	\end{equation*}
\end{enumerate}
\end{problem}
A simple check between RHP \ref{NLS} and RHP \ref{master0} for ${\bf X}(z;x,\gamma)$ reveals their equality subject to the identifications
\begin{equation}\label{l:34}
	{\bf N}(z;t,\gamma)={\bf X}\left(z;\frac{t}{2},\gamma\right),\ \ z\in\mathbb{C}\setminus\mathbb{R},\ \ t\in\mathbb{R},\ \ \gamma\in[0,1],\ \ \ \ \ \ \textnormal{with}\ \ \ \ \ r(z)=r(z;\gamma)=-\im\sqrt{\gamma}\,\e^{-\frac{1}{4}z^2}.
\end{equation}
For this reason we now establish solvability of RHP \ref{master0} and thus, in turn, existence of $y(x;\gamma)$:
\begin{theo}\label{solv:2} The RHP \ref{master0} for ${\bf X}(z;x,\gamma)$ with $r(z;\gamma)=-\im\sqrt{\gamma}\,\e^{-\frac{1}{4}z^2}$ is uniquely solvable for every $(x,\gamma)\in\mathbb{R}\times[0,1]$.
\end{theo}
\begin{proof} Observing the similarities between \eqref{l:26} and \eqref{djump} one first derives a vanishing lemma in the style of Lemma \ref{vanish}, i.e. assumes ${\bf X}$ satisfies conditions (1) and (2) in RHP \ref{master0} but instead ${\bf X}(z)=\mathcal{O}\big(z^{-1}\big),z\rightarrow\infty$. Now define ${\bf H}(z)={\bf X}(z){\bf X}^{\dagger}(\bar{z}),z\in\mathbb{C}\setminus\mathbb{R}$ and conclude $\int_{\mathbb{R}}{\bf H}_+(z)\,\d z=0$ so that, analogous to \eqref{l:27},
\begin{equation*}
	0=2\int_{-\infty}^{\infty}\left\{{\bf X}_-(z)\begin{bmatrix}1-|r(z;\gamma)|^2 & 0\\ 0 & 1\end{bmatrix}{\bf X}_-^{\dagger}(z)\right\}\,\d z.
\end{equation*}
This equation allows us to deduce ${\bf X}_-(z)\equiv 0$ and thus also ${\bf X}_+(z)\equiv 0$. By Carlson's theorem we then find ${\bf X}(z)\equiv 0$ in the whole $z$-plane and the vanishing lemma is proven. After that the proof argument of Corollary \ref{solv:1} applies verbatim to ${\bf X}(z;x,\gamma)$ and Theorem \ref{solv:2} follows.
\end{proof}
Next a short remark about (obvious) symmetries in RHP \ref{master0}, see also \eqref{dNLSsymm}.
\begin{prop}\label{connect} Besides \eqref{l:34} we also have ${\bf Y}_1(t,\gamma)={\bf X}_1\left(\frac{t}{2},\gamma\right),(t,\gamma)\in\mathbb{R}\times[0,1]$ 
and for any $(x,\gamma)\in\mathbb{R}\times[0,1]$,
\begin{equation}\label{l:35}
	{\bf X}(z;x,\gamma)=\sigma_1\overline{{\bf X}(\bar{z};x,\gamma)}\sigma_1,\ \ \ z\in\mathbb{C}\setminus\mathbb{R};\ \ \ \ \ \ \ \sigma_1=\begin{bmatrix}0 & 1\\ 1 & 0\end{bmatrix},
\end{equation}
from which we learn that
\begin{equation}\label{l:36}
	X_i^{11}(x,\gamma)=\overline{X_i^{22}(x,\gamma)},\ \ \ \ \ \ X_i^{21}(x,\gamma)=\overline{X_i^{12}(x,\gamma)},\ \ i\in\mathbb{Z}_{\geq 1}.
\end{equation}
Furthermore, for any $(x,\gamma)\in\mathbb{R}\times[0,1]$,
\begin{equation}\label{sym:1}
	{\bf X}(z;x,\gamma)=\sigma_2{\bf X}(-z;x,\gamma)\sigma_2,\ \ \ z\in\mathbb{C}\setminus\mathbb{R};\ \ \ \ \ \ \ \ \ \sigma_2=\begin{bmatrix}0&-\im\\ \im & 0\end{bmatrix},
\end{equation}
which leads to
\begin{equation}\label{sym:2}
	X_i^{11}(x,\gamma)=(-1)^iX_i^{22}(x,\gamma),\ \ \ \ \ \ X_i^{21}(x,\gamma)=(-1)^{i+1}X_i^{12}(x,\gamma),\ \ i\in\mathbb{Z}_{\geq 1}.
\end{equation}
\end{prop}
\begin{proof} The connection between ${\bf Y}_1$ and ${\bf X}_1$ follows from \eqref{l:34}. For \eqref{l:35}, respectively \eqref{sym:1}, use unique solvability of RHP \ref{master0} as $\sigma_1\overline{{\bf X}(\bar{z};x,\gamma)}\sigma_1$, respectively $\sigma_2{\bf X}(-z;x,\gamma)\sigma_2$, solve the same problem.
\end{proof}
And finally, the following straightforward and standard steps, compare \eqref{ZS:0} and \eqref{ZS:1} above: Since
\begin{equation*}
	{\bf W}(z):={\bf X}(z)\e^{-\im xz\sigma_3},\ \ \ z\in\mathbb{C}\setminus\mathbb{R},
\end{equation*}
solves a RHP with an $x$-independent jump on $\mathbb{R}$, we know that $\frac{\partial{\bf W}}{\partial x}{\bf W}^{-1}$ is an entire function. In fact, using condition (3) in RHP \ref{master0} and Liouville's theorem, we find 
\begin{equation}\label{l:37}
	\frac{\partial{\bf W}}{\partial x}=\left\{-\im z\sigma_3+2\im\begin{bmatrix}0 & X_1^{12}\smallskip\\ -X_1^{21} & 0\end{bmatrix}\right\}{\bf W}.
\end{equation}
However, by definition,
\begin{equation*}
	y=y(x,t;\gamma):=2\im X_1^{12}(x,t,\gamma)\ \ \ \ \ \ \ \ \Rightarrow\ \ \ \ \ \bar{y}(x,t;\gamma)\stackrel{\eqref{l:36}}{=}-2\im X_1^{21}(x,t,\gamma).
\end{equation*}
Moreover, expanding $\frac{\partial{\bf W}}{\partial x}{\bf W}^{-1}$ up to $\mathcal{O}(z^{-2})$ as $z\rightarrow\infty$, we obtain from comparison with \eqref{l:37},
\begin{equation*}
	\frac{\partial{\bf X}_1}{\partial x}=-2\im\begin{bmatrix}-X_1^{12}X_1^{21} & X_2^{12}-X_1^{12}X_1^{22}\smallskip\\ -X_2^{21}+X_1^{21}X_1^{11} & X_1^{21}X_1^{12}\end{bmatrix},
\end{equation*}
so in the $(22)$-entry,
\begin{equation}\label{l:40}
	\frac{\partial X_1^{22}}{\partial x}=-\frac{\im}{2}|y|^2.
\end{equation}
We summarize by combining \eqref{l:28}, Proposition \ref{connect} and \eqref{l:40},
\begin{equation*}
	\frac{\partial^2}{\partial t^2}\ln\det(1-\gamma T\chi_t\upharpoonright_{L^2(\mathbb{R})})\stackrel{\eqref{l:28}}{=}-\im\frac{\partial}{\partial t}Y_1^{22}(t,\gamma)=-\frac{\im}{2}\frac{\partial}{\partial x'}X_1^{22}(x',\gamma)\Bigg|_{x'=\frac{t}{2}}\stackrel{\eqref{l:40}}{=}-\frac{1}{4}\left|y\left(\frac{t}{2};\gamma\right)\right|^2,
\end{equation*}
so that after integration
\begin{equation}\label{l:41}
	\ln\det(1-\gamma T\chi_t\upharpoonright_{L^2(\mathbb{R})})=-\frac{1}{4}\int_t^{\infty}(x-t)\left|y\left(\frac{x}{2};\gamma\right)\right|^2\d x+c_1(\gamma)t+c_2(\gamma)
\end{equation}
for some $t$-independent $c_i$. The fastest way to show $c_i=0$ follows from a comparison of the $t\rightarrow+\infty$ asymptotic expansion in the left- and right-hand side of \eqref{l:41}.
\subsection{Right tail asymptotics I}\label{right1} Begin with
\begin{equation*}
	\det(1-\gamma T\chi_t\upharpoonright_{L^2(\mathbb{R})})=\exp\left[-\sum_{k=1}^{\infty}\frac{\gamma^k}{k}\tr_{L^2(\mathbb{R})}(T\chi_t)^k\right],
\end{equation*}
and use that, as $t\rightarrow+\infty$,
\begin{equation}\label{l:42}
	\tr_{L^2(\mathbb{R})}(T\chi_t)=\int_t^{\infty}T(x,x)\,\d x=\frac{1}{2\pi}\sqrt{\frac{\pi}{2}}\int_t^{\infty}\textnormal{erfc}(\sqrt{2}\,x)\,\d x=\frac{1}{\sqrt{2\pi}}\frac{\textnormal{erfc}(\sqrt{2}\,t)}{8t}\left(1+\mathcal{O}\left(t^{-2}\right)\right),
\end{equation}
together with
\begin{equation*}
	k\in\mathbb{Z}_{\geq 2}:\ \ \ \ \ \left|\tr_{L^2(\mathbb{R})}(T\chi_t)^k\right|\leq C\big(\textnormal{erfc}(\sqrt{2}t)\big)^k,\ \ \ C>0,\ \ \ t\rightarrow+\infty.
\end{equation*}
In summary,
\begin{lem}\label{tail:1} As $t\rightarrow+\infty$,
\begin{equation*}
	\det(1-\gamma T\chi_t\upharpoonright_{L^2(\mathbb{R})})=1-\frac{\gamma}{\sqrt{2\pi}}\frac{\textnormal{erfc}(\sqrt{2}\,t)}{8t}\left(1+\mathcal{O}\left(t^{-2}\right)\right).
\end{equation*}
\end{lem}
On the other hand we will now derive the large $t$-expansion for the right-hand side of \eqref{l:41} from a Deift-Zhou nonlinear steepest descent analysis \cite{DZ} of RHP \ref{NLS}. The results will in turn verify \eqref{easyesti} and thus, after integration in \eqref{l:41} and comparison with Lemma \ref{tail:1}, show that $c_i=0$. The detailed steps of the nonlinear steepest descent analysis are standard: assume $t>0$ throughout and first rescale while simultaneously opening lenses,
\begin{equation}\label{l:43}
	{\bf T}(z;t,\gamma):={\bf N}(zt;t,\gamma)\begin{cases}\Bigl[\begin{smallmatrix}1 & 0\smallskip\\ \im\sqrt{\gamma}\,\e^{-t^2\theta_+(z)} & 1\end{smallmatrix}\Bigr],&\Im z\in(0,\delta)\smallskip\\ \Bigl[\begin{smallmatrix}1&-\im\sqrt{\gamma}\,\e^{-t^2\theta_-(z)}\smallskip\\ 0 & 1\end{smallmatrix}\Bigr],&\Im z\in(-\delta,0)\smallskip\\ \mathbb{I},&\textnormal{else}\end{cases}\ \ \ \ \ \ \textnormal{with}\ \delta>0\ \ \textnormal{fixed}.
\end{equation}
We abbreviate $\theta_{\pm}(z):=\frac{1}{4}(z^2\mp4\im z)$ and keeping the sign charts of $\theta_{\pm}(z)$ in mind, see Figure \ref{fig6} below, the function ${\bf T}(z)$ solves the following RHP.
\begin{center}
\begin{figure}[tbh]
\resizebox{0.451\textwidth}{!}{\includegraphics{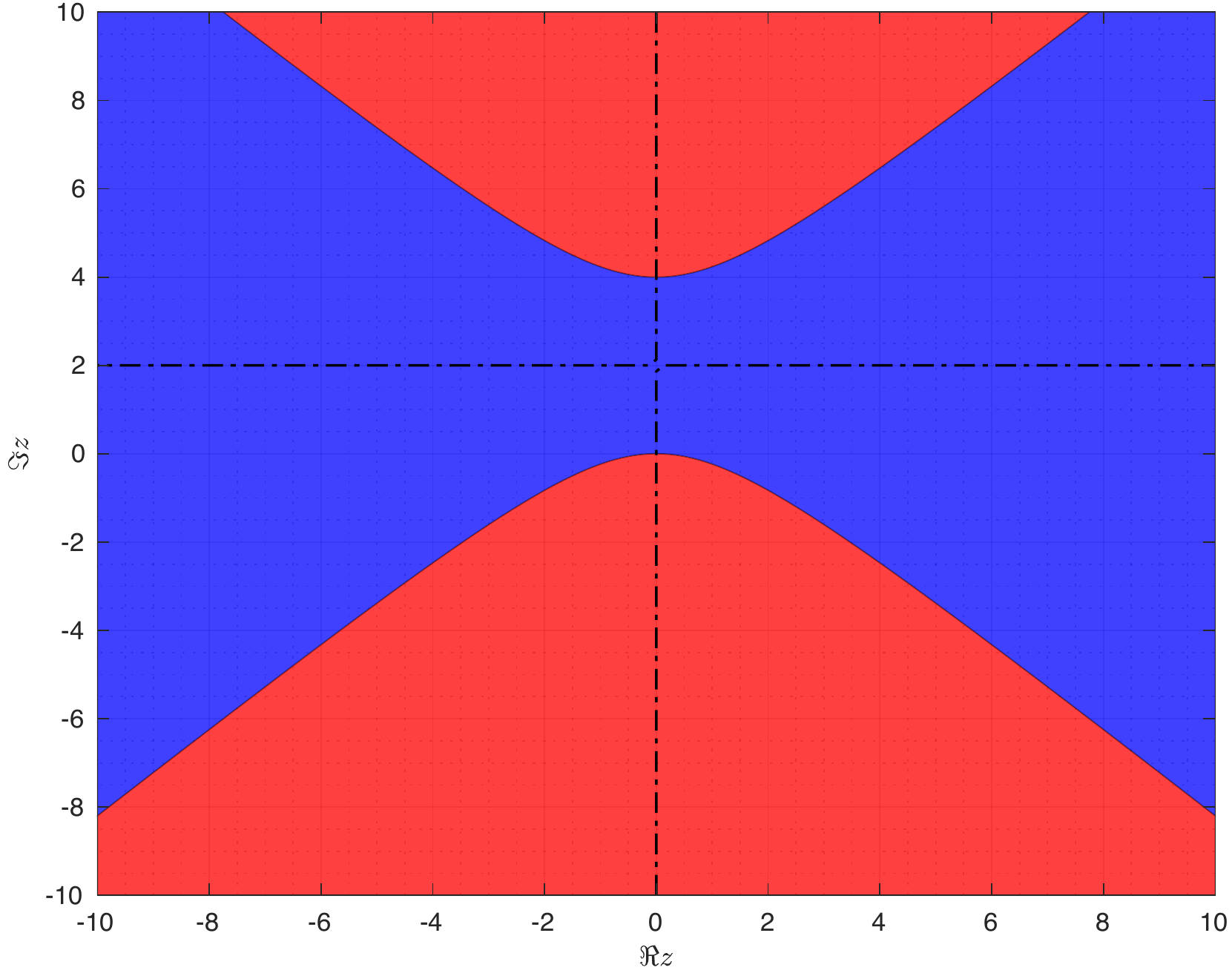}}\ \ \ \ \ \ \resizebox{0.451\textwidth}{!}{\includegraphics{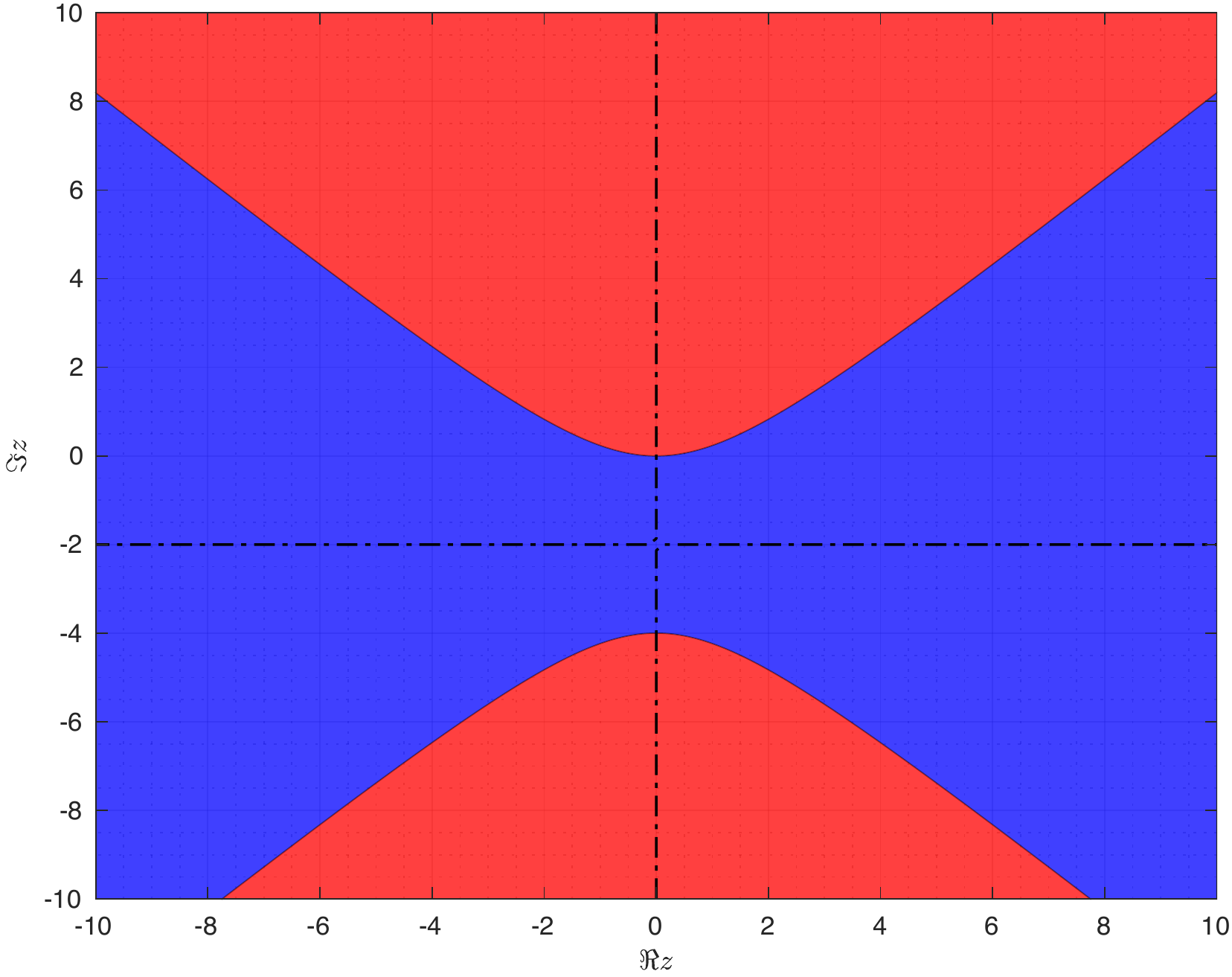}}
\caption{The sign charts of $\theta_+(z)$ on the left and $\theta_-(z)$ on the right. We indicate in red values $z\in\mathbb{C}$ where $\Re\theta_{\pm}(z)<0$ and in blue where $\Re\theta_{\pm}(z)>0$. Also $\Im\theta_{\pm}(z)=0$ is highlighted along the dash-dotted black straight lines.}
\label{fig6}
\end{figure}
\end{center}
\begin{problem}\label{easy} For any $t\in\mathbb{R}_{>0},\gamma\in[0,1]$ the transformed function ${\bf T}(z)={\bf T}(z;t,\gamma)\in\mathbb{C}^{2\times 2}$ defined in \eqref{l:43} has the following properties.
\begin{enumerate}
	\item[(1)] ${\bf T}(z)$ is analytic for $z\in\mathbb{C}\setminus\{\Im z=\pm\delta\}$ and both straight lines are oriented from left to right, see Figure \ref{fig7} below.
	\item[(2)] Along the two straight lines,
	\begin{equation*}
		{\bf T}_+(z)={\bf T}_-(z)\begin{bmatrix}1 & 0\\ -\im\sqrt{\gamma}\,\e^{-t^2\theta_+(z)} & 1\end{bmatrix},\ \ \ \Im z=\delta,
	\end{equation*}
	and
	\begin{equation*}
		{\bf T}_+(z)={\bf T}_-(z)\begin{bmatrix} 1 & -\im\sqrt{\gamma}\,\e^{-t^2\theta_-(z)} \\ 0 & 1\end{bmatrix},\ \ \ \Im z=-\delta.
	\end{equation*}
	\item[(3)] As $z\rightarrow\infty$, the leading order behavior in RHP \ref{NLS} remains formally unchanged.
	\begin{equation*}
		{\bf T}(z)=\mathbb{I}+{\bf T}_1(t,\gamma)z^{-1}+\mathcal{O}\big(z^{-2}\big),\ \ \ \ {\bf T}_1(t,\gamma)={\bf Y}_1(t,\gamma)t^{-1}.
	\end{equation*}
\end{enumerate}
\end{problem}
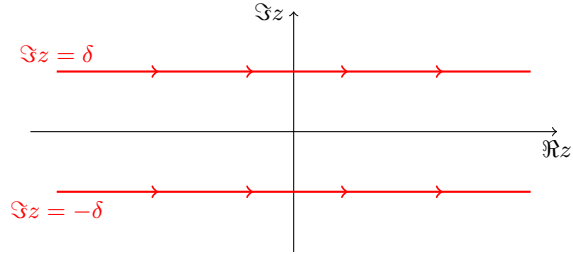
\begin{figure}[tbh]
\begin{tikzpicture}[xscale=0.7,yscale=0.4]
\draw [->] (-5,0) -- (5,0) node[below]{{\small $\Re z$}};
\draw [->] (0,-4) -- (0,4) node[left]{{\small $\Im z$}};
\draw [thick, color=red, decoration={markings, mark=at position 0.2 with {\arrow{<}}}, decoration={markings, mark=at position 0.4 with {\arrow{<}}}, decoration={markings, mark=at position 0.6 with {\arrow{<}}}, decoration={markings, mark=at position 0.8 with {\arrow{<}}}, postaction={decorate}] (4.5,2) -- (-4.5,2) node[above]{{\small $\Im z=\delta$}};
\draw [thick, color=red, decoration={markings, mark=at position 0.2 with {\arrow{<}}}, decoration={markings, mark=at position 0.4 with {\arrow{<}}}, decoration={markings, mark=at position 0.6 with {\arrow{<}}}, decoration={markings, mark=at position 0.8 with {\arrow{<}}}, postaction={decorate}] (4.5,-2) -- (-4.5,-2) node[below]{{\small $\Im z=-\delta$}};
\end{tikzpicture}
\caption{The oriented jump contour $\Im z=\pm\delta$ in RHP \ref{easy}.}
\label{fig7}
\end{figure}
The stationary points of the exponents $\theta_{\pm}(z)$ are $z_{\pm}=\pm 2\im$ and we denote by $\gamma_{\pm}$ the steepest descent contours passing through $z_{\pm}$ along which $\Im \theta_{\pm}(z)=0$. Explicitly,
\begin{equation*}
	\gamma_+:\ \ \ \Re z\in\mathbb{R},\ \Im z=2;\hspace{1.5cm}\gamma_-:\ \ \ \Re z\in\mathbb{R},\ \ \Im z=-2.
\end{equation*}
Note that both contours $\gamma_{\pm}$ lie in the domain where $\Re \theta_{\pm}(z)>0$, see Figure \ref{fig6}, so we are allowed to choose $\delta=2$ in RHP \ref{easy} to match those steepest descent contours. After that, applying standard arguments of the classical Laplace method, we derive the following estimates for the jump matrix ${\bf G}_{{\bf T}}(z;t,\gamma)$  in condition (2) of RHP \ref{easy}. 
\begin{prop}\label{DZ:es1} There exist positive constants $t_0$ and $c$ such that
\begin{equation*}
	\|{\bf G}_{\bf T}(\cdot;t,\gamma)-\mathbb{I}\|_{L^{\infty}(\Im z=\pm 2,|\d z|)}\leq c\sqrt{\gamma}\,\e^{-t^2},\ \ \ \ \|{\bf G}_{\bf T}(\cdot;t,\gamma)-\mathbb{I}\|_{L^2(\Im z=\pm 2,|\d z|)}\leq c\sqrt{\gamma}\,t^{-\frac{1}{2}}\e^{-t^2},
\end{equation*}
hold true for all $t\geq t_0$ and $\gamma\in[0,1]$.
\end{prop}
These estimates show by general theory \cite{DZ} that RHP \ref{easy} is uniquely solvable in $L^2(\Im z=\pm 2,|\d z|)$ for $t\geq t_0,\gamma\in[0,1]$ and its solution can be computed from the integral equation
\begin{equation}\label{l:44}
	{\bf T}(z)=\mathbb{I}+\frac{1}{2\pi\im}\int_{\Im\lambda=\pm 2}{\bf T}_-(\lambda)\big({\bf G}_{{\bf T}}(\lambda)-\mathbb{I}\big)\frac{\d\lambda}{\lambda-z},\ \ \ \ \Im z\neq\pm 2,
\end{equation}
using that
\begin{equation}\label{l:45}
	\|{\bf T}(\cdot;t,\gamma)-\mathbb{I}\|_{L^2(\Im z=\pm 2,|\d z|)}\leq c\sqrt{\gamma}\,t^{-\frac{1}{2}}\e^{-t^2},\ \ \ \ \ \forall\,t\geq t_0,\ \ \gamma\in[0,1].
\end{equation}
In particular, see condition (3) in RHP \ref{easy} and \eqref{l:44},
\begin{equation*}
	{\bf Y}_1(t,\gamma)=t\,{\bf T}_1(t,\gamma)=\frac{\im t}{2\pi}\int_{\Im\lambda=\pm 2}{\bf T}_-(\lambda)\big({\bf G}_{\bf T}(\lambda)-\mathbb{I}\big)\,\d\lambda,
\end{equation*}
so that with \eqref{l:45} and Proposition \ref{DZ:es1}, as $t\rightarrow+\infty$,
\begin{equation}\label{asres}
	{\bf Y}_1(t,\gamma)=\frac{\im t}{2\pi}\int_{\Im\lambda=\pm 2}\big({\bf G}_{\bf T}(\lambda)-\mathbb{I}\big)\,\d\lambda+\mathcal{O}\left(\gamma\e^{-2t^2}\right)=\frac{\sigma_1\sqrt{\gamma}}{\sqrt{\pi}}\e^{-t^2}+\mathcal{O}\left(\gamma\e^{-2t^2}\right).
\end{equation}
By Proposition \ref{connect} and the identity $y(x;\gamma)=2\im X_1^{12}(x,\gamma)$ we then find \eqref{easyesti}, namely
\begin{equation*}
	y(x;\gamma)=2\im X_1^{12}(x,\gamma)=2\im Y_1^{12}(2x,\gamma)=2\im\sqrt{\frac{\gamma}{\pi}}\,\e^{-4x^2}+\mathcal{O}\left(\gamma\e^{-8x^2}\right),\ \ \ x\rightarrow+\infty,
\end{equation*}
and with this back in the right-hand side of \eqref{l:41},
\begin{equation}\label{l:46}
	-\frac{1}{4}\int_t^{\infty}(x-t)\left|y\left(\frac{x}{2};\gamma\right)\right|^2\,\d x=\mathcal{O}\left(\gamma t^{-2}\e^{-2t^2}\right),\ \ t\rightarrow+\infty.
\end{equation}
Hence, comparing Lemma \ref{tail:1} to \eqref{l:46} in \eqref{l:41} we find $c_i=0$ and have therefore shown
\begin{equation}\label{part1fin}
	\det(1-\gamma T\chi_t\upharpoonright_{L^2(\mathbb{R})})=\exp\left[-\frac{1}{4}\int_t^{\infty}(x-t)\left|y\left(\frac{x}{2};\gamma\right)\right|^2\,\d x\right],
\end{equation}
where $y=y(x;\gamma):\mathbb{R}\rightarrow\im\mathbb{R}$ is related to the inverse scattering problem for \eqref{ZS:1} with the indicated reflection coefficient.\footnote{Proposition \ref{connect} together with $y(x;\gamma)=2\im X^{12}(x,\gamma)$ shows that $y(x;\gamma)$ is purely imaginary for $(x,\gamma)\in\mathbb{R}\times[0,1]$.}  This verifies the first part in \eqref{e:6}.
\section{Proof of Theorem \ref{main1}, part 2}\label{contin2}
Our goal in this section is to show that the integral
\begin{equation}\label{l:47}
	\Gamma_{t\gamma}=1-\gamma\int_t^{\infty}G(x)\big((1-\gamma T\chi_t\upharpoonright_{L^2(\mathbb{R})})^{-1}g\big)(x)\,\d x,\ \ \ t\in\mathbb{R};\ \ \ \ \ \ g(x)=\frac{1}{\sqrt{\pi}}\e^{-x^2},
\end{equation}
equals
\begin{equation}\label{l:48}
	\Gamma_{t\gamma}=\cosh\left[-\frac{\im}{2}\int_t^{\infty}y\left(\frac{x}{2},0;\gamma\right)\,\d x\right]-\sqrt{\gamma}\sinh\left[-\frac{\im}{2}\int_t^{\infty}y\left(\frac{x}{2},0;\gamma\right)\,\d x\right],
\end{equation}
which, in turn, will complete the proof of Theorem \ref{main1}, identity \eqref{e:6}, through \eqref{e:2}. We start by introducing
\begin{equation*}
	u^{\gamma}(t):=\gamma\int_t^{\infty}G(x)\big((1-\gamma T\chi_t\upharpoonright_{L^2(\mathbb{R})})^{-1}g\big)(x)\,\d x = \int_t^{\infty}Q^{\gamma}(x;t)\int_{-\infty}^xg^{\gamma}(v)\,\d v\,\d x,\ \ \ (t,\gamma)\in\mathbb{R}\times[0,1],
\end{equation*}
where 
\begin{equation*}
	Q^{\gamma}(x;t):=\big((1-\gamma T\chi_t\upharpoonright_{L^2(\mathbb{R})})^{-1}g^{\gamma}\big)(x),\ \ \ \ \ g^{\gamma}(x):=\sqrt{\gamma}\,g(x),\ \ \ x\in\mathbb{R}.
\end{equation*}
\begin{prop} For any $(t,\gamma)\in\mathbb{R}\times[0,1]$,
\begin{equation}\label{sweet:0}
	\frac{\d}{\d t}u^{\gamma}(t)=-Y_1^{12}(t,\gamma)\big((1-\gamma T\chi_t\upharpoonright_{L^2(\mathbb{R})})^{-1}G^{\gamma}\big)(t),\ \ \ \ G^{\gamma}(t):=\int_{-\infty}^tg^{\gamma}(y)\,\d y,
\end{equation}
in terms of ${\bf Y}_1$ from RHP \ref{master}, condition (3).
\end{prop}
\begin{proof}
Start by recalling \eqref{Fid:1}, the definition of the operator $N$, see Remark \ref{tc}, and Lemma \ref{clean}, identity \eqref{neat},
\begin{eqnarray}
	1-\gamma T\chi_t\upharpoonright_{L^2(\mathbb{R})}&=&\mathcal{F}_{\textnormal{ext}}(1-\gamma E_{\textnormal{ext}}\upharpoonright_{L^2(\Omega)})\mathcal{F}_{\textnormal{ext}}^{-1}=\mathcal{F}_{\textnormal{ext}}M(1-\gamma N\upharpoonright_{L^2(\mathring{\Omega})})M^{-1}\mathcal{F}_{\textnormal{ext}}^{-1}\nonumber\\
	&=&\mathcal{F}_{\textnormal{ext}}M(1-\sqrt{\gamma}N_1\upharpoonright_{L^2(\mathring{\Omega})})^{-1}(1-G\upharpoonright_{L^2(\mathring{\Omega})})(1-\sqrt{\gamma}N_2\upharpoonright_{L^2(\mathring{\Omega})})^{-1}M^{-1}\mathcal{F}_{\textnormal{ext}}^{-1}\label{l:49}.
\end{eqnarray}
With \eqref{l:47} we now compute
\begin{equation*}
	\big(M^{-1}\mathcal{F}_{\textnormal{ext}}^{-1}g^{\gamma}\big)(x)\stackrel{\eqref{l:1}}{=}\sqrt{\frac{\gamma}{2\pi}}\,\e^{-\frac{1}{8}x^2}\chi_{\mathbb{R}}(x)\stackrel{\eqref{l:21}}{=}f_1(x).
\end{equation*}
Since $f_1\in L^2(\mathbb{R})$ is supported on the real line only, it thus lies in the kernel of the operator $N_1:L^2(\mathring{\Omega},|d\lambda|)\rightarrow L^2(\mathring{\Omega},|\d\lambda|)$, compare \eqref{l:14}. Hence,
\begin{equation}\label{l:50}
	\big((1-G\upharpoonright_{L^2(\mathring{\Omega})})^{-1}(1-\sqrt{\gamma}N_1\upharpoonright_{L^2(\mathring{\Omega})})M^{-1}\mathcal{F}_{\textnormal{ext}}^{-1}g^{\gamma}\big)(x)=\big((1-G\upharpoonright_{L^2(\mathring{\Omega})})^{-1}f_1\big)(x)\stackrel{\eqref{l:24}}{=}F_1(x),
\end{equation}
in terms of $F_1(\lambda)=F_1(\lambda;t,\gamma)$ defined in \eqref{l:24}. But any function in the range of the operator $N_2$ will be supported on $\Gamma\neq\mathbb{R}$ only, compare again \eqref{l:14}. Hence we have $\textnormal{Ran}(N_2)\subseteq\textnormal{Ker}(\mathcal{F}_{\textnormal{ext}})$ by Convention \ref{id:1} and therefore from \eqref{l:49} and \eqref{l:50}, with $(x,t,\gamma)\in\mathbb{R}^2\times[0,1]$,
\begin{equation}\label{neeto}
	Q^{\gamma}(x;t)=\frac{1}{\sqrt{2\pi}}\int_{-\infty}^{\infty}\e^{-\im x\lambda}\e^{-\frac{1}{8}\lambda^2}F_1(\lambda)\,\d\lambda=\frac{1}{\sqrt{2\pi}}\int_{\mathring{\Omega}}\e^{-\im x\lambda}\e^{-\frac{1}{8}\lambda^2}\chi_{\mathbb{R}}(\lambda)F_1(\lambda)\,\d\lambda.
\end{equation}
As an important special case, we learn from the last equation that (compare \eqref{l:21} and \eqref{l:31})
\begin{equation}\label{sweet}
	Q^{\gamma}(t;t)=Y_1^{12}(t,\gamma),\ \ \ \ (t,\gamma)\in\mathbb{R}\times[0,1].
\end{equation}
At this point, $t$-differentiate $u^{\gamma}(t)$, using \eqref{sweet},
\begin{equation}\label{sweet:2}
	\frac{\d}{\d t}u^{\gamma}(t)=-Y_{12}(t,\gamma)\int_{-\infty}^tg^{\gamma}(v)\,\d v+\int_t^{\infty}\left[\frac{\partial}{\partial t}Q^{\gamma}(x;t)\right]\int_{-\infty}^xg^{\gamma}(v)\,\d v\,\d x
\end{equation}
and recall the following basic fact (cf. \cite[$(2.10)$]{TW3} or \cite[$(9.132)$]{F})
\begin{equation*}
	\frac{\partial}{\partial t}Q^{\gamma}(x;t)=-R^{\gamma}(x,t)Q^{\gamma}(t;t)=-R^{\gamma}(x,t)Y_1^{12}(t,\gamma),
\end{equation*}
where $R^{\gamma}(x,y)$ is the kernel of the resolvent integral operator $(1-\gamma T\chi_t\upharpoonright_{L^2(\mathbb{R})})^{-1}=1+R^{\gamma}\upharpoonright_{L^2(\mathbb{R})}$. Thus, back in \eqref{sweet:2},
\begin{align*}
	&\frac{\d}{\d t}u^{\gamma}(t)=-Y_{12}(t,\gamma)\int_{-\infty}^tg^{\gamma}(v)\,\d v-Y_1^{12}(t,\gamma)\int_t^{\infty}R^{\gamma}(x,t)\int_{-\infty}^xg^{\gamma}(v)\,\d v\,\d x\\
	=&-Y_1^{12}(t,\gamma)\int_t^{\infty}\big(1-\gamma T\chi_t\upharpoonright_{L^2(\mathbb{R})}\big)^{-1}(x,t)\int_{-\infty}^xg^{\gamma}(v)\,\d v\,\d x
	=-Y_1^{12}(t,\gamma)\big((1-\gamma T\chi_t\upharpoonright_{L^2(\mathbb{R})})^{-1}G^{\gamma}\big)(t),
\end{align*}
where we used self-adjointness of $T\chi_t$. Identity \eqref{sweet:0} is proven.
\end{proof}
\begin{prop} For any $(t,\gamma)\in\mathbb{R}\times[0,1]$,
\begin{equation}\label{sweet:3}
	\big((1-\gamma T\chi_t\upharpoonright_{L^2(\mathbb{R})})^{-1}G^{\gamma}\big)(t)=\int_{-\infty}^t\big((1-\gamma T\chi_t\upharpoonright_{L^2(\mathbb{R})})^{-1}g^{\gamma}\big)(y)\,\d y.
\end{equation}
\end{prop}
\begin{proof} Identity \eqref{sweet:3} follows from Proposition \ref{Jinho:1} and Corollary \ref{Jinho:2} in Appendix \ref{appC} with the operator identifications $K=\gamma T$, i.e. $\phi(x)=g^{\gamma}(x)=\psi(x)$, and the choice of interval $I=(-\infty,0)\subset\mathbb{R}$. In detail, for $\gamma\in[0,1)$ we use Corollary \eqref{J:2} with $x=t$ and the Neumann series representation (recall Lemma \ref{reg}) to deduce
\begin{equation}\label{sweet:10}
	\big((1-\gamma T\upharpoonright_{L^2(\mathbb{R})})^{-1}G^{\gamma}\big)(t)=\int_{-\infty}^0\big((1-\gamma T\chi_t\upharpoonright_{L^2(\mathbb{R})})\big)^{-1}g^{\gamma}\big)(y+t)\,\d y
\end{equation}
which is the right hand side in \eqref{sweet:3} after a shift. For $\gamma=1$, we take the limit $\gamma\uparrow 1$ in \eqref{sweet:10}.
\end{proof}
\begin{rem} Identity \eqref{sweet:3} is the $\gamma$-generalization of the corresponding equality in \cite[$(2.6),(2.8),(2.10)$]{PTZ}.
\end{rem}
The strategy is now to derive a coupled system of differential equations for the auxiliary function
\begin{equation}\label{l:51}
	A^{\gamma}(t):=\int_{-\infty}^t\big((1-\gamma T\chi_t\upharpoonright_{L^2(\mathbb{R})})^{-1}g^{\gamma}\big)(x)\,\d x
	\stackrel{\eqref{neeto}}{=}\frac{1}{\sqrt{2\pi}}\int_{-\infty}^t\left[\int_{\mathring{\Omega}}\e^{-\im x\lambda}\e^{-\frac{1}{8}\lambda^2}\chi_{\mathbb{R}}(\lambda)F_1(\lambda)\,\d\lambda\right]\,\d x
\end{equation}
and the (closely related) quantity
\begin{equation}\label{l:52}
	B^{\gamma}(t):=\frac{1}{\sqrt{2\pi}}\int_{-\infty}^t\left[\int_{\mathring{\Omega}}\e^{-\im x\lambda}\e^{-\frac{1}{8}\lambda^2}\chi_{\mathbb{R}}(\lambda)F_2(\lambda)\,\d\lambda\right]\,\d x,\ \ \ (t,\gamma)\in\mathbb{R}\times[0,1],
\end{equation}
where $F_2(\lambda)=F_2(\lambda;t,\gamma)$ is given in \eqref{l:24}. Imposing boundary conditions we then compute the unique solutions $(A^{\gamma}(t),B^{\gamma}(t))$ and obtain \eqref{l:48} through \eqref{sweet:3} and integration in \eqref{sweet:0}.
\begin{lem} The functions $A^{\gamma}(t)$ and $B^{\gamma}(t)$ defined in \eqref{l:51}, \eqref{l:52} for $(t,\gamma)\in\mathbb{R}\times[0,1]$ satisfy the differential equations
\begin{equation}\label{l:53}
	\frac{\d A^{\gamma}}{\d t}=\im Y_1^{12}(t,\gamma)B^{\gamma}+Y_1^{12}(t,\gamma),\ \ \ \ \ \frac{\d B^{\gamma}}{\d t}=-\im Y_1^{21}(t,\gamma)A^{\gamma},
\end{equation}
with boundary conditions $A^{\gamma}(t)\rightarrow \sqrt{\gamma},B^{\gamma}(t)\rightarrow 0$ as $t\rightarrow+\infty$. The quantity ${\bf Y}_1(t,\gamma)=[Y_1^{jk}(t,\gamma)]_{j,k=1}^2$ occurred first in RHP \ref{master}, condition $(3)$.
\end{lem}
\begin{proof} When $t$-differentiating \eqref{l:51} and \eqref{l:52} we require the partial derivatives $\frac{\partial F_j}{\partial t}(\lambda;t,\gamma)$ with $\lambda\in\mathbb{R}$. For these use \eqref{IIKSform} with, say, the $(-)$ limiting value in place, and \eqref{l:32},
\begin{equation}\label{l:54}
	{\bf F}(\lambda)={\bf N}_-(\lambda){\bf f}(\lambda),\ \ \ \lambda\in\mathbb{R}.
\end{equation}
But ${\bf N}(\lambda)$ obeys the rescaled Zakharov-Shabat system, compare \eqref{l:37},
\begin{equation*}
	\frac{\partial}{\partial t}\left({\bf N}(\lambda)\e^{-\frac{\im}{2}t\lambda\sigma_3}\right)=\left\{-\frac{\im}{2}\lambda\sigma_3+\im\begin{bmatrix}0 & Y_1^{12}\\ -Y_1^{21} & 0\end{bmatrix}\right\}{\bf N}(\lambda)\e^{-\frac{\im}{2}t\lambda\sigma_3},
\end{equation*}
so that in turn from \eqref{l:53},
\begin{equation*}
	\frac{\partial{\bf F}}{\partial t}(\lambda)=\left\{-\frac{\im}{2}\lambda\sigma_3+\im\begin{bmatrix}0 & Y_1^{12}\\ -Y_1^{21} & 0\end{bmatrix}\right\}{\bf F}(\lambda)+\frac{\im}{2}\lambda\,{\bf N}_-(\lambda)\sigma_3{\bf f}(\lambda)+{\bf N}_-(\lambda)\frac{\partial{\bf f}}{\partial t}(\lambda),\ \ \lambda\in\mathbb{R}.
\end{equation*}
But once we substitute \eqref{l:21} and \eqref{IIKSform} into this vector equation we find
\begin{equation}\label{l:55}
	\frac{\partial F_1}{\partial t}(\lambda)=\im Y_1^{12}F_2(\lambda),\ \ \ \ \ \ \frac{\partial F_2}{\partial t}(\lambda)=\im\lambda F_2(\lambda)-\im Y_1^{21} F_1(\lambda).
\end{equation}
Now $t$-differentiate $A^{\gamma}(t)$ first,
\begin{equation*}
	\frac{\d A^{\gamma}}{\d t}=\frac{1}{\sqrt{2\pi}}\int_{-\infty}^t\left[\int_{\mathring{\Omega}}\e^{-\im x\lambda}\e^{-\frac{1}{8}\lambda^2}\chi_{\mathbb{R}}(\lambda)\frac{\partial F_1}{\partial t}(\lambda)\,\d\lambda\right]\,\d x+\frac{1}{\sqrt{2\pi}}\int_{\mathring{\Omega}}\e^{-\im t\lambda}\e^{-\frac{1}{8}\lambda^2}\chi_{\mathbb{R}}(\lambda)F_1(\lambda)\,\d\lambda,
\end{equation*}
use \eqref{l:55}, \eqref{l:52} in the first integral and \eqref{l:21}, \eqref{l:31} in the second, i.e.
\begin{equation*}
	\frac{\d A^{\gamma}}{\d t}=\im Y_1^{12}B^{\gamma}+Y_1^{12}.
\end{equation*}
Similarly for $B^{\gamma}(t)$: differentiate
\begin{equation*}
	\frac{\d B^{\gamma}}{\d t}=\frac{1}{\sqrt{2\pi}}\int_{-\infty}^t\left[\int_{\mathring{\Omega}}\e^{-\im x\lambda}\e^{-\frac{1}{8}\lambda^2}\chi_{\mathbb{R}}(\lambda)\frac{\partial F_2}{\partial t}(\lambda)\,\d\lambda\right]\,\d x+\frac{1}{\sqrt{2\pi}}\int_{\mathring{\Omega}}\e^{-\im t\lambda}\e^{-\frac{1}{8}\lambda^2}\chi_{\mathbb{R}}(\lambda)F_2(\lambda)\,\d\lambda,
\end{equation*}
and use \eqref{l:55}, \eqref{l:51} in the first term and \eqref{l:21}, \eqref{l:31} in the second,
\begin{equation}\label{l:56}
	\frac{\d B^{\gamma}}{\d t}=-\frac{1}{\sqrt{2\pi}}\int_{-\infty}^t\frac{\partial}{\partial x}\left[\int_{\mathring{\Omega}}\e^{-\im x\lambda}\e^{-\frac{1}{8}\lambda^2}\chi_{\mathbb{R}}(\lambda)F_2(\lambda)\,\d\lambda\right]\,\d x-\im Y_1^{21}A^{\gamma}+Y_1^{22}.
\end{equation}
But from the Riemann-Lebesgue Lemma (using that $\e^{-\frac{1}{8}\lambda^2}\chi_{\mathbb{R}}(\lambda)F_2(\lambda)\in L^1(\mathbb{R})$ by Cauchy-Schwarz)
\begin{equation*}
	\lim_{|x|\rightarrow\infty}\int_{\mathring{\Omega}}\e^{-\im x\lambda}\e^{-\frac{1}{8}\lambda^2}\chi_{\mathbb{R}}(\lambda)F_2(\lambda)\,\d\lambda=0,
\end{equation*}
so we can simplify \eqref{l:56} further and obtain with \eqref{l:21}, \eqref{l:31} in the end
\begin{equation*}
	\frac{\d B^{\gamma}}{\d t}=-Y_1^{22}-\im Y_1^{21}A^{\gamma}+Y_1^{22}=-\im Y_1^{21}A^{\gamma}.
\end{equation*}
Since $Y_1^{12}=Y_1^{21}$ (compare Proposition \ref{connect}) is known, the system \eqref{l:53} fully determines $(A^{\gamma}(t),B^{\gamma}(t))$ once we impose boundary conditions. And for this we can use the asymptotic results derived in Proposition \ref{DZ:es1} and \eqref{l:45}. In detail we trace back our steps through \eqref{IIKSform}, \eqref{l:32} and \eqref{l:43},
\begin{equation*}
	z\in\mathbb{R}:\ \ F_1(z)=N_-^{11}(z)f_1(z)=T_-^{11}\left(\frac{z}{t}\right)f_1(z),
\end{equation*}
and use that from \eqref{l:44} and \eqref{l:45},
\begin{equation}\label{l:57}
	{\bf T}_-(z)=\mathbb{I}+\mathcal{O}\left(\frac{\e^{-t^2}}{1+|z|}\right),\ \ \ t\rightarrow+\infty,\ \ z\in\mathbb{R}.
\end{equation}
Thus,
\begin{equation*}
	A^{\gamma}(t)=\sqrt{\gamma}\,G(t)\left(1+\mathcal{O}\left(\e^{-t^2}\right)\right),\ \ \ t\rightarrow+\infty;\ \ \ G(t)=\frac{1}{\sqrt{\pi}}\int_{-\infty}^t\e^{-y^2}\,\d y,
\end{equation*}
which shows that $A^{\gamma}(t)\rightarrow \sqrt{\gamma}$ in the same limit. Quite similarly,
\begin{equation*}
	z\in\mathbb{R}:\ \ F_2(z)=N_-^{21}(z)f_1(z)=T_-^{21}\left(\frac{z}{t}\right)f_1(z),
\end{equation*}
so that with \eqref{l:57}, $B^{\gamma}(t)=\mathcal{O}\big(\sqrt{\gamma}\,\e^{-t^2}\big)=o(1)$ as $t\rightarrow+\infty$ and this completes our proof.
\end{proof}
By means of \eqref{asres} we now simply check that the unique solution to the system \eqref{l:53} with the imposed boundary conditions is given by
\begin{equation}\label{GAMForm}
	A^{\gamma}(t)=\sqrt{\gamma}\cosh\left[\int_t^{\infty}Y_1^{12}(x,\gamma)\,\d x\right]-\sinh\left[\int_t^{\infty}Y_1^{12}(x,\gamma)\,\d x\right],\ \ \ (t,\gamma)\in\mathbb{R}\times[0,1],
\end{equation}
and
\begin{equation*}
	B^{\gamma}(t)=\im\left(1-\cosh\left[\int_t^{\infty}Y_1^{12}(x,\gamma)\,\d x\right]+\sqrt{\gamma}\sinh\left[\int_t^{\infty}Y_1^{12}(x,\gamma)\,\d x\right]\right),\ \ \ (t,\gamma)\in\mathbb{R}\times[0,1].
\end{equation*}
Now we combine \eqref{sweet:0}, \eqref{sweet:3} and \eqref{l:51},
\begin{equation*}
	\frac{\d}{\d t}u^{\gamma}(t)=-Y_1^{12}(t,\gamma)A^{\gamma}(t),
\end{equation*}
integrate using \eqref{GAMForm} (with the normalization $u^{\gamma}(t)\rightarrow 0$ as $t\rightarrow\infty$),
\begin{equation}\label{GAMForm:2}
	u^{\gamma}(t)=1+\sqrt{\gamma}\sinh\left[\int_t^{\infty}Y_1^{12}(x,\gamma)\,\d x\right]-\cosh\left[\int_t^{\infty}Y_1^{12}(x,\gamma)\,\d x\right],
\end{equation}
and recall that $\Gamma_{t\gamma}=1-u^{\gamma}(t)$. This confirms \eqref{l:48} and thus, after combining $\Gamma_{t\gamma}$ with \eqref{part1fin} also Theorem \ref{main1}.
\subsection{Right tail asymptotics II}
Since we have already established Lemma \ref{tail:1}, we are now left with \eqref{l:48} and its large positive $t$-expansion. But since with \eqref{asres},
\begin{equation*}
	Y_1^{12}(t,\gamma)=\sqrt{\frac{\gamma}{\pi}}\,\e^{-t^2}+\mathcal{O}\left(\gamma\e^{-2t^2}\right),\ \ \ t\rightarrow+\infty,\ \ \gamma\in[0,1],
\end{equation*}
we obtain at once,
\begin{equation*}
	\int_t^{\infty}Y_1^{12}(x,\gamma)\,\d x=\frac{\sqrt{\gamma}}{2}\textnormal{erfc}(t)+\mathcal{O}\left(\gamma t^{-1}\e^{-2t^2}\right),
\end{equation*}
and thus in turn with \eqref{GAMForm:2} and the relation $\Gamma_{t\gamma}=1-u^{\gamma}(t)$,
\begin{lem}\label{tail:2} As $t\rightarrow+\infty$,
\begin{equation*}
	\Gamma_{t\gamma}=1-\frac{\gamma}{2}\textnormal{erfc}(t)+\mathcal{O}\left(\gamma^{\frac{3}{2}}t^{-1}\e^{-2t^2}\right).
\end{equation*}
\end{lem}
The combination of Lemma \ref{tail:1} and \ref{tail:2} with the known asymptotic behavior of $\textnormal{erfc}(z)$, cf. \cite[$7.12.1$]{NIST} proves \eqref{e:9} once substituted into \eqref{gdef}.
\section{Left tail asymptotics - Proof of Corollary \ref{main2}, expansion \eqref{e:10}}\label{left1}
Our goal in this section is to prove expansion \eqref{e:10} for values $\gamma\in(0,1)$\footnote{We shall discard the trivial case $\gamma=0$ for technical purposes.} that are either fixed or approach $\gamma=1$ at a controlled rate. These goals will be achieved by deriving the analogue of \eqref{asres} for $t\rightarrow-\infty$ through a nonlinear steepest descent analysis of RHP \ref{NLS} and subsequent integration in \eqref{e:6}. 
\subsection{Initial steps} We fix $\gamma\in(0,1)$ and first rescale similarly to \eqref{l:43},
\begin{equation}\label{l:58}
	{\bf T}(z;t,\gamma):={\bf N}(-zt;t,\gamma),\ \ \ \ \ z\in\mathbb{C}\setminus\mathbb{R},\ \ \ t<0.
\end{equation}
so that the jump condition for ${\bf T}(z;t,\gamma)$ reads as
\begin{equation}\label{l:59}
	{\bf T}_+(z)={\bf T}_-(z)\begin{bmatrix}1-\gamma\e^{-\frac{1}{2}t^2z^2} & -\im\sqrt{\gamma}\e^{-t^2\theta_+(z)}\\ -\im\sqrt{\gamma}\e^{-t^2\theta_-(z)} & 1\end{bmatrix},\ \ \ z\in\mathbb{R}.
\end{equation}
Thus, opposed to the $t\rightarrow+\infty$ analysis, the subscripts in $\theta_{\pm}(z)$ have flipped in the exponents, i.e. transformation \eqref{l:43} is no longer helpful in view of the sign chart Figure \ref{fig6}. Instead we employ a $g$-function transformation which will swap the diagonal entries in \eqref{l:59} and modify the off-diagonal ones accordingly. In detail, the upcoming $g$-function will be defined in terms of the Cauchy transform of
\begin{equation*}
	h(x;t,\gamma):=-\ln\left(1-\gamma\e^{-\frac{1}{2}t^2x^2}\right),\ \ \ x\in\mathbb{R},\ \ \gamma\in(0,1),\ \ t<0,
\end{equation*}
and for this reason we shall collect a few of its relevant analytical properties below.
\begin{prop}\label{regProp} For any $\gamma\in(0,1)$ and $t<0$, the function $h(\cdot;t,\gamma)$ exists in $L^p(\mathbb{R})$ for all $1\leq p<\infty$ and is real-analytic. Moreover, using the principal branch of the logarithm, i.e. $\ln:\mathbb{C}\setminus(-\infty,0]\rightarrow\mathbb{C}$ with $\ln z=\ln|z|+\im\textnormal{arg}\,z$ and $\textnormal{arg}\,z\in(-\pi,\pi]$, it extends analytically into the horizontal strip
\begin{equation*}
	z\in\Pi_{t\gamma}=\left\{z\in\mathbb{C}:\,\,|\Im z|<\frac{\sqrt{-2\ln\gamma}}{|t|}\,\right\},
\end{equation*}
and we have the total integral identity
\begin{equation}\label{intform}
	\int_{-\infty}^{\infty}h(x;t,\gamma)\,\d x=\frac{\sqrt{2\pi}}{|t|}\textnormal{Li}_{\frac{3}{2}}(\gamma),
\end{equation}
in terms of the polylogarithm $\textnormal{Li}_s(z)$, cf. \cite[$25.12.10$]{NIST}.
\end{prop}
\begin{proof} Integrability on $\mathbb{R}$ follows from the inequality $1-\gamma\e^{-\frac{1}{2}t^2x^2}\geq 1-\gamma>0$ and the estimate 
\begin{equation*}
	h(x;t,\gamma)=\mathcal{O}\big(\gamma\e^{-\frac{1}{2}t^2x^2}\big),\ \ \ x\rightarrow\pm\infty.
\end{equation*}
For analyticity we simply compute the pre-image of $(-\infty,0]\subset\mathbb{R}$ under the map $\mathbb{C}\ni z\mapsto 1-\gamma\e^{-\frac{1}{2}t^2z^2}$,
\begin{eqnarray*}
	\Re\left(1-\gamma\e^{-\frac{1}{2}t^2z^2}\right)&=&1-\gamma\e^{-\frac{1}{2}t^2((\Re z)^2-(\Im z)^2)}\cos(t^2\,\Re z\,\Im z),\\
	\Im\left(1-\gamma\e^{-\frac{1}{2}t^2z^2}\right)&=&\gamma\e^{-\frac{1}{2}t^2((\Re z)^2-(\Im z)^2)}\sin(t^2\,\Re z\,\Im z),
\end{eqnarray*}
i.e. $1-\gamma\e^{-\frac{1}{2}t^2z^2}$ is purely real along the family of curve $\Gamma_{t,n}:t^2\,\Re z\,\Im z=n\pi,n\in\mathbb{Z}$, so in particular real negative on the imaginary axis for $|\Im z|>\sqrt{-2\ln\gamma}/|t|$ and along the parts of the black colored hyperbolas $\Gamma_{t,n}$ shown in Figure \ref{fig8} that lie inside the red colored regions. Since $\Pi_{t\gamma}$ excludes those parts, analyticity of $h(z;t,\gamma)$ follows easily and the remaining integral \eqref{intform} is standard.
\end{proof}
\begin{center}
\begin{figure}[tbh]
\resizebox{0.451\textwidth}{!}{\includegraphics{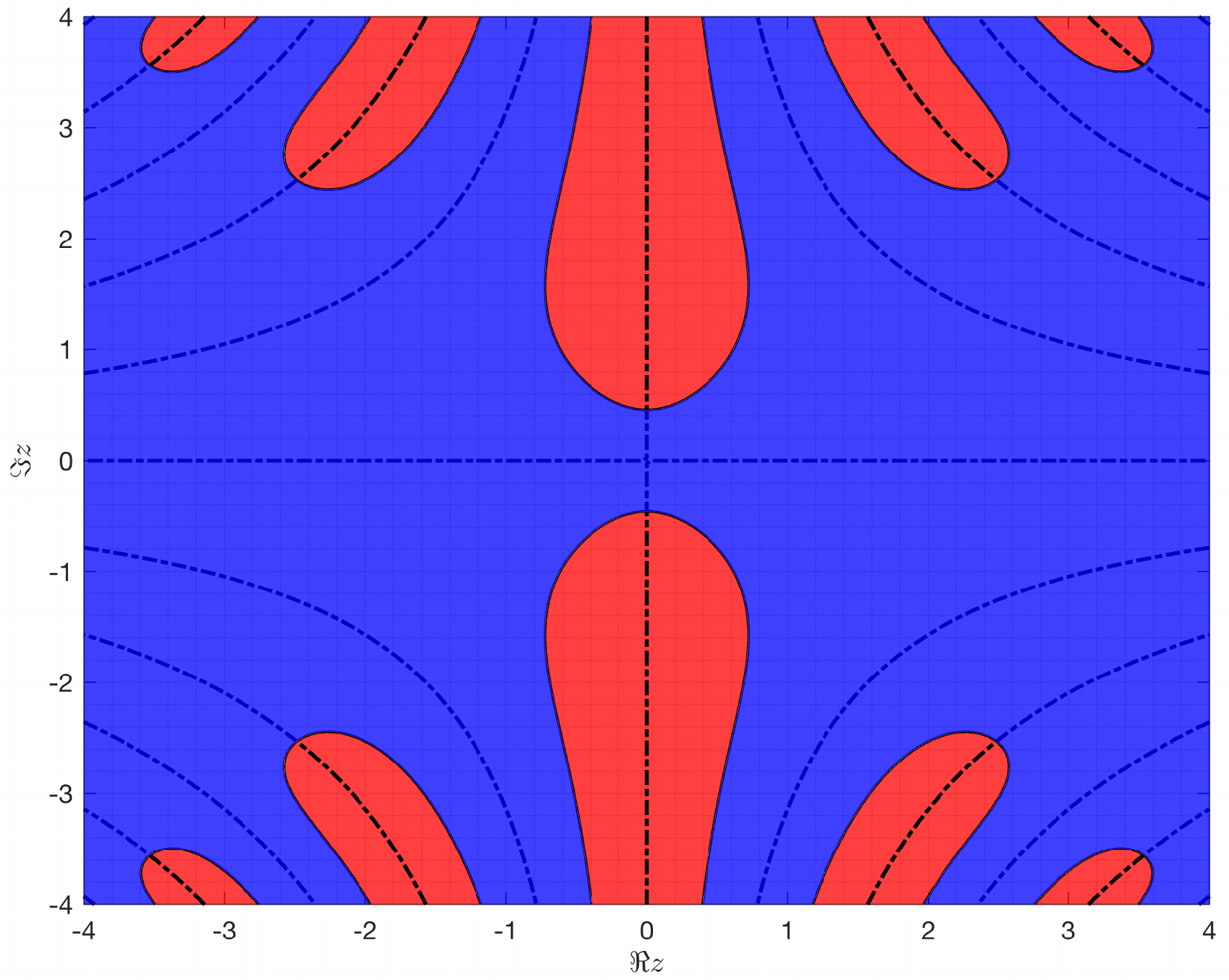}}\ \ \ \ \ \ \resizebox{0.451\textwidth}{!}{\includegraphics{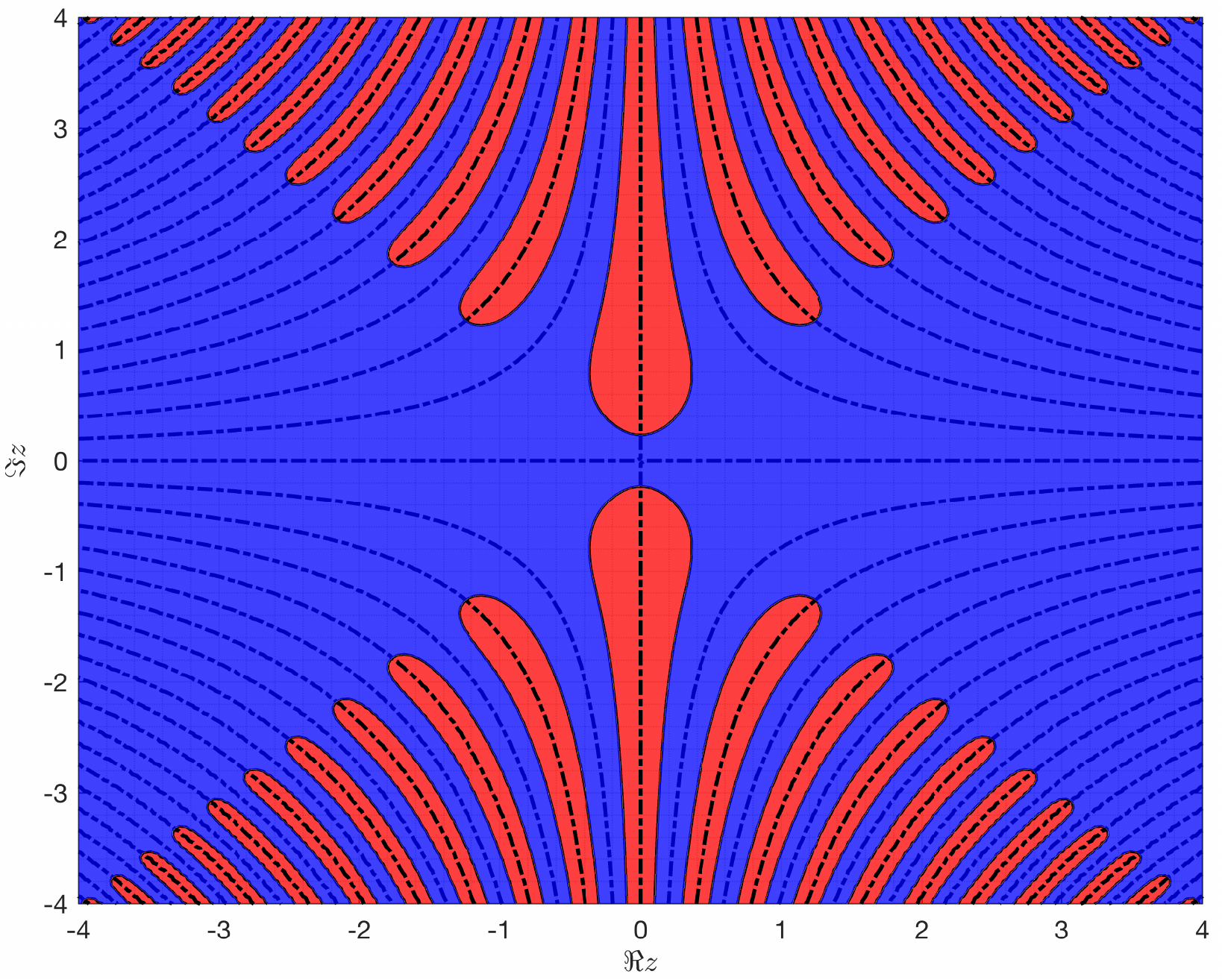}}
\caption{The sign charts of $s(z;t,\gamma):=1-\gamma\e^{-\frac{1}{2}t^2z^2}$ with $\gamma=0.9$ for $t=1$ on the left and $t=2$ on the right. In red values $z\in\mathbb{C}$ where $\Re s(z;t,\gamma)<0$ and in blue where $\Re s(z;t,\gamma)>0$. Also $\Im s(z;t,\gamma)=0$ along the dash-dotted black curves.}
\label{fig8}
\end{figure}
\end{center} 
In order to proceed we now introduce the $g$-function,
\begin{equation}\label{l:62}
	g(z)\equiv g(z;t,\gamma):=\frac{1}{2\pi\im}\int_{-\infty}^{\infty}\frac{h(s;t,\gamma)}{s-z}\,\d s,\ \ \ z\in\mathbb{C}\setminus\mathbb{R}
\end{equation}
and consider the following transformation after \eqref{l:58},
\begin{equation}\label{l:63}
	{\bf S}(z;t,\gamma):={\bf T}(z;t,\gamma)\e^{g(z)\sigma_3},\ \ \ z\in\mathbb{C}\setminus\mathbb{R}.
\end{equation}
Since $h(\cdot;t,\gamma)$ is of Hoelder class on $\mathbb{R}$ for all $\gamma\in(0,1)$ and $t<0$, the classical Plemelj-Sokhotski formula applies and we arrive at the RHP below.
\begin{problem}\label{gdeform} For any $t<0$ and $\gamma\in(0,1)$, the function ${\bf S}(z)={\bf S}(z;t,\gamma)\in\mathbb{C}^{2\times 2}$ defined in \eqref{l:63} has the following properties.
\begin{enumerate}
	\item[(1)] ${\bf S}(z)$ is analytic for $z\in\mathbb{C}\setminus\mathbb{R}$ and admits continuous boundary values on the closed upper and lower half-planes.
	\item[(2)] Along the real axis, with ${\bf S}_{\pm}(z)=\lim_{\epsilon\downarrow 0}{\bf S}(z\pm\im\epsilon),z\in\mathbb{R}$,
	\begin{equation}\label{l:64}
		{\bf S}_{+}(z)={\bf S}_{-}(z)\begin{bmatrix}1 & -\im\sqrt{\gamma}\e^{-t^2\phi_+(z)}\\ -\im\sqrt{\gamma}\e^{-t^2\phi_-(z)} & 1-\gamma\e^{-\frac{1}{2}t^2z^2}\end{bmatrix},\ \ \ \ z\in\mathbb{R},
	\end{equation}
	where 
	\begin{equation*}
		\phi_{\pm}(z)\equiv\phi_{\pm}(z;t,\gamma):=\theta_{\pm}(z)\pm\frac{1}{\im\pi}\,\textnormal{pv}\int_{-\infty}^{\infty}\frac{h(s;t,\gamma)}{t^2(s-z)}\,\d s,\ \ \ z\in\mathbb{R}.
	\end{equation*}
	\item[(3)] As $z\rightarrow\infty,z\notin\mathbb{R}$,
	\begin{equation*}
		{\bf S}(z)=\mathbb{I}+{\bf S}_{1}(t,\gamma)z^{-1}+\mathcal{O}\big(z^{-2}\big),\ \ \ \ \ {\bf S}_{1}(t,\gamma)={\bf T}_{1}(t,\gamma)-\frac{\im\sigma_3}{t\sqrt{2\pi}}\textnormal{Li}_{\frac{3}{2}}(\gamma)
	\end{equation*}
\end{enumerate}
\end{problem}
Our next step relies on the matrix factorization
\begin{equation}\label{l:65}
	\begin{bmatrix}1 & -\im\sqrt{\gamma}\e^{-t^2\phi_+(z)}\\ -\im\sqrt{\gamma}\e^{-t^2\phi_-(z)} & 1-\gamma\e^{-\frac{1}{2}t^2z^2}\end{bmatrix}=\begin{bmatrix}1 & 0\\ -\im\sqrt{\gamma}\e^{-t^2\phi_-(z)} & 1\end{bmatrix}\begin{bmatrix}1 & -\im\sqrt{\gamma}\e^{-t^2\phi_+(z)}\\ 0 & 1\end{bmatrix},\ \ z\in\mathbb{R},
\end{equation}
and the Lemma below.
\begin{lem}\label{short} The function
\begin{equation*}
	\phi(z)\equiv\phi(z;t,\gamma):=\frac{1}{4}\big(z^2-4\im z\,\textnormal{sgn}(\Im z)\big)+\frac{1}{\im\pi}\textnormal{sgn}(\Im z)\int_{-\infty}^{\infty}\frac{h(s;t,\gamma)}{t^2(s-z)}\,\d s-\frac{1}{t^2}h(z;t,\gamma),\ \ z\in\Pi_{t\gamma}\setminus\mathbb{R}
\end{equation*}
defined with the principal branch of the logarithm (see Proposition \ref{regProp}) is analytic in $\Pi_{t\gamma}\setminus\mathbb{R}$ for all $\gamma\in(0,1)$ and $t<0$. Moreover,
\begin{equation*}
	\lim_{\epsilon\downarrow 0}\phi(z\pm\im\epsilon)=\phi_{\pm}(z),\ \ \ z\in\mathbb{R}.
\end{equation*}
\end{lem}
\begin{proof} Since $h(\cdot;t,\gamma)$ is of Hoelder class on $\mathbb{R}$ and analytic in $\Pi_{t\gamma}$ by Proposition \ref{regProp}, the claims follow easily from the Plemelj-Sokhotski theorem.
\end{proof}
\noindent More precisely, we introduce for any $\gamma\in(0,1)$ and $t<0$,
\begin{equation}\label{l:666}
	{\bf M}(z;t,\gamma):={\bf S}(z;t,\gamma)\begin{cases}\Bigl[\begin{smallmatrix}1&\im\sqrt{\gamma}\,\e^{-t^2\phi(z)}\\ 0 & 1\end{smallmatrix}\Bigr],&z\in\Omega_1=\Omega_1(\delta_{t\gamma})\smallskip\\ \Bigl[\begin{smallmatrix}1 & 0\\ -\im\sqrt{\gamma}\,\e^{-t^2\phi(z)} & 1\end{smallmatrix}\Bigr],&z\in\Omega_2=\Omega_2(\delta_{t\gamma})\smallskip\\ \mathbb{I},&\textnormal{else}\end{cases}\ \ \  \textnormal{with}\ \ \ 0<\delta_{t\gamma}:=\min\left\{\frac{\sqrt{-\ln\gamma}}{|t|},\frac{1}{2}\right\}
\end{equation}
which leads to the RHP formulated on the red contour $\Sigma_{{\bf M}}$ in Figure \ref{Fig9} and with the following characteristica.
\begin{figure}[tbh]
\begin{tikzpicture}[xscale=0.7,yscale=0.4]
\draw [->] (-5,0) -- (5,0) node[below]{{\small $\Re z$}};
\draw [->] (0,-4) -- (0,4) node[left]{{\small $\Im z$}};
\draw [thick, color=red, decoration={markings, mark=at position 0.2 with {\arrow{>}}}, decoration={markings, mark=at position 0.8 with {\arrow{>}}}, postaction={decorate}] (-4.5,0) -- (-2,0);
\draw [thick, color=red, decoration={markings, mark=at position 0.5 with {\arrow{>}}}, postaction={decorate}] (-2,0) -- (-1,2);
\draw [thick, color=red, decoration={markings, mark=at position 0.5 with {\arrow{>}}}, postaction={decorate}] (-2,0) -- (-1,-2);
\draw [thick, color=red, decoration={markings, mark=at position 0.2 with {\arrow{>}}}, decoration={markings, mark=at position 0.8 with {\arrow{>}}},  postaction={decorate}] (-1,2) -- (1,2); 
\node [right, color=red] at (1,2) {{\small $\Im z=\delta_{t\gamma}$}};
\node [left, color=red] at (-1,-2) {{\small $\Im z=-\delta_{t\gamma}$}};
\node [right] at (-1,1) {{\small $\Omega_1$}};
\node [right] at (-1,-1) {{\small $\Omega_2$}};
\draw [thick, color=red, decoration={markings, mark=at position 0.2 with {\arrow{>}}}, decoration={markings, mark=at position 0.8 with {\arrow{>}}},  postaction={decorate}] (-1,-2) -- (1,-2); 
\draw [thick, color=red, decoration={markings, mark=at position 0.5 with {\arrow{>}}}, postaction={decorate}] (1,2) -- (2,0);
\draw [thick, color=red, decoration={markings, mark=at position 0.5 with {\arrow{>}}}, postaction={decorate}] (1,-2) -- (2,0);
\draw [thick, color=red, decoration={markings, mark=at position 0.2 with {\arrow{>}}}, decoration={markings, mark=at position 0.8 with {\arrow{>}}}, postaction={decorate}] (2,0) -- (4.5,0);
\end{tikzpicture}
\caption{The oriented jump contour $\Sigma_{{\bf M}}$ in RHP \ref{split}. We fix $z_1=-2,z_{2,3}=-1\pm\im\delta_{t\gamma},z_{3,4}=1\pm\im\delta_{t\gamma}$ and $z_5=2$ as location of the six vertices.}
\label{Fig9}
\end{figure}
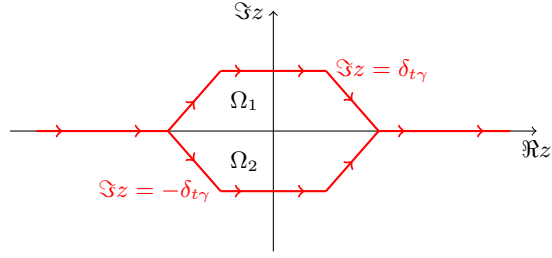
\begin{problem}\label{split} For any $\gamma\in(0,1)$ and $t<0$, the function ${\bf M}(z)={\bf M}(z;t,\gamma)\in\mathbb{C}^{2\times 2}$ has the following properties.
\begin{enumerate}
	\item[(1)] ${\bf M}(z)$ is analytic for $z\in\mathbb{C}\setminus\Sigma_{{\bf M}}$ and all eight straight lines are oriented from left to right, see Figure \ref{Fig9}.
	\item[(2)] The square integrable limiting values obey
	\begin{equation*}
		{\bf M}_{+}(z)={\bf M}_{-}(z)\begin{bmatrix}1 & -\im\sqrt{\gamma}\e^{-t^2\phi(z)}\\ 0 & 1\end{bmatrix},\ \ z\in\Sigma_{{\bf M}}\cap\{z\in\mathbb{C}:\,\Im z>0\};
	\end{equation*}
	and
	\begin{equation*}
		{\bf M}_{+}(z)={\bf M}_{-}(z)\begin{bmatrix}1 & 0\\ -\im\sqrt{\gamma}\e^{-t^2\phi(z)} & 1\end{bmatrix},\ \ z\in\Sigma_{{\bf M}}\cap\{z\in\mathbb{C}:\,\Im z<0\};
	\end{equation*}
	as well as
	\begin{equation*}
		{\bf M}_{+}(z)={\bf M}_{-}(z)\begin{bmatrix}1 & -\im\sqrt{\gamma}\e^{-t^2\phi_+(z)}\\ -\im\sqrt{\gamma}\e^{-t^2\phi_-(z)} & 1-\gamma\e^{-\frac{1}{2}t^2z^2}\end{bmatrix},\ \ z\in\mathbb{R}\setminus[-2,2].
	\end{equation*}
	By construction, see Lemma \ref{short} and equation \eqref{l:65}, there is no jump along $(-2,2)\subset\mathbb{R}$.
	\item[(3)] For large $|z|$ the leading order asymptotic behavior of RHP \ref{gdeform} is unchanged,
	\begin{equation*}
		{\bf M}(z)=\mathbb{I}+{\bf S}_{1}(t,\gamma)z^{-1}+\mathcal{O}\big(z^{-2}\big),\ \ \ z\rightarrow\infty.
	\end{equation*}
\end{enumerate}
\end{problem}
We now proceed with the necessary small norm estimates for the jump matrix ${\bf G}_{{\bf M}}(z;t,\gamma)$ in condition (2) of RHP \ref{split}. And since we are about to investigate the limit $\gamma\uparrow 1$ in the upcoming sections we shall already now indicate the $\gamma$-dependency in all error estimates.
\begin{prop}\label{smallnorm} There exists positive universal constants $t_0,c_j$ such that for any  $\gamma\in(0,1)$,
\begin{equation}\label{DZ:es}
	\|{\bf G}_{{\bf M}}(\cdot;t,\gamma)-\mathbb{I}\|_{L^{\infty}(\Sigma_{{\bf M}},|\d z|)}\leq c_1 \frac{\gamma^{\frac{1}{4}}\e^{-t^2\delta_{t\gamma}}}{1-\gamma},\ \ \ \ \ \ \  
	\|{\bf G}_{{\bf M}}(\cdot;t,\gamma)-\mathbb{I}\|_{L^2(\Sigma_{{\bf M}},|\d z|)}\leq c_2\frac{\gamma^{\frac{1}{4}}\e^{-t^2\delta_{t\gamma}}}{1-\gamma}\,|t|^{-\frac{1}{2}},
\end{equation}
hold true for all $(-t)\geq t_0$.
\end{prop}
\begin{proof} For $z\in\mathbb{R}\setminus[-2,2]$ we clearly have $\|{\bf G}_{{\bf M}}(\cdot;t,\gamma)-\mathbb{I}\|=\mathcal{O}(\e^{-t^2})$ in $L^2\cap L^{\infty}$ and thus no contribution to the leading order in \eqref{DZ:es}. Next, for $z\in\mathbb{C}\setminus\mathbb{R}$,
\begin{equation*}
	\e^{-t^2\phi(z)}=\frac{\e^{-t^2\theta_{\pm}(z)}}{1-\gamma\e^{-\frac{1}{2}t^2z^2}}\,\exp\left[\pm\frac{\im}{\pi}\int_{-\infty}^{\infty}\frac{h(s;t,\gamma)}{s-z}\,\d s\right],
\end{equation*}
so we have to estimate three factors. With the parametrizations $z=z(\lambda)=\lambda\pm\im\delta_{t\gamma},\lambda\in[-1,1]$ we find
\begin{equation*}\label{l:66}
	\left|\e^{-t^2\theta_{\pm}(z(\lambda))}\right|\leq\e^{-t^2\delta_{t\gamma}-\frac{1}{4}t^2\lambda^2}\gamma^{-\frac{1}{4}},\ \ 
	\left|1-\gamma\e^{-\frac{1}{2}t^2z^2(\lambda)}\right|=\left[1-2\gamma\e^{-\frac{1}{2}t^2(\lambda^2-\delta_{t\gamma})}\cos(t^2\lambda\delta_{t\gamma})+\gamma^2\e^{-t^2(\lambda^2-\delta_{t\gamma}^2)}\right]^{\frac{1}{2}}
\end{equation*}
and 
\begin{equation}\label{l:68}
	\left|\exp\left[\pm\frac{\im}{\pi}\int_{-\infty}^{\infty}\frac{h(s;t,\gamma)}{s-z(\lambda)}\,\d s\right]\right|=\exp\left[-\frac{\delta_{t\gamma}}{\pi}\int_{-\infty}^{\infty}\frac{h(s;t,\gamma)}{(s-\lambda)^2+\delta_{t\gamma}^2}\,\d s\right].
\end{equation}
Hence,
\begin{equation*}
	\sup_{\lambda\in[-1,1]}\left|\frac{\e^{-t^2\theta_{\pm}(z(\lambda))}}{1-\gamma\e^{-\frac{1}{2}t^2z^2(\lambda)}}\right|\leq\frac{\gamma^{-\frac{1}{4}}\e^{-t^2\delta_{t\gamma}}}{1-\sqrt{\gamma}},
\end{equation*}
and since the integrand in \eqref{l:68} is non-negative we have therefore established the above $L^{\infty}$ estimate. For the $L^2$ estimate we use Laplace's method for the factor $\exp[-t^2\theta_{\pm}(z)]$ and the same reasonings that were used in the second and third while deriving the previous $L^{\infty}$ estimate. For the four slanted straight lines consider, say, $z=z(\lambda)=\lambda-2\pm\im\lambda\delta_{t\gamma}$ with $\lambda\in[0,1]$. Then,
\begin{equation*}
	\left|\e^{-t^2\theta_{\pm}(z(\lambda))}\right|=\e^{-\frac{1}{4}t^2(\lambda^2(1-\delta_{t\gamma}^2)-4\lambda(1-\delta_{t\gamma})+4)},
\end{equation*}
as well as
\begin{equation*}
	\left|1-\gamma\e^{-\frac{1}{2}t^2z^2(\lambda)}\right|=\left[1-2\gamma\e^{-\frac{1}{2}t^2((\lambda-2)^2-\lambda^2\delta_{t\gamma}^2)}\cos\big(t^2\lambda(\lambda-2)\delta_{t\gamma}\big)+\gamma^2\e^{-t^2((\lambda-2)^2-\lambda^2\delta_{t\gamma}^2)}\right]^{\frac{1}{2}},
\end{equation*}
and
\begin{equation}\label{l:69}
	\left|\exp\left[\pm\frac{\im}{\pi}\int_{-\infty}^{\infty}\frac{h(s;t,\gamma)}{s-z(\lambda)}\,\d s\right]\right|=\exp\left[-\frac{\lambda\delta_{t\gamma}}{\pi}\int_{-\infty}^{\infty}\frac{h(s;t,\gamma)}{(s-\lambda+2)^2+\lambda^2\delta_{t\gamma}^2}\,\d s\right].
\end{equation}
Thus
\begin{equation*}
	\sup_{\lambda\in[0,1]}\left|\frac{\e^{-t^2\theta_{\pm}(z(\lambda))}}{1-\gamma\e^{-\frac{1}{2}t^2z^2(\lambda)}}\right|\leq\frac{\e^{-\frac{1}{4}t^2(1-\delta_{t\gamma}^2)}}{1-\gamma\e^{-\frac{1}{2}t^2(1-\delta_{t\gamma}^2)}},
\end{equation*}
which is of sub-leading order for $\gamma\in(0,1)$ when compared to \eqref{DZ:es}, see \eqref{l:65}. By non-negativity of the integrand in \eqref{l:69} we have therefore completed our proof.
\end{proof}
Since we are dealing with a $(t,\gamma)$-dependent contour in RHP \ref{split} (the hexagon in Figure \ref{Fig9} is collapsing to the real axis as $t\rightarrow-\infty$ or $\gamma\uparrow 1$) the general framework of \cite{DZ} is not directly applicable to Proposition \ref{smallnorm}. Still, using somewhat similar ideas as in \cite{BK}, the results of Appendix \ref{appA} below guarantee unique solvability of RHP \ref{split} in $L^2(\Sigma_{{\bf M}},|\d z|)$ for all $(-t)\geq t_0$ and either any fixed $\gamma\in(0,1)$ or $\gamma\uparrow 1$ at a certain controlled rate. 
\begin{theo}\label{collapse} For any fixed $0<\epsilon<2$, RHP \ref{split} is uniquely solvable in $L^2(\Sigma_{{\bf M}},|\d z|)$ for $(-t)$ sufficiently large and all $0<\gamma\leq 1-|t|^{-\epsilon}$. Moreover, 
\begin{equation}\label{l:70}
	{\bf S}_{1}(t,\gamma)=\lim_{z\rightarrow\infty}z\big({\bf M}(z;t,\gamma)-\mathbb{I}\big)=\mathcal{O}\left(|t|^{-1+\epsilon}\e^{-|t|^{1-\frac{\epsilon}{2}}}\right),\ \ \ \forall\ (-t)\geq t_0,\ \ 0<\gamma\leq 1-|t|^{-\epsilon}.
	%
\end{equation}
\end{theo}
\subsection{Left tail asymptotics} In order to complete the derivation of \eqref{e:10}, we first recall \eqref{l:28}, transformations \eqref{l:32}, \eqref{l:58}, \eqref{l:63}, \eqref{l:65}, 
\begin{equation*}
	\frac{\partial}{\partial t}\ln\det(1-\gamma T\chi_t\upharpoonright_{L^2(\mathbb{R})})=\frac{1}{\sqrt{2\pi}}\textnormal{Li}_{\frac{3}{2}}(\gamma)+\im t S_1^{22}(t,\gamma).
\end{equation*}
Thus, with Theorem \ref{collapse} and an indefinite $t$-integration,
\begin{prop}\label{intf:1} For any fixed $0<\epsilon<2$, there exist positive constants $t_0=t_0(\epsilon)$ and $c=c(\epsilon)$ such that
\begin{equation*}
	\ln\det(1-\gamma T\chi_t\upharpoonright_{L^2(\mathbb{R})})=\frac{t}{\sqrt{2\pi}}\textnormal{Li}_{\frac{3}{2}}(\gamma)+D_1(\gamma)+r_1(t,\gamma),
\end{equation*}
for all $(-t)\geq t_0$ and $0<\gamma\leq 1-|t|^{-\epsilon}$. The term $D_1(\gamma)$ is $t$-independent and we record the error estimate
\begin{equation*}
	\big|r_1(t,\gamma)\big|\leq c(\epsilon)|t|^{\frac{3}{2}\epsilon}\e^{-|t|^{1-\frac{\epsilon}{2}}}\ \ \ \forall\,(-t)\geq t_0\ \ \textnormal{and}\ \ 0<\gamma\leq 1-|t|^{-\epsilon}.
\end{equation*}
\end{prop}
Next we use the estimate (based again on the transformations \eqref{l:32}, \eqref{l:58}, \eqref{l:63}, \eqref{l:65} and Theorem \ref{collapse})
\begin{equation*}
	\int_t^{\infty}Y_1^{12}(x,\gamma)\,\d x=I(\gamma)-\int_{-\infty}^tY_1^{12}(x,\gamma)\,\d x=I(\gamma)+\int_{-\infty}^tx S_1^{12}(x,\gamma)\,\d x=I(\gamma)+\mathcal{O}\left(|t|^{\frac{3}{2}\epsilon}\e^{-|t|^{1-\frac{\epsilon}{2}}}\right),
\end{equation*}
where $I(\gamma)$ denotes the total integral of $Y_1^{12}(x,\gamma)$ over $x\in\mathbb{R}$. From this, with \eqref{GAMForm:2} and $\Gamma_{t\gamma}=1-u^{\gamma}(t)$ we obtain in turn
\begin{prop}\label{intf:2} For any fixed $0<\epsilon<2$, there exist positive constants $t_0=t_0(\epsilon)$ and $c=c(\epsilon)$ such that
\begin{equation*}
	\ln\Gamma_{t\gamma}=D_2(\gamma)+r_2(t,\gamma)
\end{equation*}
for all $(-t)\geq t_0$ and $0<\gamma\leq 1-|t|^{-\epsilon}$. The term $D_2(\gamma)$ is $t$-independent and we record the error estimate
\begin{equation*}
	\big|r_2(t,\gamma)\big|\leq c(\epsilon)|t|^{\frac{3}{2}\epsilon}\e^{-|t|^{1-\frac{\epsilon}{2}}}\ \ \ \forall\,(-t)\geq t_0\ \ \textnormal{and}\ \ 0<\gamma\leq 1-|t|^{-\epsilon}.
\end{equation*}
\end{prop}
Once we combine Propositions \ref{intf:1} and \ref{intf:2}, Corollary \ref{main2}, expansion \eqref{e:10} follows easily.

\section{Proof of Theorem \ref{main3}}\label{Spohn}
The following lines are near copies of the argument given in \cite{FS} in the derivation of the analogue of \eqref{e:11} for the GOE. First, by \eqref{e:1} and \eqref{e:3},
\begin{equation}\label{FS:1}
	\big(F(t)\big)^2=\det(1-T\chi_t\upharpoonright_{L^2(\mathbb{R})})\Gamma_{t}=\det\big(1-( T+U\otimes V)\upharpoonright_{L^2(t,\infty)}\big),
\end{equation}
where $U\otimes V$ denotes the finite rank integral operator on $L^2(t,\infty)$ with kernel
\begin{equation*}
	(U\otimes V)(x,y)=g(x)G(y),\ \ \ \ \ g(x)=\frac{1}{\sqrt{\pi}}\,\e^{-x^2},\ \ \ G(x)=\int_{-\infty}^xg(y)\,\d y,
\end{equation*}
i.e. $U$ is the operator which multiplies by $g(x)$ and $V$ the integral operator with kernel $G(y)$. Indeed, \eqref{FS:1} follows by noting that for any operator $A$ we have $A(U\otimes V)=(AU\otimes V)$, applying the factorization
\begin{equation}\label{FS:2}
	\det\big(1-(T+U\otimes V)\upharpoonright_{L^2(t,\infty)}\big)=\det(1-T\upharpoonright_{L^2(t,\infty)})\det\big(1-(1-T)^{-1}(U\otimes V)\upharpoonright_{L^2(t,\infty)}\big),
\end{equation}
and using the eigenvector/value equation (see \cite[Chapter $9.7$]{F} for a similar argument in the GOE)
\begin{equation}\label{FS:3}
	\big((1-T)^{-1}(U\otimes V)\upharpoonright_{L^2(t,\infty)}F\big)(x)=\left[\int_t^{\infty}G(y)\big((1- T\upharpoonright_{L^2(t,\infty)})^{-1}g\big)(y)\,\d y\right]F(x),
\end{equation}
where
\begin{equation*}
	F(x)=\big((1-T\upharpoonright_{L^2(t,\infty)})^{-1}g\big)(x).
\end{equation*}
Precisely, \eqref{FS:3} computes the finite rank operator determinant in \eqref{FS:2} as
\begin{equation*}
	\det\big(1-(1-T)^{-1}(U\otimes V)\upharpoonright_{L^2(t,\infty)}\big)=1-\int_t^{\infty}G(y)\big((1- T\upharpoonright_{L^2(t,\infty)})^{-1}g\big)(y)\,\d y=\Gamma_{t},
\end{equation*}
and \eqref{FS:1} follows from \eqref{FS:2} since $\det(1-T\upharpoonright_{L^2(t,\infty)})=\det(1- T\chi_t\upharpoonright_{L^2(\mathbb{R})})$. Secondly, it will be more convenient to move the $t$-dependency in the right-hand side of \eqref{FS:1} into the integral operators,
\begin{equation}\label{FS:4}
	\big(F(t)\big)^2=\det\big(1-(T_t+U_t\otimes V_t)\upharpoonright_{L^2(0,\infty)}\big),\ \ \ t\in\mathbb{R},
\end{equation}
where $T_t:L^2(0,\infty)\rightarrow L^2(0,\infty)$ has kernel $T_t(x,y)=T(x+t,y+t)$, compare \eqref{e:2}, $U_t$ is multiplcation by $g(x+t)$ and $V_t$ has kernel $G(y+t)$. Thirdly, we note that $T_t=S_tS_t$ where $S_t:L^2(0,\infty)\rightarrow L^2(0,\infty)$ has kernel
\begin{equation*}
	S_t(x,y)=\frac{1}{\sqrt{\pi}}\e^{-(x+y+t)^2},\ \ \ x,y>0.
\end{equation*}
\begin{lem}\label{reggood} For every $t\in\mathbb{R}$, the operator $S_t$ satisfies $\|S_t\|\leq 1$ and $1\mp  S_t$ are invertible on $L^2(0,\infty)$.
\end{lem}
\begin{proof} Since $S_t$ is self-adjoint we have for any $f\in L^2(0,\infty)$,
\begin{equation*}
	\|S_tf\|^2_{L^2(0,\infty)}=\langle S_tf,S_tf\rangle_{L^2(0,\infty)}=\langle f,S_t^2f\rangle_{L^2(0,\infty)}=\langle f, T_tf\rangle_{L^2(0,\infty)}\stackrel{\eqref{l:3}}{\leq}\langle f,f\rangle_{L^2(0,\infty)},
\end{equation*}
and therefore
\begin{equation*}
	\|S_t\|=\sup_{\|f\|_{L^2(0,\infty)}=1}\|S_tf\|_{L^2(0,\infty)}\leq 1.
\end{equation*}
Also, $(1-S_t)(1+S_t)=(1+S_t)(1-S_t)=1-T_t$ on $L^2(0,\infty)$ and since $1-T_t$ is invertible by Lemma \ref{reg}, so are $1\mp S_t$.
\end{proof}
\noindent The last Lemma allows us to transform the right-hand side in \eqref{FS:4} through the following factorization,
\begin{align}
	\det&\big(1-(S_t^2+U_t\otimes V_t)\upharpoonright_{L^2(0,\infty)}\big)=\det(1-S_t\upharpoonright_{L^2(0,\infty)})\det\big(1+S_t-(1-S_t)^{-1}(U_t\otimes V_t)\upharpoonright_{L^2(0,\infty)}\big)\nonumber\\
	=&\,\det(1-S_t\upharpoonright_{L^2(0,\infty)})\det\big(1+S_t-(U_t\otimes W)\upharpoonright_{L^2(0,\infty)}\big),\label{FS:5}
\end{align}
where $W$ multiplies by the characteristic function $\chi_0(x)$.\footnote{Evidently, $\chi_0$ and the later on used $\delta_0$ are not in $L^2(0,\infty)$. Still, using regularity and decay properties of the involved integral kernels all subsequent Fredholm determinant and inner product manipulations are justifiable, we refer the interested reader to \cite[Section VIII]{TW2}.} To get to \eqref{FS:5} we have used that for any operator $A$ we have $A(U\otimes V)=(U\otimes A^{\ast}V)$ in terms of the real adjoint $A^{\ast}$ and that
\begin{equation*}
	G(x+t)=1-\int_0^{\infty}S_t(x,y)\,\d y=\big((1-S_t\upharpoonright_{L^2(0,\infty)})\chi_0\big)(x),\ \ \ \ \ \chi_0(x)=\begin{cases}1,&x\geq 0\\ 0,&x<0\end{cases}.
\end{equation*}
Continuing with \eqref{FS:5}, another factorization yields
\begin{align}
	\det\big(1+S_t&-(U_t\otimes W)\upharpoonright_{L^2(0,\infty)}\big)=\det(1+S_t\upharpoonright_{L^2(0,\infty)})\det\big(1-(1+S_t)^{-1}(U_t\otimes W)\upharpoonright_{L^2(0,\infty)}\big)\nonumber\\
	&=\det(1+S_t\upharpoonright_{L^2(0,\infty)})\left\{1-\int_0^{\infty}\chi_0(y)\big((1+S_t\upharpoonright_{L^2(0,\infty)})^{-1}g\big)(y+t)\,\d y\right\}\label{FS:6}
\end{align}
where we have used a variation of the eigenvector/value trick \eqref{FS:3} in the last equality. Since
\begin{equation*}
	g(x+t)=\int_0^{\infty}S_t(x,y)\delta_0(y)\,\d y=\big(S_t\delta_0\big)(x);\ \ \ \ \ \ \int_0^{\infty}f(x)\delta_0(x)\,\d x:=f(0),
\end{equation*}
for any test function $f$ we can simplify the second factor in \eqref{FS:6} further,
\begin{align}
	1-\int_0^{\infty}\chi_0(y)\big((1+&S_t\upharpoonright_{L^2(0,\infty)})^{-1}g\big)(y+t)\,\d y=\int_0^{\infty}\chi_0(y)\Big(\big(1-(1+S_t)^{-1}S_t\upharpoonright_{L^2(0,\infty)}\big)\delta_0\Big)(y)\,\d y\nonumber\\
	&\,=\int_0^{\infty}\chi_0(y)\big((1+S_t\upharpoonright_{L^2(0,\infty)})^{-1}\delta_0\big)(y)\,\d y=\big\langle\chi_0,(1+S_t\upharpoonright_{L^2(0,\infty)})^{-1}\delta_0\big\rangle_{L^2(0,\infty)}\label{FS:7}
\end{align}
The proof of Theorem \ref{main3} would thus be completed through \eqref{FS:5}, \eqref{FS:6}, \eqref{FS:7} and the identity $\det(1-S_t\upharpoonright_{L^2(0,\infty)})=\det(1-S\chi_t\upharpoonright_{L^2(\mathbb{R})})$ if we manage to proof the following
\begin{equation}\label{FS:8}
	\det(1-S_t\upharpoonright_{L^2(0,\infty)})=\det(1+S_t\upharpoonright_{L^2(0,\infty)})\big\langle\chi_0,(1+S_t\upharpoonright_{L^2(0,\infty)})^{-1}\delta_0\big\rangle_{L^2(0,\infty)},\ \ t\in\mathbb{R};
\end{equation}
or equivalently (taking logarithmic derivatives, then observing the unity normalization of all three factors in \eqref{FS:8} as $t\rightarrow+\infty$, using Lemma \ref{reggood} and self-adjointness),	
\begin{equation}\label{FS:9}
	\tr_{L^2(0,\infty)}\left((1-S_t^2)^{-1}\frac{\d S_t}{\d t}\right)=-\frac{1}{2}\frac{\d}{\d t}\ln\big\langle\delta_0,(1+S_t\upharpoonright_{L^2(0,\infty)})^{-1}\chi_0\big\rangle_{L^2(0,\infty)}.
\end{equation}
But integrating by parts in the left-hand side of \eqref{FS:9} with $\frac{\d}{\d t}S_t(x,y)=\frac{\d}{\d x}S_t(x,y)=\frac{\d}{\d y}S_t(x,y)$ shows that (cf. \cite[Lemma $2$]{FS}),
\begin{equation*}
	\tr_{L^2(0,\infty)}\left((1-S_t^2)^{-1}\frac{\d S_t}{\d t}\right)=-\frac{1}{2}\big\langle\delta_0,((1-S_t^2)^{-1}S_t\upharpoonright_{L^2(0,\infty)})\delta_0\big\rangle_{L^2(0,\infty)},
\end{equation*}
so we need to establish the equality
\begin{equation}\label{FS:10}
	\big\langle\delta_0,((1-S_t^2)^{-1}S_t\upharpoonright_{L^2(0,\infty)})\delta_0\big\rangle_{L^2(0,\infty)}=\frac{\d}{\d t}\ln\big\langle\delta_0,(1+S_t\upharpoonright_{L^2(0,\infty)})^{-1}\chi_0\big\rangle_{L^2(0,\infty)}.
\end{equation}
\begin{lem}{\cite[Lemma $3$]{FS}} Let $\Delta_0$ denote multiplication by $\delta_0(x)$, and $D$ $(=\frac{\d}{\d x})$ differentiation. Then
\begin{equation}\label{FS:11}
	\frac{\d}{\d t}(1+S_t)^{-1}=(1-S_t^2)^{-1}S_tD+(1-S_t^2)^{-1}S_t\Delta_0(1+S_t)^{-1}.
\end{equation}
\end{lem}
\begin{proof} Integrating by parts, we easily find that $DS_t=-S_t\Delta_0-S_tD$, or equivalently
\begin{equation*}
	-(1-S_t)DS_t=S_tD(1+S_t)+S_t\Delta_0.
\end{equation*}
Multiplying this last identity by $(1-S_t^2)^{-1}$ from the left and $(1+S_t)^{-1}$ from the right (recall Lemma \ref{reggood}),
\begin{equation*}
	-(1+S_t)^{-1}DS_t(1+S_t)^{-1}=(1-S_t^2)^{-1}S_tD+(1-S_t^2)^{-1}S_t\Delta_0(1+S_t)^{-1},
\end{equation*}
and the last left-hand side equals precisely $\frac{\d}{\d t}(1+S_t)^{-1}$ because of the simple fact that 
\begin{equation*}
	\frac{\d}{\d t}(1+S_t)^{-1}=-(1+S_t)^{-1}\frac{\d S_t}{\d t}(1+S_t)^{-1},
\end{equation*}
and $\frac{\d}{\d t}S_t(x,y)=\frac{\d}{\d x}S_t(x,y)$. Identity \eqref{FS:11} is thus proven.
\end{proof}
\noindent We now use \eqref{FS:11} in the right-hand side of \eqref{FS:10} (suppressing at times the Hilbert space references for compact notation),
\begin{align*}
	\frac{\d}{\d t}\ln\big\langle\delta_0,(1+&S_t\upharpoonright_{L^2(0,\infty)})^{-1}\chi_0\big\rangle_{L^2(0,\infty)}=\frac{\langle\delta_0,\frac{\d}{\d t}(1+S_t\upharpoonright_{L^2(0,\infty)})^{-1}\chi_0\rangle_{L^2(0,\infty)}}{\langle\delta_0,(1+S_t\upharpoonright_{L^2(0,\infty)})^{-1}\chi_0\rangle_{L^2(0,\infty)}}\\
	\stackrel{D\chi_0=0}{=}&\,\,\frac{\langle\delta_0,(1-S_t^2)^{-1}S_t\Delta_0(1+S_t)^{-1}\chi_0\rangle}{\langle\delta_0,(1+S_t)^{-1}\chi_0\rangle}=\frac{\langle\delta_0,(1-S_t^2)^{-1}S_t\delta_0\rangle\langle\delta_0,(1+S_t)^{-1}\chi_0\rangle}{\langle\delta_0,(1+S_t)^{-1}\chi_0\rangle}\\
	=&\,\,\langle\delta_0,((1-S_t^2)^{-1}S_t\upharpoonright_{L^2(0,\infty)})\delta_0\rangle_{L^2(0,\infty)},
\end{align*}
which is \eqref{FS:10}. In turn the proof of Theorem \ref{main3} is thus completed.
\begin{appendix}
\section{Small norm estimates for collapsing contours}\label{appA}
The jump contour $\Sigma_{{\bf M}}$ of RHP \ref{split} collapses for large $(-t)$ to the real axis and we shall invoke ideas from \cite{BK} in the solution of the underlying singular integral equation (\cite{BK} does not apply verbatim to RHP \ref{split} as we are not dealing with contracting disk contours, for those the scaling invariance of $z^{-1}\d z$ is central).\smallskip

Write $\Sigma_{{\bf M}}=\bigcup_{j=1}^8\Sigma_j$ as union of the eight straight lines shown in Figure \ref{FigA} below and suppose $\rho\in L^2(\Sigma_{{\bf M}},|\d z|)$ is a $2\times 2$ matrix-valued function which is Lipschitz on $\Sigma_j$ and that satisfies 
\begin{equation}\label{app:1}
	\rho(z)=\mathbb{I}+\frac{1}{2\pi\im}\int_{\Sigma_{{\bf M}}}\rho(\lambda)\big({\bf G}_{{\bf M}}(\lambda)-\mathbb{I}\big)\frac{\d\lambda}{\lambda-z_-},\ \ \ z\in\Sigma_j,
\end{equation}
where $z_-$ denotes the limiting value of the integral from the right side. The explicit form of the jump matrix ${\bf G}_{\bf M}(z)={\bf G}_{\bf M}(z;t,\gamma)$ is stated in RHP \ref{split}, condition (2).
\begin{lem}\label{equal} On each $\Sigma_j$ the function
\begin{equation*}
	{\bf M}(z):=\mathbb{I}+\frac{1}{2\pi\im}\int_{\Sigma_{{\bf M}}}\rho(\lambda)\big({\bf G}_{{\bf M}}(\lambda)-\mathbb{I}\big)\frac{\d\lambda}{\lambda-z},\ \ \ z\in\mathbb{C}\setminus\Sigma_{{\bf M}}
\end{equation*}
satisfies 
\begin{equation*}
	{\bf M}_{+}(z)={\bf M}_{-}(z){\bf G}_{{\bf M}}(z),\ \ \ z\in\Sigma_j.
\end{equation*}
\end{lem}
\begin{proof} Since ${\bf M}_{-}(z)=\rho(z)$ for $z\in\Sigma_j$ the jump behavior of the Cauchy transform implies
\begin{equation*}
	{\bf M}_{+}(z)-{\bf M}_{-}(z)=\rho(z)\big({\bf G}_{{\bf M}}(z)-\mathbb{I}\big)={\bf M}_{-}(z)\big({\bf G}_{{\bf M}}(z)-\mathbb{I}\big),\ \ z\in\Sigma_j.
\end{equation*}
\end{proof}
\begin{figure}[tbh]
\begin{tikzpicture}[xscale=0.7,yscale=0.4]
\draw [->] (-5,0) -- (5,0) node[below]{{\small $\Re z$}};
\draw [->] (0,-4) -- (0,4) node[left]{{\small $\Im z$}};
\draw [thick, color=red, decoration={markings, mark=at position 0.2 with {\arrow{>}}}, decoration={markings, mark=at position 0.8 with {\arrow{>}}}, postaction={decorate}] (-4.5,0) -- (-2,0);
\node [above,color=red] at (-3.25,0.15) {{\small $\Sigma_1$}};
\node [above, color=red] at (-1.85,1) {{\small $\Sigma_2$}};
\node [below, color=red] at (-1.85,-1) {{\small $\Sigma_3$}};
\node [above, color=red] at (1.85,1) {{\small $\Sigma_6$}};
\node [below, color=red] at (1.85,-1) {{\small $\Sigma_7$}};
\node [above, color=red] at (-0.4,2.15) {{\small $\Sigma_4$}};
\node [below, color=red] at (-0.4,-2.15) {{\small $\Sigma_5$}};
\draw [thick, color=red, decoration={markings, mark=at position 0.5 with {\arrow{>}}}, postaction={decorate}] (-2,0) -- (-1,2);
\draw [thick, color=red, decoration={markings, mark=at position 0.5 with {\arrow{>}}}, postaction={decorate}] (-2,0) -- (-1,-2);
\draw [thick, color=red, decoration={markings, mark=at position 0.2 with {\arrow{>}}}, decoration={markings, mark=at position 0.8 with {\arrow{>}}},  postaction={decorate}] (-1,2) -- (1,2); 
\draw [thick, color=red, decoration={markings, mark=at position 0.2 with {\arrow{>}}}, decoration={markings, mark=at position 0.8 with {\arrow{>}}},  postaction={decorate}] (-1,-2) -- (1,-2); 
\draw [thick, color=red, decoration={markings, mark=at position 0.5 with {\arrow{>}}}, postaction={decorate}] (1,2) -- (2,0);
\draw [thick, color=red, decoration={markings, mark=at position 0.5 with {\arrow{>}}}, postaction={decorate}] (1,-2) -- (2,0);
\draw [thick, color=red, decoration={markings, mark=at position 0.2 with {\arrow{>}}}, decoration={markings, mark=at position 0.8 with {\arrow{>}}}, postaction={decorate}] (2,0) -- (4.5,0);
\node [above,color=red] at (3.25,0.15) {{\small $\Sigma_8$}};
\end{tikzpicture}
\caption{The oriented jump contour $\Sigma_{{\bf M}}=\bigcup_{j=1}^8\Sigma_j$ in RHP \ref{split}.}
\label{FigA}
\end{figure}
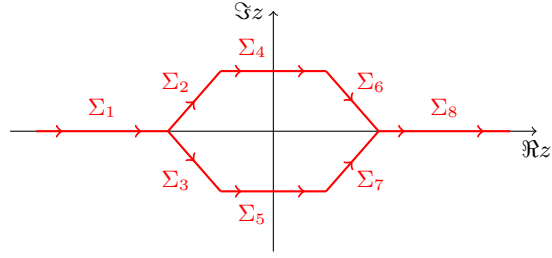
We shall now solve \eqref{app:1} in $L^2(\Sigma_{{\bf M}},|\d z|)$ by the Neumann series
\begin{eqnarray*}
	\rho(z)=\mathbb{I}+\sum_{k=1}^{\infty}\rho_k(z);\ \ \ \ \ \ \ \ \rho_k(z)&=&\frac{1}{2\pi\im}\int_{\Sigma_{{\bf M}}}\rho_{k-1}(\lambda)\big({\bf G}_{{\bf M}}(\lambda)-\mathbb{I}\big)\frac{\d\lambda}{\lambda-z_-},\ \ z\in\Sigma_{{\bf M}},\ k\in\mathbb{Z}_{\geq 1}\\
	\rho_0(z)&=&\mathbb{I},
\end{eqnarray*}
and thus need to estimate $\rho_k(z)$. Recall that $L^2(\Sigma_{{\bf M}},|\d z|)$ is the space of (matrix-valued) measurable functions such that
\begin{equation*}
	\|f\|_{L^2(\Sigma_{{\bf M}},|\d z|)}=\left\{\int_{\Sigma_{{\bf M}}}\|f(z)\|^2|\d z|\right\}^{\frac{1}{2}}<\infty.
\end{equation*}
Let $C_{\Sigma_j}^{\pm}$ denote the Cauchy operators on $L^2(\Sigma_{{\bf M}},|\d z|)$,
\begin{equation*}
	(C_{\Sigma_j}^{\pm}\rho)(z)=\lim_{\epsilon\downarrow 0}\int_{\Sigma_j}\frac{\rho(s)}{s-(z\pm\im\epsilon)}\frac{\d s}{2\pi\im},\ \ z\in\Sigma_j,
\end{equation*}
which obey (cf. \cite[Chapter II]{S} or \cite[Section $5.5$]{BDS})
\begin{equation}\label{PS}
	C_{\Sigma_j}^+-C_{\Sigma_j}^-=\textnormal{id},\ \ \ \ \ C_{\Sigma_j}^++C_{\Sigma_j}^-=\im H_{\Sigma_j},\ \ \textnormal{a.e. on}\ \Sigma_j;\ \ \ \ (H_{\Sigma_j}\rho)(z)=\frac{1}{\pi}\,\textnormal{pv}\int_{\Sigma_j}\rho(\lambda)\frac{\d s}{z-s}.
\end{equation}
\begin{prop}[\cite{CMM}, Theoreme I]\label{Coif} If an oriented contour $\Gamma\subset\mathbb{C}$, given by the parametric equations
\begin{equation*}
	\Gamma=\{z\in\mathbb{C}:\ \Re z=t\in\mathbb{R},\ \Im z=\phi(t)\}
\end{equation*}
satisfies a uniform Lipschitz condition, i.e. there exists $M\geq 0$ such that
\begin{equation*}
	\big|\phi(x)-\phi(y)\big|\leq M|x-y|,
\end{equation*}
then there exists a universal $C_0>0$ such that
\begin{equation}\label{app:2}
	\|H_{\Gamma}f\|_{L^2(\Gamma,|\d z|)}\leq C_0(1+M)^{10}\|f\|_{L^2(\Gamma,|\d z|)}.
\end{equation}
\end{prop}
Observe that our eight pieces $\Sigma_j$ fit into the context of Proposition \ref{Coif} with a $(t,\gamma)$-independent constant $M$, thus we are now prepared to the derive our central estimate.
\begin{theo} For any $\gamma\in(0,1)$ and $t<0$, let $\Sigma_{{\bf M}}=\bigcup_{j=1}^8\Sigma_j$ denote the $(t,\gamma)$-dependent contour of RHP \ref{split}. Then there exists a universal constant $K>0$ such that
\begin{equation}\label{app:3}
	\|H_{\Sigma_{{\bf M}}}f\|_{L^2(\Sigma_{{\bf M}},|\d z|)}\leq K\ln(\delta_{t\gamma}^{-1})\,\|f\|_{L^2(\Sigma_{{\bf M}},|\d z|)}.
\end{equation}
\end{theo}
\begin{proof} We show that for some $\widehat{K}>0$,
\begin{equation*}
	\big|\langle H_{\Sigma_{{\bf M}}}(\chi_{\Sigma_j}f),\chi_{\Sigma_k}g\rangle_{L^2(\Sigma_{{\bf M}},|\d z|)}\big|\leq\widehat{K}\ln(\delta_{t\gamma}^{-1})\,\|f\|_{L^2(\Sigma_{{\bf M}},|\d z|)}\|g\|_{L^2(\Sigma_{{\bf M}},|\d z|)},\ \ \ \ 1\leq j,k\leq 8.
\end{equation*}
Indeed, for $j=k$ this follows at once from \eqref{app:2} and for $j\neq k$ we use the following estimates (derived from the Cauchy-Schwarz inequality while using polar coordinates and standard manipulations)
\begin{align*}
	\int_{-1}^{1}\int_{-1}^{1}\frac{|f(x+\im\delta_{\gamma t})g(y-\im\delta_{t\gamma})|}{\sqrt{(x-y)^2+4\delta_{t\gamma}^2}}\,\d x\,\d y\,\leq&\,\,C_1\ln(\delta_{t\gamma}^{-1})\,\|f\|_{L^2(\Sigma_4,|\d z|)}\|g\|_{L^2(\Sigma_5,|\d z|)};\\
	\int_0^{\infty}\int_0^{\infty}\frac{|f(x)g(y)|}{\sqrt{(x-y)^2+\delta_{t\gamma}^2x^2}}\,\d x\,\d y\,\leq&\,\,C_2\ln(\delta_{t\gamma}^{-1})\|f\|_{L^2(0,\infty)}\|g\|_{L^2(0,\infty)},\ \ \ \ C_j>0;\\
	\int_0^{\infty}\int_0^{\infty}\frac{|f(x)g(y)|}{x+y}\,\d x\,\d y\,\leq&\,\,C_3\|f\|_{L^2(0,\infty)}\|g\|_{L^2(0,\infty)}.
\end{align*}
\end{proof}
From estimate \eqref{app:3} we derive the operator norm estimate
\begin{equation*}
	\|H_{\Sigma_{{\bf M}}}\|_{L^2(\Sigma_{{\bf M}},|\d z|)}\leq K\ln(\delta_{t\gamma}^{-1}),\ \ \ \ K>0\ \ \textnormal{universal}
\end{equation*}
and then in turn, with \eqref{PS},
\begin{equation}\label{epp:4}
	\|C_{\Sigma_{{\bf M}}}^{\pm}\|_{L^2(\Sigma_{{\bf M}},|\d z|)}\leq \frac{1}{2}(K+1)\ln(\delta_{t\gamma}^{-1}),\ \ \forall\,(-t)\geq t_0,\ \ \ \gamma\in(0,1).
\end{equation}
Returning now to the iterates $\{\rho_k(z)\}_{k=0}^{\infty}$ introduced above, we have 
\begin{equation*}
	\rho_k(z)=\Big(C^-_{\Sigma_{{\bf M}}}\big[\rho_{k-1}({\bf G}_{{\bf M}}-\mathbb{I})\big]\Big)(z),\ \ \ z\in\Sigma_{{\bf M}}.
\end{equation*}
But
\begin{equation*}
	\|\rho_{k-1}({\bf G}_{{\bf M}}-\mathbb{I})\|_{L^2(\Sigma_{{\bf M}},|\d z|)}\leq\|{\bf G}_{{\bf M}}(\cdot;t,\gamma)-\mathbb{I}\|_{L^{\infty}(\Sigma_{{\bf M}})}\cdot\|\rho_{k-1}\|_{L^2(\Sigma_{{\bf M}},|\d z|)},
\end{equation*}
and for all $(-t)\geq t_0,\gamma\in(0,1)$,
\begin{equation*}
	\|\rho_1\|_{L^2(\Sigma_{{\bf M}},|\d z|)}=\|C^-_{\Sigma_{{\bf M}}}({\bf G}_{{\bf M}}-\mathbb{I})\|_{L^2(\Sigma_{{\bf M}},|\d z|)}\leq\frac{1}{2}(K+1)\ln(\delta_{t\gamma}^{-1})\cdot\|{\bf G}_{{\bf M}}(\cdot;t,\gamma)-\mathbb{I}\|_{L^2(\Sigma_{{\bf M}},|\d z|)},
\end{equation*}
so that with Proposition \ref{smallnorm},
\begin{prop}\label{key} There exist positive universal constants $t_0,c$ such that for any $\gamma\in(0,1)$,
\begin{equation*}
	\|\rho_k\|_{L^2(\Sigma_{{\bf M}_{\gamma}},|\d z|)}\leq \left(c\ln(\delta_{t\gamma}^{-1})\frac{\e^{-t^2\delta_{t\gamma}}}{1-\gamma}\right)^k|t|^{-\frac{1}{2}},\ \ \ \forall\,(-t)\geq t_0,\ \ k\in\mathbb{Z}_{\geq 1}.
\end{equation*}
\end{prop}
Thus, given any $0<\epsilon<2$, Proposition \ref{key} implies convergence of the Neumann series $\mathbb{I}+\sum_{k=1}^{\infty}\rho_k(z)$ in $L^2(\Sigma_{{\bf M}},|\d z|)$ for sufficiently large $(-t)\geq t_0$ and any $\gamma\in(0,1)$ such that $0<\gamma\leq 1-|t|^{-\epsilon}$. 
At this point, using $\rho(z)$, we define
\begin{equation}\label{app:5}
	{\bf M}(z):=\mathbb{I}+\frac{1}{2\pi\im}\int_{\Sigma_{{\bf M}}}\rho(\lambda)\big({\bf G}_{{\bf M}}(\lambda)-\mathbb{I}\big)\frac{\d\lambda}{\lambda-z},\ \ \ z\in\mathbb{C}\setminus\Sigma_{{\bf M}}
\end{equation}
which coincides with the function ${\bf M}(z)$ defined by \eqref{l:65} (compare Lemma \ref{equal} and the argument in \cite[(A.37)-(A.39)]{BK} near a triple point, a point on $\Sigma_{{\bf M}}$ where three arcs meet). Thus, compare RHP \ref{split},
\begin{equation*}
	{\bf S}_{1}(t,\gamma)=\lim_{z\rightarrow\infty}z\big({\bf M}(z)-\mathbb{I}\big)=\frac{\im}{2\pi}\int_{\Sigma_{{\bf M}}}\rho(\lambda)\big({\bf G}_{{\bf M}}(\lambda)-\mathbb{I}\big)\,\d\lambda=\sum_{k=0}^{\infty}{\bf S}_{1k}(t,\gamma),
\end{equation*}
where
\begin{equation*}
	{\bf S}_{1k}(t,\gamma)=\frac{\im}{2\pi}\int_{\Sigma_{{\bf M}}}\rho_k(\lambda)\big({\bf G}_{{\bf M}}(\lambda)-\mathbb{I}\big)\,\d\lambda.
\end{equation*}
Since for $(-t)\geq t_0$ and $\gamma\in(0,1)$, with $C_j>0$,
\begin{align*}
	k\geq 1:&\,\,\,\,\,\,\|{\bf S}_{1k}(t,\gamma)\|\leq C_1\|\rho_k\|_{L^2(\Sigma_{{\bf M}},|\d z|)}\|{\bf G}_{{\bf M}}(\cdot;t,\gamma)-\mathbb{I}\|_{L^2(\Sigma_{{\bf M}},|\d z|)}\leq C_1\left(c\ln(\delta_{t\gamma}^{-1})\frac{\e^{-t^2\delta_{t\gamma}}}{1-\gamma}\right)^k\frac{\e^{-t^2\delta_{t\gamma}}}{1-\gamma}|t|^{-1},\\
	k=0:&\,\,\,\,\,\,\|{\bf S}_{1k}(t,\gamma)\|\leq C_2\|{\bf G}_{{\bf M}}(\cdot;t,\gamma)-\mathbb{I}\|_{L^1(\Sigma_{{\bf M}},|\d z|)}\leq C_2\frac{\e^{-t^2\delta_{t\gamma}}}{1-\gamma}|t|^{-1},
\end{align*}
we can now sum all inequalities from $k=0$ to $k=+\infty$: for any fixed $\epsilon\in(0,2)$,
\begin{equation*}
	{\bf S}_{1}(t,\gamma)=\mathcal{O}\left(|t|^{-1+\epsilon}\e^{-|t|^{1-\frac{\epsilon}{2}}}\right),\ \ \ \forall\ (-t)\geq t_0,\ \ 0<\gamma\leq 1-|t|^{-\epsilon}.
\end{equation*}
This completes the proof of Theorem \ref{collapse}.
\section{Permuting resolvent and integration}\label{appC}
Let $\phi$ and $\psi$ be two functions on $\mathbb{R}$ that decay exponentially fast at $+\infty$. Introduce
\begin{equation*}
	K(x,y):=\int_0^{\infty}\phi(x+s)\psi(y+s)\,\d s,\ \ \ \ x,y\in\mathbb{R},
\end{equation*}
and the corresponding integral operator $K$ on $L^2(\mathbb{R})$ with kernel $K(x,y)$. For any function $f$, we shall denote by $f_y(x):=f(x+y)$ the horizontal translation of $f$ by $-y$. Then,
\begin{prop}\label{Jinho:1} For any $x,y,t\in\mathbb{R}$ and $k\in\mathbb{Z}_{\geq 1}$, we have
\begin{equation}\label{J:1}
	\big((K\chi_t\upharpoonright_{L^2(\mathbb{R})})^k\phi_y\big)(x)=\big((K\chi_t\upharpoonright_{L^2(\mathbb{R})})^k\phi_{x-t}\big)_y(t).
\end{equation}
\end{prop}	
\begin{proof} We prove \eqref{J:1} inductively using 
\begin{equation*}
	K(x,y)=\int_t^{\infty}\phi(x+u-t)\psi(y+u-t)\,\d u.
\end{equation*}
Indeed, for $k=1$, we have
\begin{equation*}
	\big(K\chi_t\upharpoonright_{L^2(\mathbb{R})}\phi_y\big)(x)=\int_t^{\infty}K(x,s)\phi(s+y)\,\d s=\int_s^{\infty}\int_s^{\infty}\phi(x+u-t)\psi(s+u-t)\phi(s+y)\,\d u\,\d s,
\end{equation*}
so that by Fubini's theorem,
\begin{equation*}
	\big(K\chi_t\upharpoonright_{L^2(\mathbb{R})}\phi_y\big)(x)=\int_t^{\infty}K(y+t,u)\phi(u+x-t)\,\d u=\big(K\chi_t\upharpoonright_{L^2(\mathbb{R})}\phi_{x-t}\big)_y(t).
\end{equation*}
For general $k$, assuming that \eqref{J:1} is true for $k-1$,
\begin{align*}
	&\hspace{1cm}\big((K\chi_t\upharpoonright_{L^2(\mathbb{R})})^k\phi_y\big)(x)=\int_t^{\infty}K(x,s)\big((K\chi_t\upharpoonright_{L^2(\mathbb{R})})^{k-1}\phi_y\big)(s)\,\d s\\
	=&\,\int_t^{\infty}\!K(x,s)\big((K\chi_t\upharpoonright_{L^2(\mathbb{R})})^{k-1}\phi_{s-t}\big)_y(t)\,\d s=\int_t^{\infty}\!\!\int_t^{\infty}\!\!K(x,s)\big(K\chi_t\upharpoonright_{L^2(\mathbb{R})}\big)^{k-1}(t+y,u)\phi(u+s-t)\,\d u\,\d s,
\end{align*}
and thus again by Fubini's theorem
\begin{equation*}
	\big((K\chi_t\upharpoonright_{L^2(\mathbb{R})})^k\phi_y\big)(x)=\int_t^{\infty}\big(K\chi_t\upharpoonright_{L^2(\mathbb{R})}\big)^{k-1}(t+y,u)\big(K\chi_t\upharpoonright_{L^2(\mathbb{R})}\phi_{u-t}\big)(x)\,\d u.
\end{equation*}
Now apply the base case $k=1$ result and derive
\begin{align*}
	\big((K\chi_t\upharpoonright_{L^2(\mathbb{R})})^k\phi_y\big)(x)=&\,\int_t^{\infty}\big(K\chi_t\upharpoonright_{L^2(\mathbb{R})}\big)^{k-1}(t+y,u)\big(K\chi_t\upharpoonright_{L^2(\mathbb{R})}\phi_{x-t}\big)(u)\,\d u\\
	=&\,\big((K\chi_t\upharpoonright_{L^2(\mathbb{R})})^k\phi_{x-t}\big)(t+y)=\big((K\chi_t\upharpoonright_{L^2(\mathbb{R})})^k\phi_{x-t}\big)_y(t),
\end{align*}
which completes the proof.
\end{proof}
The above Proposition \ref{Jinho:1} leads to the following useful Corollary
\begin{cor}\label{Jinho:2} Let
\begin{equation*}
	\Phi(x):=\int_I\phi_y(x)\,\d y,
\end{equation*} 
where $I$ is a subset of $\mathbb{R}$. Then, for any $x,t\in\mathbb{R}$ and $k\in\mathbb{Z}_{\geq 1}$,
\begin{equation}\label{J:2}
	\big((K\chi_t\upharpoonright_{L^2(\mathbb{R})})^k\Phi\big)(x)=\int_I\big((K\chi_t\upharpoonright_{L^2(\mathbb{R})})^k\phi_{x-t}\big)_y(t)\,\d y.
\end{equation}
\end{cor}
\begin{proof} We have by linearity and \eqref{J:1},
\begin{align*}
	\big((K\chi_t\upharpoonright_{L^2(\mathbb{R})})^k\Phi\big)(x)=&\,(K\chi_t\upharpoonright_{L^2(\mathbb{R})})^k\left[\int_I\phi_y(\cdot)\,\d y\right](x)=\int_I\big((K\chi_t\upharpoonright_{L^2(\mathbb{R})})^k\phi_y\big)(x)\,\d y\\
	&\,\stackrel{\eqref{J:1}}{=}\int_I\big((K\chi_t\upharpoonright_{L^2(\mathbb{R})})^k\phi_{x-t}\big)_y(t)\,\d y.
\end{align*}
\end{proof}	

\end{appendix}

\end{document}